\documentclass[letterpaper,11pt,oneside]{article}
\usepackage{graphicx}
\usepackage{amssymb}
\usepackage{amsthm}
\usepackage{amsmath}
\usepackage{tikz}
\usetikzlibrary{shapes}
\usetikzlibrary{patterns}
\usetikzlibrary{arrows,positioning,decorations.pathmorphing,trees}
\usetikzlibrary{arrows.meta}
\usepackage{cancel}
\usepackage{tikz-cd}
\makeatletter
\tikzset{
  column sep/.code=\def\pgfmatrixcolumnsep{\pgf@matrix@xscale*(#1)},
  row sep/.code   =\def\pgfmatrixrowsep{\pgf@matrix@yscale*(#1)},
  matrix xscale/.code=
    \pgfmathsetmacro\pgf@matrix@xscale{\pgf@matrix@xscale*(#1)},
  matrix yscale/.code=
    \pgfmathsetmacro\pgf@matrix@yscale{\pgf@matrix@yscale*(#1)},
  matrix scale/.style={/tikz/matrix xscale={#1},/tikz/matrix yscale={#1}}}
\def\pgf@matrix@xscale{1}
\def\pgf@matrix@yscale{1}
\makeatother
\usepackage{stackengine}
\usepackage{pgfplots}
\usepackage{float}
\usepackage{braket}
\usepackage{setspace}
\usepackage[ruled, lined, linesnumbered, commentsnumbered, longend]{algorithm2e}
\usepackage{fullpage}
\usepackage{comment}
\usepackage[square,numbers]{natbib}
\usepackage{authblk}
\usepackage{hyperref}
\hypersetup{
    colorlinks=true,
    linkcolor=purple,
    citecolor=purple
    }

\newtheorem{theorem}{Theorem}[section]

\newtheorem{corollary}[theorem]{Corollary}

\newtheorem{lemma}[theorem]{Lemma}

\newtheorem{proposition}[theorem]{Proposition}
\newtheorem{definition}{Definition}

\newtheorem{example}{Example}
\newtheorem*{remark*}{Remark}
\newtheorem*{inftheorem*}{Theorem (informal)}
\newtheorem*{notation*}{Notation}
\newtheorem*{observation*}{Observation}
\newtheorem*{theorem*}{Theorem}
\newtheorem*{proposition*}{Proposition}
\newtheorem*{definition*}{Definition}
\newtheorem*{axiom*}{Axiom}
\newtheorem*{claim*}{Claim}
\newtheorem*{lemma*}{Lemma}

\title{Semicoarse Correlated Equilibria and LP-Based Guarantees for Gradient Dynamics in Normal-Form Games}
\author[1]{Mete \c{S}eref Ahunbay\thanks{E-mail: \texttt{mete-seref.ahunbay@univ-grenoble-alpes.fr}}}
\author[1]{Martin Bichler\thanks{E-mail: \texttt{m.bichler@tum.de}}}
\affil[1]{Technische Universit\"{a}t München $\rightarrow$ Université Grenoble Alpes / CNRS / INRIA / LIG \newline
Bâtiment IMAG, Université Grenoble Alpes, 150 Place du Torrent, 38401, St. Martin d'Hères, France.}
\affil[2]{Technische Universit\"{a}t München, School of Computation, Information and Technology, Department of Computer Science. Boltzmannstraße 3, 85748, Garching bei München, Germany.}
\date{4th June, 2025}

\begin{document}
\maketitle

\begin{abstract}
    Projected gradient ascent is known to satisfy no-external regret as a learning algorithm. However, recent empirical work shows that projected gradient ascent often finds the Nash equilibrium in settings beyond two-player zero-sum interactions or potential games, including those where the set of coarse correlated equilibria is very large. We show that gradient ascent in fact satisfies a stronger class of linear $\Phi$-regret in normal-form games; resulting in a refined solution concept which we dub semicoarse correlated equilibria. Our theoretical analysis of the discretised Bertrand competition mirrors those recently established for mean-based learning in first-price auctions. With at least two firms of lowest marginal cost, Nash equilibria emerge as the only semicoarse equilibria under concavity conditions on firm profits. In first-price auctions, the granularity of the bid space affects semicoarse equilibria, but finer granularity for lower bids also induces convergence to Nash equilibria. Unlike previous work that aims to prove convergence to a Nash equilibrium that often relies on epoch based analysis and probability theoretic machinery, our LP-based duality approach enables a simple and tractable analysis of equilibrium selection under gradient-based learning.
\end{abstract}

\allowdisplaybreaks

\section{Introduction}

The central problem we address in this paper is the development of a refined framework to analyse gradient dynamics and online (projected) gradient ascent \cite{zinkevich2003online} in normal-form games. With the recent prominence of machine learning, it has become a central issue to understand such outcomes, as most models are trained via gradient based methods. Moreover, common learning algorithms such as projected gradient ascent and Hedge \cite{freund1997decision} are based on ascent along a (potentially rescaled) utility gradient, whereas their restricted feedback counterparts (e.g. EXP3 \cite{auer2002nonstochastic}) use an unbiased estimator of the gradient instead. 

These algorithms are of course no-external regret, however, this property is in general too weak to conclude convergence to a Nash equilibrium (NE) even in the setting of the first-price auctions \cite{FLN16,ahunbay2025uniqueness}. And even though the correlated equilibrium (CE) is essentially unique in this setting, Hedge and projected gradient ascent do not satisfy no-internal regret. On the other hand, empirical work shows that the gradient-based algorithms reliably find a(n approximate) NE in a wide variety of auction settings \cite{BFHKS21,soda2023,banchio2023artificial}. Theoretical justification of such convergence behaviour to date is lacking. The literature on Nash convergence of learning algorithms depends overwhelmingly on assumptions of monotonicity of utility gradients or the existence of a potential \cite{wang2023noregret}, with neither assumption satisfied for the complete information first-price auction. Recent positive results \cite{kolumbus2022auctions,deng2022nash} focus on the subclass of mean-based learning algorithms, which include neither projected gradient ascent nor Hedge with adaptive step sizes. In this paper, we thus adopt a bottom-up perspective, seeking an answer to the question:

\emph{``Restricting attention to the setting where all players learn via projected gradient ascent, with full feedback on their utility gradients, can we provide stronger convergence guarantees?''}

Our answer in the affirmative is \emph{semicoarse (correlated) equilibrium}, an equilibrium refinement over coarse correlated equilibrium (CCE). Its definition leverages recent advances in non-convex game theory \cite{cai2024tractable,ahunbay2024local}, which shows that trajectories of projected gradient ascent satisfy further regret guarantees than those implied by the no-external regret property. The set of semicoarse equilibria is described by linear inequalities over the correlated distributions of outcomes, and is in general a strict superset of CE. As a consequence, convergence to equilibrium can be established via straightforward arguments based on strong duality of linear programming, without explicitly accounting for projections of utility gradients. Our analysis of semicoarse equilibrium of Bertrand competitions yields guarantees for projected gradient ascent parallelling the state-of-the-art for mean-based learners \cite{deng2022nash}.

\subsection{Contributions}

Our first contribution is the delineation of linear correlated equilibrium constraints which are necessarily satisfied by projected gradient ascent. Key to our work is \cite{ahunbay2024local}, which proves that the underlying gradient dynamics of a game ensures that trajectories of gradient ascent satisfy no ``regret'' with respect to strategy modifications generated by gradient fields of functions, in first-order. Coarse correlated equilibria of normal-form games are generated by linear vector fields in this sense, whereas correlated equilibria are generated via no regret against first-order strategy modifications generated by non-conservative vector fields. 

In both cases, the requirement that the strategy modifications are linear endomorphisms of the simplex translates to the requirement that the vector fields satisfy a tangency condition. Intuitively, the vector field must point inwards towards the probability simplex. Moreover, the (affine-)linearity of the associated gradient field restricts our attention to quadratic functions. The set of suitable functions then forms a polyhedral cone, whose rays we are completely able to classify. This provides us with our desired subset of strategy modifications.

\begin{theorem*}[Informal, Definition \ref{def:semicoarse}, Theorem \ref{thm:strategy-def}]
    In a normal-form game, if some player $i$ employs projected gradient ascent with suitably decreasing step sizes, they will incur vanishing regret against affine-symmetric strategy modifications $x_i \mapsto (\mathbf{1}+Q_i)x_i +q_i$, where $Q$ is a symmetric matrix. No regret against these strategy modifications is equivalent to no regret against the classes of linear endomorphisms of $\Delta(A_i)$:\begin{enumerate}
        \item For a proper subset $S_i \subseteq A_i$, any action in $S_i$ is mapped to the uniform distribution on its complement, whereas actions not in $S_i$ are left unmodified.
        \item For a cycle of actions $C = (a_1, a_2, ..., a_k)$ with $a_{k+1}=a_1$, actions on the cycle mix to their two neighbours each with probability $1/2$. Meanwhile, $a_i \notin C$ are left unmodified.
    \end{enumerate}
\end{theorem*}

To our knowledge, such explicit incentive guarantees for projected gradient ascent is unprecedented in literature. Moreover, for normal-form games larger than $3 \times 3$, we demonstrate through arguments similar to \cite{viossat2015evolutionary} that our linear equilibrium constraints are the largest possible for the gradient dynamics of a normal-form game (Proposition \ref{prop:linmax}). Whereas there are superexponentially many semicoarse equilibrium constraints, we are able to provide a polynomial size primal extension (Theorem \ref{thm:short-extension}) and a dual LP (\ref{opt:lyapunov-LP}) which is better behaved for solvers in practice.

To assess the strength of our equilibrium refinement, we consider the setting of discretised Bertrand competition with complete information. We assume $N$ firms, where each firm $i$ has marginal cost $c_i/n$, and may post a price $p_i$, both in $\{0,1/n,...,1\}$. The firm which posts the lowest price captures the entire demand $D(p_i)$, divided equally between such firms in case of ties. Here, the complete information first-price auction is equivalent to the case of inelastic demand. In this case, we find that convergence to the expected equilibrium outcome follows from reasonable concavity assumptions on firm profits, and the number of firms with lowest marginal cost.

\begin{theorem*}[Informal, Theorems \ref{thm:bertrand-unique}, \ref{thm:bertrand-unique-three}]
    If there are $\geq 2$ ($\geq 3$) firms with lowest marginal cost $c/n$, if the profit function $(p_i-c/n)D(p_i)$ is strictly (weakly) concave, then any outcome assigned positive probability in a semicoarse equilibrium is a pure strategy Nash equilibrium; at least two firms of lowest marginal cost post prices in $c/n + \{0,1/n\}$, where if the minimum price is $(c+1)/n$, all firms of lowest marginal cost do so.
\end{theorem*}

We compare this result with that of \cite{deng2022nash} for mean-based learning in first-price auctions. There, with $\geq 3$ buyers of maximum valuation, last-iterate convergence to Nash equilibrium is guaranteed; thus we recover the same last-iterate convergence guarantee. However, in the case of two such buyers, \cite{deng2022nash} prove that time-average, but not last-iterate convergence is guaranteed instead. So we study this threshold setting; here, we find that the notion of semicoarse equilibrium is sensitive to \emph{how} the density of the discretised bids is distributed. This is distinct from the notions of (C)CE, which are necessarily robust with respect to small changes in the action sets, and invariant under duplicating strategies. The deeper reason for this sensitivity is that such transformations distort the gradient dynamics of the game \cite{heemels2020oblique}, and semicoarse equilibrium is a correlated notion of equilibrium intrinsically tied to the gradient dynamics of the game. Nevertheless, we are able to prove (analogously for the first-price auction via an equivalence):

\begin{proposition*}[Informal, Proposition \ref{prop:weighted-bertrand}]
    If there are two firms with lowest marginal cost $c/n$, if the profit function $(p_i-c/n)D(p_i)$ is weakly concave, there are arbitrarily small modifications of the game's gradient dynamics via invertible linear transformations of players' individual sets of mixed strategies, such that projected gradient dynamics necessarily converges to a Nash equilibrium.
\end{proposition*}

Our proof uses the same arguments for the derivation of semicoarse equilibria, and when the linear transformations correspond to a weighing of players' strategies, we are able to give an explicit classification of the finitely many strategy deviations against which players incur no regret. Our result depends on a refinement of the no-regret result of \cite{ahunbay2024local} that applies for arbitrary polyhedral action sets in a smooth game, which might be of independent interest.

\subsection{Related Work}\label{sec:related}
Our paper touches on several themes in learning in games, all constituting lines of active research -- last-iterate and time-average convergence, stronger incentive guarantees, and notions of $\Phi$-equilibrium --, as well as themes in auction theory relating to the theoretical study of learnability of equilibrium. In this section, we provide a brief overview of their relevance to our paper, within the context of normal-form games.

Game theory provides the formal framework of analysis of the strategic interaction of self interested agents. This analysis often requires necessitates the notion of an ``expected outcome'' of a game, which is often taken to be a Nash equilibrium \cite{Nash50}. However, computing an exact or approximate Nash equilibrium is $PPAD$-hard even in two player bimatrix games \cite{CD06,DGP09}, and finding optimal Nash equilibria is $NP$-hard in general \cite{CS08}. From the perspective of algorithmic game theory, which is concerned with computational complexity and performance guarantees through the lens of approximation algorithms, this renders the assumption that agents reach a Nash equilibrium implausible \cite{nisan2007algorithmic}.

This brings up the question of what the correct notion of an outcome is, for which some answers come from learning theory. The overall idea is that we may work under the assumption that agents, through their interaction, implement learning algorithms to optimise their payoffs. Assumptions on these learning algorithms then provides us with incentive guarantees in hindsight. However, even in the setting of normal-form games, our current understanding of learning behaviour is limited. Learning algorithms are known to exhibit cyclic behaviour \cite{MPP18}, which can be formally chaotic \cite{CP19}. Recent results shed light on a deeper reason why; the multiplicative weights update over normal-form games can approximate arbitrary dynamical systems \cite{AFP21}, and it is Turing complete to determine the whether replicator dynamics in normal-form games reaches a given open set \cite{andrade2023turing}. 

The ubiquitousness of learning cycles has motivated both the study of learning cycles themselves as the expected outcome of a game \cite{PP19}, and the study of sufficient conditions which allow for algorithms with provable convergence to equilibrium. On the latter, positive results for convergence in the literature overwhelmingly depend on assumptions of monotonicity of utility gradients or the existence of a potential function \cite{MZ19}; we refer the reader to \cite{wang2023noregret} for an extended discussion. One divergent line of work has focused on ``near''-potential games \cite{candogan2013dynamics}; this assures convergence to an approximate equilibrium.

These assumptions are, of course, violated in many game theoretic settings of practical relevance. Auction theory is particularly rich in such instances, with even simple settings such as the first-price auction violate monotonicity. For the complete information case, if the utility gradients for the mixed extension were monotone \cite{MZ19}, then the unique CCE of the auction would have been a pure strategy Nash equilibrium \cite{ahunbay2025uniqueness}, which is not the case \cite{FLN16}. In the Bayesian setting, the non-monotonicity of the utility gradients are stated explicitly in \cite{bichler2023convergence} for the model where buyers modify their pure bidding strategies, and is observed empirically \cite{ahunbay2025uniqueness} for the mixed extension of the agent-normal form game \cite{forges1993five,HST15}\footnote{Except when the prior distribution is concave -- a very restrictive assumption in practice.}. This is in apparent contradiction with empirical work \cite{BFHKS21,soda2023,banchio2023artificial}, which demonstrate that algorithms based on online gradient ascent or other gradient-based methods converge to Nash equilibrium in a broad class of auction games. 

This suggests a gap in our understanding with regards to convergence behaviour, and underlines the necessity of moving beyond monotonicity. Thus, establishing convergence behaviour in auctions has become an important line of research. Several proofs of conditional convergence to Nash equilibrium exist for mean-based learners \cite{braverman2018selling} in first-price auctions with complete  \cite{kolumbus2022auctions,deng2022nash} or incomplete \cite{FGLMS21} information. The family of mean-based algorithms is rich, rendering the result of \cite{deng2022nash} very general\footnote{\cite{deng2022nash} remark, however, that Hedge with decreasing step sizes $\propto 1/\sqrt{t}$ fails to be mean-based.}. However, the result of \cite{FGLMS21} depends on a lengthy pretraining period, and is limited to the case of the uniform prior. And besides, online projected gradient ascent is demonstrably not a mean-based learning algorithm (cf. Appendix \ref{sec:not-mean-based}), which precludes its application in our setting.

The guarantee commonly known for online gradient ascent is no-external regret (follows from \cite{zinkevich2003online}), which implies that the empirical sequence of play converges to a coarse correlated equilibrium (also known as Hannan consistency, \cite{hannan1957approximation}). In the case of the first-price auction, however, the coarse correlated equilibria of the game are far from its Nash equilibrium, despite the fact that all of its correlated equilibria \cite{aumann1987correlated} are necessarily convex combinations of Nash equilibria \cite{FLN16}. Thus, linear programming duality based methods, such as those of \cite{lopomo2011lp,ahunbay2025uniqueness} fall short of proving convergence for online gradient ascent in this setting, absent stronger incentive guarantees.

These ``missing'' stronger incentives may be understood through the lens of $\Phi$-equilibria, measuring regret against a given set of strategy modifications \cite{greenwald2003general,stoltz2007learning}. For a fixed, finite set $\Phi$ of strategy modifications, these algorithms combine a no-regret learning algorithm, a potential function which maps regret over $\Phi$ to a strategy deviation, and access to fixed-point computation over the convex hull of strategy deviations \cite{gordon2008no}; their validity follows fundamentally from Blackwell's approachability theory \cite{blackwell1956analog}. We remark that both Hedge and ``lazy'' gradient ascent (dual averaging with a quadratic regulariser) are subsumed within the framework of \cite{gordon2008no}. 

But more directly relevant for us, recent work has shown that learning algorithms enjoy stronger guarantees in terms of regret guarantees when their associated dynamics are smooth; explicitly, external regret minimising algorithms can end up being no-regret for a larger set of strategy deviations as a consequence of additional structure in their time evolution. This is hinted in the result of \cite{anagnostides2022optimistic}, where it is shown that optimistic mirror descent with constant step sizes $O(\epsilon^2)$ and a smooth regulariser reaches either an $\epsilon$-Nash equilibrium or a $\textnormal{poly}(\epsilon)$-strong CCE. In turn, \cite{piliouras2022evolutionary} shows that the multiplicative weights update, in $2\times 2$ normal form games, exhibits the strongest form of $\Phi$-regret for the mixed-extension of the game. Results pertaining directly to online gradient ascent then come from recent breakthroughs in non-concave game theory. \cite{cai2024tractable} show that online gradient ascent incurs vanishing $\Phi$-regret for a class of local strategy modifications, even in non-concave games. \cite{ahunbay2024local} then identifies these strategy modifications as those generated by gradient fields of functions; and shows that the continuous trajectories generated by online gradient ascent incurs vanishing regret against such strategy modifications in first-order. 

\section{Preliminaries}

\newcommand{\bilin}[2]{\left\langle #1 , #2 \right\rangle}
\newcommand{\TC}[2]{\mathcal{TC}_{#1}[#2]}
\newcommand{\NC}[2]{\mathcal{NC}_{#1}[#2]}

In what follows, we will denote by $\mathbb{N}$ the set of (non-zero) natural numbers, $\mathbb{R}$ the set of real numbers, and $\mathbb{R}_+$ the set of non-negative reals. For each $N \in \mathbb{N}$, we will identify $N$ also with the set $\{n \in \mathbb{N} \ | \ n \leq N \}$. Following standard game theoretic notation, we shall often denote an indexed tuple by dropping its subscript, e.g. writing $s \equiv (s_i)_{i \in N}$. Given such a tuple $s$, we shall also denote by $s_{-i}$ the tuple $(s_j)_{j \in N, j \neq i}$, and write $(s'_i,s_{-i})$ for the tuple which differs from $s$ only on its $i$'th term with $s_i$ replaced by $s'_i$. This notation also extends to differentiation; e.g. given a set $X = \times_{i \in N} X_i$ and a function $f : X \rightarrow \mathbb{R}$, $\nabla_i f(x)$ denotes the coordinates of the gradient $\nabla f(x)$ over $X_i$. Given a finite set $S$, we let $\Delta(S)$ be the set of probability distributions over $S$. Then, $\mathbb{E}_{s \sim \sigma}[f(s)]$ denotes the expectation of $f : S \rightarrow \mathbb{R}$ when its input is drawn from $\sigma \in \Delta(S)$. We shall also write $\mathbb{I}$ for the proposition valued indicator function which equals $1$ if it is true, and $0$ otherwise.

Regarding matrices and vectors, we shall alternate notation depending on context. When vector-matrix manipulation is emphasised, we will use standard notation, with $(\cdot)^T$ denoting the transpose of a vector or matrix. In this case, for a pair of vectors $v$ and $w$, $v^Tw$ and $vw^T$ denote respectively their inner and outer products. When the inner product structure is emphasised, we shall denote $\bilin{v}{w} \equiv v^Tw$, and when the sums need to be evaluated we shall expand them in summation notation. We will use $\Sigma$ to denote an all-ones vector of appropriate size\footnote{Conveniently, $\sum_i v_i = \Sigma^T v$ for any vector $v$, and a matrix with all entries equal to $1$ may be written $\Sigma\Sigma^T$.}, and denote as $e_i$ the vector which has $1$ on its $i$'th coordinate and $0$ otherwise. Finally, $\mathbf{1}$ will denote the identity matrix, and for a subset of coordinates $C$, $\mathbf{1}_C$ will denote the matrix with its block $C$ equal to the identity matrix, and all other entries $0$.

For a convex set $X \subseteq \mathbb{R}^D$, we will write $\Pi_X(x)$ as the projection of $x$ onto $X$; i.e. $\Pi_X(x) = \arg \min_{x' \in x} \| x-x' \|^2$. Two important classes of projections are those onto the tangent and normal cones; the tangent cone $\TC{X}{x}$ to $X$ at $x$ is the closure of the pointed cone $\{ \eta (y -x)  \ | \ \eta \geq 0, y \in X\}$, and the normal cone is its dual, $\NC{X}{x} = \{ y  \ | \ \forall z \in \TC{X}{x}, y^Tz \leq 0\}$. For the reader's convenience, our notation is tabulated in Appendix \ref{sec:notation}.

\subsection{Normal Forms Games \& Classical Notions of Equilibria}

Our central object of study in this paper will be a finite normal-form game of complete information, which we shall simply refer to as a ``normal-form game''. Formally:

\begin{definition}[Normal-form games and mixed-strategies]\label{def:NFG}
    A \textbf{normal-form game} is specified by a tuple $\Gamma = \left(N,(A_i)_{i \in N},(u_i)_{i \in N}\right)$. Here, $N$ is the number and set of players (agents). Each agent $i$ has an action set $A_i$, and we denote by $A \equiv \times_{i \in N} A_i$ the set of pure action profiles, or outcomes. Payoffs to agents depend on the outcome, specified by a utility function $u_i : A \rightarrow \mathbb{R}$ for each agent $i$. We will allow agents to randomise their actions; a \textbf{mixed strategy} for agent $i$ is a probability distribution $x_i \in \Delta(A_i)$, and a vector $x \equiv (x_i)_{i \in N}$ of mixed strategies for each agent is called a \textbf{(mixed) strategy profile}.
\end{definition}

The usual notion of an expected, ``stable'' outcome of a game is that of a Nash equilibrium (NE). Intuitively, in a NE, each agent independently implements their mixed-strategies, such that no agent has any strict \emph{ex-ante} incentive in expectation to deviate to any other action.

\begin{definition}
    A \textbf{Nash Equilibrium} is a strategy profile $(x^*_i)_{i \in N}$ such that for every agent $i$ and every action $a'_i \in A_i$, $\mathbb{E}_{\forall j \in N, a_j \sim x^*_j}[u_i(a'_i, a_{-i}) - u_i(a)] \leq 0$.
\end{definition}

Every finite normal-form game has a NE \cite{Nash50}, but computing one even approximately is $PPAD$-hard \cite{CD06,DGP09}. On the other hand, allowing for correlations between the players' actions in the distribution of outcomes allows for two generalisations of NE. The notion of a \emph{correlated equilibrium} is due to \cite{Aumann74}, whereas the equilibrium conditions for a \emph{coarse correlated equilibrium} were already considered in the work of \cite{hannan1957approximation}.

\begin{definition}
    Let $\sigma \in \Delta(A)$, then $\sigma$ is called a \textbf{correlated equilibrium} (CE) if for every $i \in N$ and $a_i, a'_i \in A_i$, $\mathbb{E}_{a_{-i} \sim \sigma(a_i,\cdot)}[u_i(a'_i,a_{-i}) - u_i(a)] \leq 0$. In turn, $\sigma$ is called a \textbf{coarse correlated equilibrium} (CCE) if for every $i \in N$ and $a'_i \in A_i$, $\mathbb{E}_{a\sim \sigma}[u_i(a'_i,a_{-i}) - u_i(a)] \leq 0$. 
\end{definition}

The set of (coarse) correlated equilibria of a game can be represented as the feasible region of an LP of size polynomial in the size of the normal-form game. Given an equilibrium concept, the resulting LP has $|A|$ non-negative variables, corresponding to $\sigma(a)$ for each $a \in A$, and a constraint $\sum_{a \in A} \sigma(a) = 1$ which enforces that $\sigma \in \Delta(A)$. For a CCE, we have $\sum_{i \in N} |A_i|$ equilibrium constraints, whereas a CE is defined via $\sum_{i \in N} |A_i| (|A_i|-1)$ equilibrium constraints instead\footnote{This LP formulation implicitly suggests no-internal regret \cite{foster1997calibrated} rather than no-swap regret \cite{blum2007external}. The two notions are equivalent for exact, but not for approximate equilibria; $\epsilon$-swap regret is stronger than $\epsilon$-internal regret.}.

\subsection{Smooth / Concave Games, Mixed Extensions, and $\Phi$-Equilibria}

The notions of (coarse) correlated equilibria can also be interpreted as instances of $\Phi$-equilibria \cite{greenwald2003general,stoltz2007learning}, a notion of equilibrium that is based on a prescribed set $\Phi$ of \emph{action transformations}. Depending on the specification of $\Phi$ and players' utility functions, this notion of equilibrium may exist also for games with possibly infinite action sets; moreover the empirical distribution of play resulting from some learning algorithms can be shown to converge to such equilibria  \cite{stoltz2007learning,gordon2008no}. Two important cases to consider are when the game is smooth or concave \cite{cai2024tractable}. 

\begin{definition}
    A \textbf{smooth game} is specified by a tuple $\Gamma = \left(N,(X_i)_{i \in N}, (u_i)_{i \in N}\right)$. Generalising from Definition \ref{def:NFG} of normal-form games, $X_i \subseteq \mathbb{R}^{D_i}$ as action sets are taken to be compact\footnote{We will consider relaxing this to include unbounded polyhedra in Theorem \ref{thm:P-scaled-sets}.} and convex, and $u_i : X \rightarrow \mathbb{R}$ are differentiable functions with Lipschitz gradients, i.e. $\exists \ G_i, L_i > 0$ such that for any $x,x' \in X$, $\| \nabla_i u_i(x)\| \leq G_i$ and $\|\nabla_i u_i(x) - \nabla_i u_i(x')\| \leq L_i \| x - x'\|$. The game $\Gamma$ is said to be \textbf{concave} if all $u_i$ are concave functions in the argument $x_i$. 
\end{definition}

\begin{definition}
    For a smooth game $\Gamma$, consider a set $\Phi$ of \textbf{action transformations} (or \textbf{strategy modifications}) $\phi_i : X_i \rightarrow X_i$ for each player $i$. Then a distribution $\sigma$ on $X$ is called a $(\epsilon,\Phi)$-equilibrium if for any $i \in N$ and $\phi_i \in \Phi$, $\mathbb{E}_{x \sim \sigma}[u_i(\phi_i(x_i),x_{-i})-u_i(x)] \leq \epsilon$. 
\end{definition}

Throughout this paper, we focus on normal-form games; however, a normal-form game can be considered a game with infinite action sets, via its \emph{mixed extension}.  

\begin{definition}
    The \textbf{mixed extension} of a normal-form game $\Gamma$ is a smooth game $\Gamma'$ such that $X_i = \Delta(A_i)$, and the utilities are calculated via expectation; for any player $i$ and any mixed-strategy profile $x$, $u_i(x) = \sum_{a \in A} u_i(a) \cdot \prod_{j \in N} x_j(a_j)$.
\end{definition}

The mixed extension of any normal-form game is smooth and concave, which implies that for certain choices of $\Phi$, algorithms with access to fixed-points for convex combinations of the action transformations in $\Phi$ can compute an $(\epsilon$,$\Phi)$-equilibrium in $O(1/\epsilon^2)$ iterations \cite{gordon2008no}. Of particular interest to us will be the following two classes of $\Phi$-equilibria.

\begin{example}
    For a smooth game, we denote by $\Phi_{EXT}$ the set of constant transformations $X_i \rightarrow X_i$, and $\Phi_{LIN}$ the set of linear transformations $X_i \rightarrow X_i$ for each $i \in N$. Then for a normal-form game, the set of its $\Phi_{EXT}$-equilibria corresponds to coarse correlated equilibria, whereas the set of its $\Phi_{LIN}$-equilibria corresponds to its correlated equilibria.  
\end{example}

\subsection{Projected Gradient Ascent \& Local Correlated Equilibria}

Our goal in this paper is to prove an $\Phi$-equilibrium refinement over coarse correlated equilibria, which captures sharper guarantees when players learn their strategies via projected gradient ascent\footnote{Also known as \emph{generalised infinitesimal gradient ascent} (GIGA) \cite{zinkevich2003online}, or \emph{online gradient ascent}.} on their utilities. That is, over time periods $t = 1, 2, \ldots, T$, suppose that each player $i$ updates their mixed strategies in the full feedback model in the steepest direction of ascent, fixing 
\begin{equation}\label{eq:GIGA}x_i^{t+1} = \Pi_{X_i}[x_i^t + \eta_{it} \cdot \nabla_i u_i(x^t)].\end{equation}
Here, $\Pi_X[x]$ denotes the projection of the vector $x$ onto the set $X$, and $\eta_{it}$ is the step size for player $i$ at time $t$. It follows from \cite{zinkevich2003online} that if each player $i$ chooses step sizes $\eta_{it} \propto 1/\sqrt{t}$, then the time-averaged distribution of play $\sigma(a) = (1/T) \cdot \sum_{t = 1}^T \prod_{j \in J} x^t_j(a_j)$ forms a $O(1/\sqrt{T})$-approximate CCE. Meanwhile, projected gradient ascent is known not to converge to a correlated equilibrium in general; a deeper reason for this is because continuous-time gradient dynamics of games can avoid all strategies used in any correlated equilibrium \cite{viossat2015evolutionary}.

However, recent work on non-concave games \cite{cai2024tractable,ahunbay2024local} shows that if all players employ projected gradient ascent, the sequence of play satisfies a greater set of equilibrium constraints than those implied merely by no-external regret. These constraints on the sequence of play are characterised by small changes in the players' actions, or their differential limit.

\newcommand{\poly}{\textnormal{poly}}

\begin{definition}
    For a smooth game $\Gamma$, an $(\epsilon,\Delta)$-\textbf{local correlated equilibrium} with respect to a set $F$ of Lipschitz continuous vector fields $f : X \rightarrow \mathbb{R}^{D}$ is a probability distribution $\sigma \in \Delta(X)$, such that for any $f \in F$ and any $\delta \in (0,\Delta]$,
    $$ \sum_{i\in N} \mathbb{E}_{x\sim \sigma}[u_i(\Pi_{X_i}[x_i + \delta f_i(x)],x_{-i}) - u_i(x)] \leq \epsilon\delta \cdot \poly(G,L,G_f,L_f),$$
    where $G_f$ is a bound on the magnitude of $f$ on $X$ and $L_f$ is its Lipschitz modulus. Considering the limit $\delta \downarrow 0$, $\sigma$ is called an $\epsilon$\textbf{-local correlated equilibrium} if for any $f \in F$,
    $$ \sum_{i\in N} \mathbb{E}_{x\sim \sigma}[\langle \Pi_{\TC{X_i}{x_i}}[f_i(x)] , \nabla_i u_i(x) \rangle] \leq \epsilon \cdot \poly(G,L,G_f,L_f).$$
    Finally, if $F$ is the set of gradient fields of a set $H$ of functions $h : X \rightarrow \mathbb{R}$, i.e. $f = \nabla h$, the equilibrium is said to be \textbf{coarse} with respect to $H$.
\end{definition}

Our analysis will follow from the guarantees within  \cite{ahunbay2024local}, which apply for coarse equilibria generated by functions of tangential gradient fields. For what follows, we will say that $h : X \rightarrow \mathbb{R}$ is \textbf{tangential} if $\nabla h(x) \in \Pi_{\TC{X}{x}}[\nabla h(x)]$ for any $x \in X$; intuitively the gradient of $h$ should specify strategy modifications that do not require projecting back onto the set of action profiles.

\begin{proposition}[\cite{ahunbay2024local}, Theorems 3.9 \& 5.3]\label{prop:LCCE}
    In a smooth game, suppose that a subset $N' \subseteq N$ of players implement projected gradient ascent with the same non-increasing step sizes $\eta_t$, and suppose that for any $i \in N'$, $X_i$ either has a smooth boundary of bounded curvature, or is a polytope define via inequalities $\bilin{a_{i \ell}}{x_i} \leq b_\ell$ such that for any $\ell \neq \ell'$, $\bilin{a_{i\ell}}{a_{i\ell'}}\leq 0$. Then for any tangential $h : X \rightarrow \mathbb{R}$ with bounded, Lipschitz continuous gradients, if $\nabla_j h(x) = 0$ for any $j \notin N'$, at any $T > 0$,
    $$ \frac{1}{T} \cdot \sum_{t = 1}^T \sum_{i \in N'} \bilin{\nabla_i h(x^t)}{\nabla_i u_i(x^t)} \leq \frac{1}{T} \cdot \left( \frac{1}{\eta_{T}} + \frac{1}{\eta_1} + \sum_{t=1}^T \eta_t\right) \cdot \poly(G,L,G_h,L_h).$$
    As a consequence, if all players $i \in N'$ employ projected gradient ascent with step sizes satisfying $\eta_{t} = \omega(1/t)$ and $\sum_{t = 1}^T \eta_t = o(T)$, they incur vanishing regret against the set of all such $h$.
\end{proposition}

Important for our work here, Proposition \ref{prop:LCCE} extends to a refined equilibrium notion over usual CCE for normal form-games that is valid for the outcomes of projected gradient ascent, as the defining inequalities of the probability simplices $\Delta(A_i)$ satisfy the desired regularity assumption. Moreover, the framework allows bounding the expectation of quantities over the set of mixed-strategy profiles $\times_{i \in N} \Delta(A_i)$, allowing for primal-dual proofs of performance guarantees. Explicitly, suppose that $\gamma \in \mathbb{R}$ and $d, h : X \rightarrow \mathbb{R}$ are continuously differentiable functions, with $h$ tangential with Lipschitz gradients, such that
\begin{equation}\label{eq:PD-sol}\forall \ x \in X, \gamma + \sum_{i \in N} \bilin{\nabla_i h(x)}{\nabla_i u_i(x)} \geq d(x),\end{equation}
Then the time average distribution of play $\sigma$ resulting from projected gradient ascent, obtained by sampling $x^t$ for $1 \leq t \leq T$ with probability $1/T$, satisfies $\mathbb{E}_{x\sim \sigma}[d(x)] \leq \gamma + \epsilon \cdot \poly(G,L,G_h,L_h)$, where $\epsilon$ depends on the step sizes $\eta_t$ and $T$ as in Proposition \ref{prop:LCCE}. The result follows from considering (\ref{eq:PD-sol}) as dual feasibility constraints on $(\gamma,h)$ for the primal problem
\begin{align}
    \max_{\sigma \geq 0} \int_X d\sigma(x) \cdot d(x) \textnormal{ subject to } &  \label{opt:cont-PD}\\
    \int_X d\sigma(x) & = 1 \tag{$\gamma$}\\
   \forall \ h \in H, \int_X d\sigma(x) \cdot \sum_{i \in N} \bilin{\Pi_{\TC{X_i}{x_i}}[\nabla_i h(x)]}{\nabla_i u_i(x)} & \leq 0. \tag{$\epsilon_h$} \label{con:eq-h}
\end{align} 
The consideration of such local (coarse) correlated equilibria can then be considered an equilibrium refinement over usual (C)CE, as $\epsilon$-correlated equilibria of a normal-form game are equivalent to $\epsilon$-local correlated equilibria of its mixed extension with respect to the set of vector fields 
$F = \{ e_{ia_i} - e_{ia'_i} \ | \ i \in N, a_i,a'_i \in A_i \}$. 
Likewise, $\epsilon$-coarse correlated equilibria of a normal-form game are equivalent to the $\epsilon$-local coarse correlated equilibria of its mixed extension with respect to the set of functions $\{-\|x_i - x^*_i\|^2/2 \ | \ i \in N, x^* \in \Delta(A_i) \}$.

\section{Equilibrium Refinement \& Semicoarse Equilibria}\label{sec:semicoarse-def}

Our central goal in this paper is to refine the guarantees of projected gradient ascent through linear programming. For an illustrative example, suppose that we would like to argue that in a normal-form game, when all players employ projected gradient ascent with the same step sizes $\eta_t \propto 1/\sqrt{t}$, the aggregate distribution of play converges to a probability distribution supported only over the set $O \subseteq A$. Since we know that the aggregate distribution of play $\sigma(a) = \frac{1}{T} \cdot \sum_{t = 1}^T \prod_{j \in J} x^t_j(a_j)$ converges to a coarse correlated equilibrium as $T \rightarrow \infty$, one way to prove convergence in this sense would be to show that the value of the LP
\begin{align}
    \max_{\sigma \geq 0} \sum_{a \notin O} \sigma(a) \textnormal{ subject to } &  \label{opt:easy}\\
    \sum_{a \in A} \sigma(a) & = 1 \tag{$\gamma$}\\
   \forall \ i \in N, a'_i \in A_i, \sum_{a \in A} \sigma(a) \cdot (u_i(a'_i,a_{-i}) - u_i(a)) & \leq 0 \tag{$\epsilon_i(a'_i)$}
\end{align}
equals $0$. This is, for instance, the approach considered in \cite{ahunbay2025uniqueness} to certify (near-)convergence to equilibrium of no-regret learning in Bayesian single-item auctions. However, even in very simple games where such convergence behaviour is \emph{``obvious''}, it is possible for the notion of coarse correlated equilibrium to be too weak to certify it. 

\begin{example}\label{ex:bad-game}
    Consider the following $2 \times 3$ bimatrix game:
    \begin{center}
        \begin{tabular}{ r|c|c|c| }
            \multicolumn{1}{r}{}
             &  \multicolumn{1}{c}{L}
             & \multicolumn{1}{c}{M} & \multicolumn{1}{c}{R} \\
            \cline{2-4}
            T & 0,1 & 0,0 & 0,0 \\
            \cline{2-4}
            B & 0,0 & 0,0 & 0,1 \\
            \cline{2-4}
        \end{tabular}
    \end{center} 
    Player $1$, the row player, has utility gradient $0$ for any choice of strategies $x_1 \in \Delta(A_1)$, whereas player $2$, the column player, clearly best responds by playing $L$ if $x_1(T) > 1/2$, by playing $R$ if $x_1(B) < 1/2$, or by playing any distribution which assigns $x_2(M) = 0$ otherwise. Specifically, at any $x_1 \in \Delta(A_1)$, the action $M$ is utility minimising for player $2$ and  $\Pi_{\TC{\Delta(A_2)}{x_2}}[\nabla_2 u_2(x_1,x_2)]_M \leq 0$, with equality if and only if $x_2(M) = 0$. However, for small enough $\epsilon$, 
    $$ \sigma(a) = \begin{cases}
        \frac{1}{2} - \epsilon & (a_1,a_2) \in \{(T,L),(B,R)\} \\
        \epsilon & a_2 = M \\
        0 & \textnormal{otherwise}
    \end{cases}$$
    is a coarse correlated equilibrium. Thus for $O = \{a \in A \ | \ a_2 \neq M \}$, the LP (\ref{opt:easy}) has value $> 0$.
\end{example}

Our conclusion is that we need to somehow strengthen the LP (\ref{opt:easy}), refining its feasible region via the addition of further, linear equilibrium constraints. For a given distribution $\sigma \in \Delta(X)$, denote $\tilde\sigma(a) = \int_X d\sigma(x) \cdot \prod_{j \in N} x_j(a_j)$ as the induced probability distribution over the set of outcomes. Comparing (\ref{opt:easy}) with (\ref{opt:cont-PD}), we see that to identify our equilibrium refinement we need to identify the maximal set of functions $H$ such that for each $h \in H$, there exists a function $\tilde{h} : A \rightarrow \mathbb{R}$, such that for any $\sigma \in \Delta(X)$,
\begin{equation}\label{eq:shatter}
\int_X d\sigma(x) \cdot \sum_{i \in N} \bilin{\nabla_i h(x)}{\nabla_i u_i(x)} = \sum_{a \in A} \tilde{h}(a) \cdot \tilde{\sigma}(a).\end{equation}
This equality is, of course, satisfied if and only if it holds pointwise in $x$ for the integrands. Or in words, the inner product of $\nabla_i h$ and $\nabla_i u_i$ over the set of mixed strategies must \emph{``shatter''} nicely onto a sum over the set of pure action profiles. Now, in the mixed extension of the normal-form game, utilities are given via expectation and $u_i(x) = \sum_{a \in A} u_i(a) \cdot \prod_{j \in N} x_j(a_j)$. Therefore, 
$\partial u_i(x) / \partial x_i(a'_i) = \sum_{a_{-i} \in A_{-i}} u_i(a'_i,a_{-i}) \cdot \prod_{j \neq i} x_j(a_j)$. 
Hence, the condition (\ref{eq:shatter}) is equivalent to the existence of a function $\tilde{h}$ such that for any $x \in X$,
\begin{equation}
    \sum_{i \in N} \sum_{a'_i \in A_i, a_{-i} \in A_{-i}} \frac{\partial h(x)}{\partial x_i(a'_i)} \cdot u_i(a'_i,a_{-i}) \cdot \prod_{j \neq i} x_j(a_j) = \sum_{a \in A} \tilde{h}(a) \cdot \prod_{j \in N} x_j(a_j).
\end{equation}
The RHS is multilinear, and thus so too must be the LHS. Now, consider setting $h(x) = \sum_{i \in N} \frac{1}{2} x_i^T Q_i x_i + q_i^T x_i$ for some symmetric matrices $Q_i$ and vectors $q_i$ for each player $i$. Then for any pair of players $i \neq j$ and any pair of actions $a_i \in A_i,a_j \in A_j$, $\partial^2 h(x) / \partial x_i(a_i) \partial x_j(a_j) = 0$. Moreover, for each $i$,
\begin{align*}
    & \quad \sum_{a_i'\in A_i, a_{-i} \in A_{-i}} \frac{\partial h(x) }{\partial x_i(a'_i)} \cdot u_i(a'_i,a_{-i}) \cdot \prod_{j \neq i} x_j(a_j) \\ & = \sum_{a_i'\in A_i, a_{-i} \in A_{-i}} \left( q_i(a'_i) \cdot 1 + \sum_{a_i \in A_i} Q_i(a'_i,a_i) \cdot x_i(a_i) \right) \cdot u_i(a'_i,a_{-i}) \cdot \prod_{j \neq i} x_j(a_j) \\
    & = \sum_{a'_i \in A_i, a \in A} \left( \prod_{j \in N} x_j(a_j) \right) \cdot \left( Q_i(a'_i,a_i) + q_i(a'_i) \right) \cdot u_i(a'_i,a_{-i}) \\
    & = \sum_{a \in A} \tilde{\sigma}(a) \cdot \sum_{a'_i \in A_i} \left( Q_i(a'_i,a_i) + q_i(a'_i) \right) \cdot u_i(a'_i,a_{-i}). \\
    \Rightarrow \tilde{h}(a) & = \sum_{i \in N} \sum_{a'_i \in A_i} \left( Q_i(a'_i,a_i) + q_i(a'_i) \right) \cdot u_i(a'_i,a_{-i}).
\end{align*}
Here, the second inequality follows from $\sum_{a_i \in A_i} x_i(a_i) = 1$, and the third follows from exchanging the order of the sum. Thus $\tilde{h}(a)$ is a conical sum of functions  
$\sum_{a'_i \in A_i} \left( Q_i(a'_i,a_i) + q_i(a'_i) \right) \cdot u_i(a'_i,a_{-i})$ defined by a pair $(Q_i,q_i)$ for some player $i$, which suggests that our inequality constraints should be generated by functions $h(x) = \frac{1}{2} x_i^T Q_i x_i + q_i^T x_i$ for some player $i$ and pair $(Q_i,q_i)$. 

We still require conditions on $(Q_i,q_i)$ such that $\nabla h$ is tangential; otherwise, the constraints (\ref{con:eq-h}) is equal to neither term in (\ref{eq:shatter}) due to the tangent cone projection required, and thus the inequality constraints can no longer be added to the LP (\ref{opt:easy}). Since the value of $h$ can only change with $x_i$,  $\nabla_i h(x)$ must point to the relative interior of the probability simplex $\Delta(A_i) = \{ x_i \in \mathbb{R}^{d_i} \ | \ \sum_{a_i \in A_i} x_i(a_i) = 1, x_i \geq 0\}$. For this, $\nabla_i h(x)$ must conserve probability, hence we require 
$\sum_{a'_i \in A_i} Q_i(a'_i,a_i) \cdot x_i(a_i) + q_i(a'_i) = 0$ for any $x_i \in \Delta(A_i)$. Considering the cases when $x_i(a''_i) = 1$ for some $a''_i$ and $= 0$ otherwise, we conclude that $(Q_i,q_i)$ satisfies
\begin{equation}\label{con:preserve}
    \sum_{a'_i \in A_i} Q_i(a'_i,a_i) + q_i(a'_i) = 0 \ \forall \ a_i \in A_i.
\end{equation}
Moreover, whenever $x_i(a_i') = 0$, it must be the case that $\partial h(x) / \partial x_i(a'_i) \geq 0$, otherwise following along $\nabla_i h(x)$ would result in assigning negative probability to $a'_i$, necessitating a projection back onto $\Delta(A_i)$. Formally, we require
$$ \frac{\partial h(x)}{\partial x_i(a'_i)} = \sum_{a_i \in A_i} Q_i(a'_i,a_i) \cdot x_i(a_i) + q_i(a'_i) \geq 0 \ \forall \ x_i \in \Delta(A_i), x_i(a'_i) = 0.$$
Similarly considering the cases when $x_i(a_i) = 1$ for $a_i \neq a_i$, we conclude
\begin{equation}\label{con:tangency}
    Q_i(a'_i,a_i) + q_i(a'_i) \geq 0 \ \forall \ a_i \neq a'_i.
\end{equation}
Finally, the symmetry of $Q_i$ may be expressed explicitly, 
\begin{equation}\label{con:symmetry}
    Q_i(a'_i,a_i) - Q_i(a_i,a'_i) = 0 \ \forall \ a_i, a'_i \in A_i.
\end{equation}
These conditions on $(Q_i,q_i)$ altogether allows us to pinpoint our desired equilibrium refinement.

\begin{definition}\label{def:semicoarse}
    For a normal-form game $\Gamma$, a probability distribution $\sigma \in \Delta(A)$ is called a \textbf{semicoarse correlated equilibrium} if for any player $i$ and any pair $(Q_i,q_i)$ which satisfy the tangency conditions (\ref{con:preserve}), (\ref{con:tangency}), and the symmetry condition (\ref{con:symmetry}),
    \begin{equation}\label{con:semicoarse}
        \sum_{a \in A} \sigma(a) \cdot \sum_{a'_i \in A_i} \left( Q_i(a'_i,a_i) + q_i(a'_i) \right) \cdot u_i(a'_i,a_{-i}) \leq 0.
    \end{equation}
\end{definition}

By Proposition \ref{prop:LCCE}, we infer that we do not necessarily require all players to employ the same step sizes for projected gradient ascent after all; as each pair $(Q_i,q_i)$ specifies a tangential function which depends on $x_i$ only. As long as for each player $i$, the non-increasing step sizes satisfy $\eta_{it} = \omega(1/t)$ and $\sum_{t=1}^T \eta_{it} = o(T)$, any convergent subsequence of the aggregate distribution of play does so to a semicoarse equilibrium as $T \rightarrow \infty$.

We remark that whereas the set of possible $(Q_i,q_i)$ forms a polyhedral cone, they can be interpreted as generators of a subclass of linear endomorphisms on $\Delta(A_i) \rightarrow \Delta(A_i)$. That is, for small enough $\delta$, $x_i \mapsto x_i + \delta(Q_ix_i + q_i)$ is indeed such a map; (7) \& (8) ensure the mapping is an endomorphism, whereas (9) enforces the symmetry of $Q$. As a consequence, the resulting strategy modifications in fact correspond to a subset $\Phi_{SEMI}$ of the strategy modifications $\Phi_{LIN}$ of the mixed extension of the game; and hence the set of semicoarse equilibria of a normal-form game is a superset of its correlated equilibria.

\begin{proposition}\label{prop:linsubset}
    Suppose the pair $(Q_i,q_i)$ satisfies the tangency \& symmetry conditions. Then there exists $\delta > 0$ such that the matrix
    $P_i(a'_i,a_i) = \mathbf{1}(a'_i, a_i) + \delta \cdot \left( Q_i(a'_i,a_i) + q_i(a'_i) \right)$ 
    is a (left) stochastic matrix. As a consequence, if $\sigma$ is a semicoarse correlated equilibrium, then player $i$ incurs no regret against such a strategy modification $x_i \mapsto P_i x_i$.
\end{proposition}

Moreover, the subset $\Phi_{SEMI}$ of $\Phi_{LIN}$ we consider is \emph{maximal} for the no-regret guarantees of the  gradient dynamics of the game. That is, with initial conditions $x(0)$, a solution to the projected gradient dynamics of the normal-form game is a curve\footnote{That such a solution exists in our setting follows from \cite{dupuis1993dynamical,cojocaru2004existence}.} $x : [0,\infty) \rightarrow X$ such that for any player $i$, 
$ dx_i(t) /dt = \Pi_{\TC{X_i}{x_i(t)}}[\nabla_i u_i(x(t))]$
almost everywhere on $[0,\infty)$. As remarked in \cite{ahunbay2024local}, whenever the curve remains in the relative interior of $X = \times_{i \in N} \Delta(A_i)$, then for any continuously differentiable function $h : X \rightarrow \mathbb{R}$,
$\frac{1}{T} \cdot \int_0^T dt \cdot \sum_{i \in N} \bilin{\nabla_i h(x(t))}{\nabla_i u_i(x(t))} = \frac{h(T)-h(0)}{T}$, 
which goes to $0$ as $T \rightarrow \infty$. However, replacing $\nabla h$ with an arbitrary vector field $f$ above does not guarantee that the expectation of the inner product along the curve vanishes. And for a linear transformation defined by a non-conservative vector field $x_i \mapsto P_i x_i$, we may fix utilities and an initial condition $x(0)$, such that the solution $x$ is cyclic, and for the induced probability distribution $\sigma$ on $X$, player $i$ incurs positive regret against the associated strategy modification. 

\begin{proposition}\label{prop:linmax}
 Let $N \geq 2, (A_i)_{i \in N}$ be fixed where $|A_j| > 2$ for at least two players $j$. Then for any linear transformation $\phi : \Delta(A_i) \rightarrow \Delta(A_i)$ in $\Phi_{LIN} \setminus \Phi_{SEMI}$, there exists utilities $(u_i)_{i \in N}$ such that for some learning cycle in the resulting game, player $i$ incurs positive regret against $\phi$.
\end{proposition}

We defer the proofs of Propositions \ref{prop:linsubset} and \ref{prop:linmax} to Appendix \ref{sec:semicoarse-def-proofs}.

\subsection{Polyhedral Representations of Semicoarse Equilibria}\label{sec:representation}

Our discussion in Section \ref{sec:semicoarse-def} suggests that, to obtain sharper guarantees for the outcomes of projected gradient ascent, we could simply analyse the set of semicoarse equilibrium of the game. However, the set of possible $(Q_i,q_i)$ is uncountably infinite in general. This also obfuscates any intuition regarding the strategy modifications; while we can see from (\ref{con:semicoarse}) and Proposition \ref{prop:linsubset} that the actions $a_i$ are transformed to a probability distribution over actions $a'_i$ in some manner, \emph{how} these transformations are done are not yet clear. To illustrate the kind of strategy modifications that we may need to include, it is illustrative to revisit Example \ref{ex:bad-game}.

\begin{example}\label{ex:good-proof}
    In Example \ref{ex:bad-game}, consider the function 
    $h(x) = -\frac{1}{4}\cdot\left( 1 - x_2(L) - x_2(R) \right)^2 - \frac{1}{2} \cdot x_2(M)$. Then $\nabla_2 h(x) = x_2(M) \cdot (1/2,-1,1/2)$, and we observe immediately that $\nabla_2 h(x)$ is a tangential vector field over $\Delta(A_1) \times \Delta(A_2)$. Moreover, $\nabla_2 h(x)$ admits a straightforward interpretation; player $2$ deviates from the action $M$ by playing $L$ or $R$ both with probability $1/2$ each. Meanwhile, they do not deviate from playing $L$ or $R$. Therefore, it is not possible to represent this strategy modification as a combination of uniform deviations.
    
    Moreover, we remark that 
    $\bilin{\nabla_2 h(x)}{\nabla_2 u_2(x)} = x_2(M) \cdot (x_1(T) + x_1(B))/2 = x_2(M)/2 \geq 0$, 
    with equality only when $x_2(M) = 0$. However, the semicoarse equilibrium constraint necessitates that $\int_X d\sigma(x) \cdot \bilin{\nabla_2 h(x)}{\nabla_2 u_2(x)} \leq 0$, which is only possible if $\sigma$ assigns all probability to the facet of $\Delta(A_1) \times \Delta(A_2)$ such that $x_2(M) = 0$. This implies that the semicoarse equilibrium constraints adequately capture the behaviour of gradient dynamics in this game.
\end{example}

From Example \ref{ex:good-proof}, one may conjecture that the defining strategy modifications of semicoarse equilibrium are generated by functions of the form 
\begin{equation}\label{eq:gen-ansatz}h(x) = -\frac{1}{2|A_i \setminus S_i|} \cdot \left( 1 - \sum_{a_i \in A_i \setminus S_i} x_i(a_i) \right)^2 - \frac{1}{2} \cdot \sum_{a_i \in S_i} x_i(a_i)^2\end{equation}
for proper subsets $S_i$ of $A_i$, corresponding to deviating from $a_i$ to the uniform distribution on $A_i \setminus S_i$ for each $a_i \in S_i$ and not deviating otherwise. This turns out to be not entirely correct, but provides a key idea for the proof of our characterisation theorem. Namely, we first recall that the set of possible $(Q_i,q_i)$ which defines the generators of our linear transformations form a polyhedral cone; this follows since the tangency conditions (\ref{con:preserve}), (\ref{con:tangency}), and (\ref{con:symmetry}) are all $= 0$ or $\geq 0$ constraints. A polyhedral cone is generated by its finitely many rays, and thus characterising them is equivalent to characterising the strategy modifications that define semicoarse equilibria by Proposition \ref{prop:linsubset}. Our observation is that the functions (\ref{eq:gen-ansatz}) \emph{completely capture} the effect of the constant vector $q_i$. Factoring out their contribution to $(Q_i,q_i)$ then allows us to reduce the problem to that of representations of symmetric doubly stochastic matrices as a convex combination of symmetric $0-1$ matrices, which were studied in \cite{cruse1975note}. The result, the proof of which is deferred to Appendix \ref{sec:representation-proofs}, is then:

\begin{theorem}\label{thm:strategy-def}
    The set of inequalities (\ref{con:semicoarse}) are equivalent to the inequalities
    \begin{equation}
        \sum_{a \in A} \sigma(a) \cdot \left[ \left( \sum_{a'_i \in A_i} P_i(a'_i,a_i) 
 \cdot u_i(a'_i,a_{-i})\right) - u_i(a) \right] \leq 0 \ \forall \ P_i \in \mathbf{P}(A_i),
    \end{equation}
    where $\mathbf{P}(A_i)$ is the set of left stochastic matrices corresponding to the strategy modifications:
    \begin{enumerate}
        \item For a proper subset $S_i \subset A_i$, each action $a_i \in S_i$ is transformed to the uniform distribution on $A_i \setminus S_i$, whereas for $a_i \in S_i \setminus A_i$, $a_i \mapsto a_i$.
        \item For a cycle of actions $C_i = \{a_{i1}, a_{i2}, \ldots, a_{ik}\}$, writing $a_{i(k+1)} = a_{i1}$, each action $a_{i\ell}$ is transformed to the uniform distribution on $\{a_{i(\ell-1)},a_{i(\ell+1)}\}$, whereas for $a_i \in A_i \setminus C_i$, $a_i \mapsto a_i$.
    \end{enumerate}
\end{theorem}

We remark that Theorem \ref{thm:strategy-def} shows that the set of semicoarse correlated equilibrium of a game are defined by a superexponential\footnote{For each player $i$, there are $\sum_{k = 2}^{|A_i|} \binom{|A_i|}{k} \cdot \frac{(k-1)!}{2}$ cycles of actions; the index $k$ starts at $2$ as we allow for cycles of length $2$.} number of constraints; and even restricting attention to the arguably ``more interesting'' transformations considered in Theorem \ref{thm:strategy-def}.(1), we are left with $2^{|A_i|}-2$ equilibrium constraints for each player $i$. Unlike both correlated and coarse correlated equilibria, this representation of the equilibrium set is totally unsuited for numerical evaluation.

However, the set of semicoarse equilibria of a game does turn out to have low extension complexity, which we show in Appendix \ref{sec:representation-proofs}. The high level idea is that each $(Q_i,q_i)$ are defined by $O(|A_i|^2)$ variables subject to $O(|A_i|^2)$ constraints, which allows us to obtain a succint LP formulation via taking the dual of the inner LP in the associated feasibility problem. 

\begin{theorem}\label{thm:short-extension}
    For a normal-form game $\Gamma$, suppose that $(\sigma,\omega,\rho)$ satisfies the linear (in)equalities:
    \begin{align}
        \sum_{a \in A} \sigma(a) & = 1, \label{cons:prob} \tag{$\gamma$} \\
        \forall \ i \in N, \ a_i \neq a'_i, \gamma_i(a'_i,a_i) + \rho_i(a_i,a'_i) - \rho_i(a'_i,a_i) + \sum_{a_{-i} \in A_{-i}} \sigma(a) \cdot \left( u_i(a'_i,a_{-i}) - u_i(a) \right) & = 0 \label{cons:Q-off} \tag{$Q_i(a'_i,a_i)$}\\
        \forall i \in N, \forall a'_i \in A_i, \sum_{a_i \neq a'_i} \gamma_i(a'_i,a_i) + \sum_{a \in A} \sigma(a) \cdot \left( u_i(a'_i,a_{-i}) - u_i(a) \right) & = 0 \label{cons:q} \tag{$q_i(a'_i)$} \\
        \sigma,(\gamma_i)_{i \in N} & \geq 0.
    \end{align}
    Then $\sigma$ is a semicoarse correlated equilibrium of $\Gamma$. As a consequence, the set of semicoarse correlated equilibrium of a normal-form game can be represented via an extension which is polynomial in the size of the normal-form game, with $|A| + \sum_{i \in N} \frac{3}{2}|A_i|(|A_i|-1)$ variables and $1 + \sum_{i \in N} |A_i|^2$ constraints.
\end{theorem}

Whereas the extension in Theorem \ref{thm:short-extension} is tractable to evaluate in general, we find that it is prone to numerical instability when input into a solver. However, for an objective $d : A \rightarrow \mathbb{R}$, the LP
\begin{align}
    \max_{\sigma,\omega,\rho} \sum_{a \in A} \sigma(a) \cdot d(a) \textnormal{ subject to (\ref{cons:prob}),(\ref{cons:Q-off}),(\ref{cons:q})} \label{opt:primal}
\end{align}
has the same value as the dual Lyapunov function estimation problem,
\begin{align}
    \min_{\gamma,(Q_i,q_i)_{i \in N}} \gamma \textnormal{ subject to } & \textnormal{(\ref{con:preserve}), (\ref{con:tangency}), (\ref{con:symmetry}) for each $(Q_i,q_i)$, and } \label{opt:lyapunov-LP} \\
    \gamma + \sum_{i \in N} \sum_{a'_i \in A_i} \left( Q_i(a'_i,a_i) + q_i(a'_i) \right) \cdot u_i(a'_i,a_{-i}) & \geq d(a) \ \forall \ a \in A, \label{cons:lyapunov-cond} \tag{$\sigma(a)$}
\end{align}
which turns out to be better behaved in practice. The form of the dual problem is, of course, in line with the duality between a primal problem for worst-case bounds for a game's (coarse) correlated equilibria and its dual best-fit Lyapunov function problem, as remarked in \cite{ahunbay2024local}.

\section{Semicoarse Equilibrium of (Simple) Auction Games}

Examples \ref{ex:bad-game}, \ref{ex:good-proof} show that the notion of a semicoarse equilibrium can indeed provide sharper guarantees for the gradient dynamics of a normal-form game, compared to those of a coarse correlated equilibrium. However, the game we considered there is rather artificial; which raises the question of whether there are interesting settings wherein semicoarse equilibrium constraints allow us sharper predictions, and if so, how tight these predictions are. In this section, we shall address this question, by considering two settings. In the complete information model of Bertrand competition with linearly decreasing demand, we will see that the semicoarse equilibrium of the game is necessarily close to the equilibrium of the game whenever there are at least two firms of lowest marginal cost. Meanwhile, our analysis of the complete information first-price auction will highlight a special property of semicoarse equilibria -- sensitivity to the ``density'' of actions -- which will motivate us to consider how our equilibrium concept generalises to arbitrary polyhedral action sets.

\subsection{Convergence to Equilibrium in Bertrand Competition}

We consider a discretised Bertrand competition \cite{bertrand1883review} game with complete information and linear costs for firms. In our setting, $N$ firms who supply an identical item compete in a market. Each firm $i$ has marginal cost $c_i/n \in \{0,1/n,2/n,...,1\}$, and they post a price $p_i \in \{0,1/n,2/n,...,1\} = A_i$ simultaneously. Given the price vector $(p_i)_{i \in N}$, the demand is captured by the firm posting the lowest price, with ties broken uniformly at random. Therefore, the utilities are given
$$u_i(p) = \begin{cases}
    \frac{1}{|\arg\min_{j \in N} p_j |}  \cdot (p_i-c_i/n) D(p_i) & p_i \in \arg\min_{j \in N} p_j  \\
    0 & \textnormal{otherwise.}
\end{cases}$$
Here, $D : [0,1] \rightarrow (0,\infty)$ is the demand function, generally taken to be non-increasing in the price. If a firm $i$ has cost $c_i/n$, we will also refer to $(p-c_i/n)D(p)$ as the associated profit function.

\begin{example}\label{ex:bertrand-bad-CCE}
    Consider the case with two firms with $0$ marginal cost, where the demand is decreasing linearly in price, $D(p_i) = 1-p_i$ up to normalisation. With continuous action sets $p_i \in [0,1]$, the unique equilibrium of the game is in pure strategies, where both firms post price $p_i = 0$. However, the continuous game does not have a unique coarse correlated equilibrium. Following a construction similar to \cite{FLN16}, consider the case when both firms post the same price $p$, following the distribution on $[0,1/4]$ with probability density function $f(p) = 1/\sqrt{p}$. In this case, each firm makes profit
    $\int_0^{1/4} dp \cdot f(p) \cdot p(1-p)/2 = 17/480 \simeq .035417...$ 
    On the other hand, if buyer $i$ deviates uniformly to bidding $p_i > 1/4$, they obtain $0$ utility. Meanwhile, deviating to a price $p_i \in [0,1/4]$, their utility equals
    $p_i(1-2\sqrt{p_i})(1-p_i)$, with maximum $\simeq .0330892... $
    
    Then, for the size $n$ discretisation of the game, we consider in (\ref{opt:primal}) maximising the expected square distance from the equilibrium bid, 
    $\max_{\sigma \geq 0} \sum_{p_1,p_2 \in \{0,1/n,2/n,...,1\}} \sigma(p_1,p_2) \cdot \left( p_1^2 + p_2^2 \right)$, 
    over the set of the game's coarse and semicoarse correlated equilibria. Figure \ref{fig:bertrand} illustrates our observation; the discretised game \emph{``inherits''} poorly behaved coarse correlated equilibria of its continuous counterpart, whereas its distance maximal semicoarse correlated equilibrium assigns probability $1$ on the strategy pair $(p_1,p_2) = (1/n,1/n)$.

    \begin{figure}
        \begin{center}
            \includegraphics[width=0.49\textwidth]{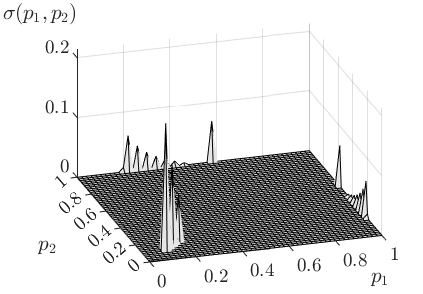}   
            \includegraphics[width=0.49\textwidth]{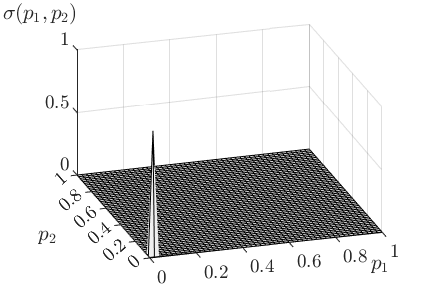}
        \end{center}
        \caption{Distributions $\sigma$ which are square price maximising over the set of coarse (left) and semicoarse (right) correlated equilibria, for the discretisation with $n = 50$.}\label{fig:bertrand}
    \end{figure}
\end{example}

The observation of Example \ref{ex:bertrand-bad-CCE} turns out to be generalisable. In particular, we obtain results that parallel those of \cite{deng2022nash} regarding mean-based learning and convergence to Nash equilibria\footnote{In this version, we have opted to not discuss precise convergence rates to retain focus on the characterisation \& validity of the equilibrium concept. The reader, if interested in the derivation of exact time-average / finite iterate convergence guarantees, is referred to Appendix \ref{sec:explicit-convergence}.}. The high level implication is that when all firms employ projected gradient ascent with the same step sizes, under suitable concavity assumptions on the profit function, when the number of firms with minimum marginal cost is $\geq 2$ we have convergence to Nash equilibrium. 

\begin{theorem}\label{thm:bertrand-unique}
    Suppose there exists at least two firms with minimum marginal cost $c/n \in \{0,1/n,...,1-1/n\}$, and that $(p - c/n) D(p)$ is strictly concave over $[c/n,1]$. Then in every semicoarse equilibrium $\sigma$ of the discretised Bertrand competition game with bidding sets $\{0,1/n,...,1\}$, each outcome $(p_i)_{i \in N}$ assigned positive probability $\sigma(p) > 0$ is a pure strategy Nash equilibrium of the game, with at least two firms of minimum marginal cost posting prices in $c/n + \{0,1/n\}$.
\end{theorem}

The proof of Theorem \ref{thm:bertrand-unique} proceeds by constructing an explicit solution to the dual Lyapunov function estimation problem (\ref{opt:lyapunov-LP}), where $f(p) = 0$ in the case two firms post prices in $c/n + \{0,1/n\}$ and $= 1$ otherwise. We consider, for the firms with marginal cost $c/n$, action transformations $\phi_k$ of the form of those in Theorem \ref{thm:strategy-def}.(1), where each $\phi_k$ maps actions \emph{not} in $B^k = \{c+1/n,...,c+k/n\}$ to the uniform distribution on $B^k$. In turn, for firms with marginal cost $> c/n$, we only impose deviating away from the bids strictly lesser than their marginal cost. A brief case analysis shows that the interesting price vectors have multiple firms of minimum marginal cost who post the same minimum price $> (c+1)/n$; strict concavity of the profit function then allows us to inductively define dual multipliers for each $\phi_k$. The details are deferred to Appendix \ref{sec:bertrand-proof}.

We are able to conclude a slightly weaker result when the function $(p-c/n)D(p)$ is merely weakly concave. This covers the case when the demand is \emph{inelastic}, i.e. $D(p)$ is constant, which implies that $(p-c/n)D(p)$ is affine-linear. The key observation is that the strict concavity assumption enters our proof only in the case when there are exactly two firms with minimum marginal cost through applications of Jensen's inequality; and hence the identical line of arguments works for weak concavity when there are at least \emph{three} firms with minimum marginal cost.

\begin{theorem}\label{thm:bertrand-unique-three}
    Suppose there exists at least \underline{three} firms with minimum marginal cost $c/n \in \{0,1/n,...,1-1/n\}$, and that $(p - c/n) D(p)$ is \underline{weakly concave} over $[c/n,1]$. Then in every semicoarse equilibrium $\sigma$ of the discretised Bertrand competition game with bidding sets $\{0,1/n,...,1\}$, each outcome $(p_i)_{i \in N}$ assigned positive probability $\sigma(p) > 0$ is a Nash equilibrium.
\end{theorem}

\subsection{First-Price Auctions and the Gauge-Dependency Problem}\label{sec:first-and-gauge}

In the discretised first-price auction with complete information, there are $N$ buyers who bid on a single-item. Each buyer $i$ has value $v_i/n \in \{0,1/n,...,1\}$, and they simultaneously place bids $b_i \in  \{0,1/n,...,1\} = A_i$. Given the bid vector $(b_i)_{i \in N}$, the buyer with the highest bid wins the item, with ties broken uniformly at random. The buyers' utilities are thus given,
$$u_i(p) = \begin{cases}
    \frac{1}{|\arg\max_{j \in N} b_j |}  \cdot (v_i/n - b_i) & b_i \in \arg\min_{j \in N} p_j  \\
    0 & \textnormal{otherwise.}
\end{cases}$$
Letting $c_i = n-v_i$, $p_i = 1-b_i$ and $D(p) = 1$, the first-price auction is thus equivalent to the Bertrand competition with complete information and inelastic demand. We may thus conclude:

\begin{corollary}\label{thm:fp-unique-three}
    Suppose there exists at least three buyers with maximum valuation $v/n \in \{1/n,...,1\}$. Then in every semicoarse equilibrium $\sigma$ of the discretised first-price auction with bidding sets $\{0,1/n,...,1\}$, each outcome $(b_i)_{i \in N}$ assigned positive probability $\sigma(p) > 0$ is necessarily a pure strategy Nash equilibrium, where at least two buyers of maximum value $v/n$ bid $v/n - \{0,1/n\}$.
\end{corollary}

We thus would like to take a closer look at the setting with two buyers, both with a valuation of $1$. Keeping in mind the equivalence with Bertrand competition, we see that if the utilities were instead defined, for an arbitrary $\varepsilon > 0$
$$u_i(p) = \begin{cases}
    \frac{1}{|\arg\max_{j \in N} b_j |}  \cdot (1 - b_i) \cdot (1+\varepsilon \cdot b_i)) & b_i \in \arg\min_{j \in N} p_j  \\
    0 & \textnormal{otherwise,}
\end{cases}$$
then the unique semicoarse correlated equilibrium of the auction would again concentrate the highest two bids amongst the bidders on $1-\{0,1/n\}$; for the equivalent Bertrand competition game, $D(p) = 1 + \varepsilon - \varepsilon p$ in this case, and the profit function $(p+v/n-1)D(p)$ is strictly concave. Restricting attention to the case when $v/n = 1$, we see that 
$(1-b_i)(1+\varepsilon b_i) = 1 - b_i + \varepsilon b_i(1-b_i)$,
which corresponds to shifting the bids by a quadratic polynomial. This hints that if the set of bids $A_i$ were to be replaced by the corresponding bids, all semicoarse equilibrium of the game would result in both buyers concentrating their bids close to the Nash equilibrium.

\begin{figure}
        \begin{center}
            \includegraphics[width=0.49\textwidth]{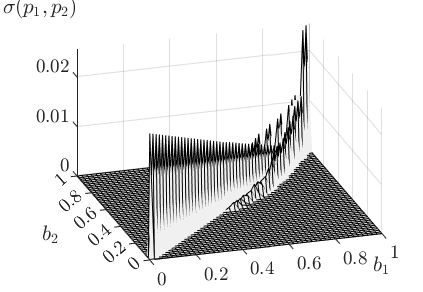}   
            \includegraphics[width=0.49\textwidth]{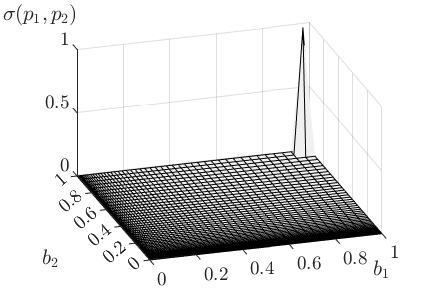}
        \end{center}
        \caption{Distributions $\sigma$ which are square distance maximising from the equilibrium bids $(1,1)$, for the first-price auction with gauges $A^{(1)}$ and $A^{(2)}$ and discretisation size $n = 50$.}\label{fig:gauge}
\end{figure}

In Figure \ref{fig:gauge} we demonstrate this phenomenon, presenting optimal solutions to (\ref{opt:primal}) with its objective $\sum_{b \in A} \sigma(b) \cdot ((1-b_1)^2 + (1-b_2)^2)$ maximising the square distance of the bids from the Nash equilibrium bid $1$. We consider our previous choice of $A^{(1)}_i = \{0,1/n,...,1\}$ as well as the modified set of bids, $A^{(2)}_i = \{0,1/n^2,(2/n)^2,...,1\}$, which corresponds to a choice\footnote{We have found that picking $\varepsilon$ smaller tends to result in numerical instability for the solver. Our dual solution in Appendix \ref{sec:bertrand-proof} provides an explanation; the dual solutions diverge to $\infty$ as $\varepsilon$ is picked closer to $0$.} of $\varepsilon = 1$. We observe that the former results in a semicoarse correlated equilibrium clearly far from the Nash equilibrium in square distance, whereas numerical solutions to the latter problem suggest bids are necessarily concentrated on $\{1,(1-1/n)^2\}$. The action set $A^{(2)}_i$ is, of course, equivalent to the pre-image $F^{-1}(\{0,1/n,...,1\})$ for the function $F(b) = \sqrt{b}$. Thus, $A^{(2)}_i$ may be seen as a discrete, uniformly spaced interval from a different gauge on $[0,1]$; intuitively, modifying the \emph{density} of the discretised action set within the action set of the corresponding continuous auction results in certifiable convergence to near-equilibrium.

\section{Weighted Actions Sets \& Semicoarse Equilibria Under Linear Transformations}

The sensitivity of the set of semicoarse equilibria of the first-price auction with respect to the distribution of the discretised bids motivates us to inspect the consequences of changing the density of each individual action, without actually modifying the actions themselves. Formally, we consider a normal-form game $\Gamma = (N,(A_i)_{i\in N},(u_i)_{i \in N})$, and integer weights $w_i \in \mathbb{N}^{A_i}$ on the actions of each player $i$. The $w$-weighted normal-form game $\Gamma^w = (N,(A^w_i)_{i \in N},(u^w_i)_{i \in N})$ is then such that each $A_i^w$ admits a partition $\cup_{a_i \in A_i} \kappa(a_i)$, defined by the map $\kappa_i : A_i \rightarrow 2^{A_i^w}$, such that $|\kappa(a_i)| = w_i(a_i)$ for every action $a_i$. By the partition property, $\kappa_i$ admits a left-inverse, $\lambda_i : A^w_i \rightarrow A_i$, such that $\lambda\kappa(a_i) = \{a_i\}$ for any $a_i \in A_i$. For each $a^w \in A^w$ the utilities are then defined, $u_i^w(a^w) = u_i(\lambda(a^w))$. Intuitively, we extend our original game by creating $w_i(a_i)$ copies of each action $a_i$. 

In this case, each distribution $\sigma^w \in \Delta(A^w)$ induces a distribution on $\sigma \in \Delta(A)$ by contracting $\sigma^w$ along $\kappa$, i.e. $\sigma(a) = \sum_{a^w \in \kappa(a)} \sigma^w(a^w)$. It is straightforward to verify that if $\sigma^w$ is a (coarse) correlated equilibrium of the weighted game $\Gamma^w$, then $\sigma$ is a (coarse) correlated equilibrium of the original game $\Gamma$. However, it is in general false that the induced distribution $\sigma$ is a semicoarse equilibrium whenever $\sigma^w$ is so. Leveraging the fact that two polyhedra $P,P'$ are equal if and only if $\max_{x \in P} x^Tc = \max_{x \in P'} x^T c$ for any vector $c$, we show the following equilibrium constraints instead apply to $\sigma$ for a given set of weights $w$; the details of the proof, as all others in this section, are deferred to Appendix \ref{sec:scale-proofs}.

\begin{theorem}\label{thm:weight-constraint}
    Suppose that $\sigma^w \in \Delta(A^w)$ is a semicoarse equilibrium of the weighted normal-form game $\Gamma^w$. Then the induced distribution $\sigma(a) = \sum_{a^w \in \kappa(a)} \sigma^w(a^w)$ on $A$ satisfies the equilibrium constraints
    \begin{equation}
        \sum_{a \in A} \sigma(a) \cdot \left[ \left( \sum_{a'_i \in A_i} P_i(a'_i,a_i) 
 \cdot u_i(a'_i,a_{-i})\right) - u_i(a) \right] \leq 0 \ \forall \ P_i \in \mathbf{P}(A_i),
    \end{equation}
    where $\mathbf{P}(A_i)$ is the set of left stochastic matrices corresponding to the strategy modifications:
    \begin{enumerate}
        \item For a proper subset $S_i \subset A_i$, each action $a_i \in S_i$ is transformed to the distribution $F$ on $A_i \setminus S_i$, such that $F(a'_i) \propto w_i(a'_i)$ for each $a'_i \in A_i \setminus S_i$. Meanwhile, for $a_i \in S_i \setminus A_i$, $a_i \mapsto a_i$.
        \item For a cycle of actions $C_i = \{a_{i1}, a_{i2}, \ldots, a_{ik}\}$, writing $a_{i(k+1)} = a_{i1}$ and $\delta = \min_{\ell} w_i(a_{i\ell})$, each action $a_{i\ell}$ is transformed to the distribution on $\{a_{i(\ell-1)},a_{i\ell}, a_{i(\ell+1)}\}$, where $a_{i\ell}$ has probability $(1-\delta/w_i(a_{i\ell}))$ and the actions $\{a_{i(\ell-1)},a_{i(\ell+1)}\}$ are equally likely. Whereas for $a_i \in A_i \setminus C_i$, $a_i \mapsto a_i$.
    \end{enumerate}
\end{theorem}

Thus, given non-constant weights on players' action sets, we obtain distinct equilibrium constraints for $\sigma \in \Delta(A)$; we will call $\sigma$ a $w$\textbf{-weighted semicoarse equilibrium} whenever it is induced by the semicoarse equilibrium $\sigma^w$ of the $w$-weighted game. Turning our attention back to the setting of Bertrand competition, we note that the deviations we considered for the firms with lowest cost become strictly utility improving in aggregate, if the weights are strictly increasing in the prices. This allows us to conclude convergence in the setting when only two firms have the lowest marginal cost, though with respect to a slightly larger set of prices.

\begin{proposition}\label{prop:weighted-bertrand}
    Suppose there exists at least two firms with minimum marginal cost $c/n \in \{0,1/n,...,1-1/n\}$, and that $(p - c/n) D(p)$ is weakly concave over $[c/n,1]$. Then if $w : \{0,1/n,...,1\} \rightarrow \mathbb{N}$ is strictly increasing, then in every $w$-weighted semicoarse equilibrium $\sigma$ of the discretised Bertrand competition game with bidding sets $\{0,1/n,...,1\}$, each outcome $(p_i)_{i \in N}$ assigned positive probability $\sigma(p) > 0$ is necessarily a Nash equilibrium. An analogous statement also holds for the first-price auction when $w$ is strictly decreasing, through its equivalence with Bertrand competition with inelastic demand.
\end{proposition}

In particular, our LP-based formalism results in different certifiable convergence behaviour over weighted versions of the game. The deeper reason for this is because \emph{the weights distort the gradient dynamics of the game}. Indeed, for the $w$-weighted version of the game, at the interior of the action set, the projected utility gradient has terms 
\begin{align*}  \frac{ \partial u^w_i(x^w)}{\partial x^w_i(a^w_i)} - \frac{1}{|A^w_i|} \sum_{a'^w_i \in A^w_i}  \frac{ \partial u^w_i(x^w)}{\partial x^w_i(a'^w_i)}, \\ \textnormal{ with a contracted sum } \left(\sum_{a_i^w \in \kappa_i(a_i)}  \frac{\partial u^w_i(x^w)}{\partial x^w_i(a^w_i)}\right) - \frac{w_i(a_i)}{|A^w_i|} \sum_{a'^w_i \in A^w_i} \frac{ \partial u^w_i(x^w)}{\partial x^w_i(a'^w_i)}. \end{align*}
This implies that in the original game, in the interior of $\Delta(A_i)$, player $i$'s utility gradient is projected effectively with respect to a $w$-weighted average of the utilities for playing each action $a_i$.

We would like to obtain a general description of such changes in the projection of the utility gradients, and the resulting equilibrium constraints. This can be modelled through a linear transformation of the feasible strategy set over which projected gradient ascent is employed; for a normal form game, for each player $i$ and a set of invertible linear transformations over $\mathbb{R}^{A_i}$ defined via matrices $P_i$, we consider the smooth game with action sets $Y_i = P_i^{-1} \Delta(A_i)$ and utilities 
\begin{equation}\label{def:util-scaled}u^P_i(y) = \sum_{a \in A} \left(\prod_{j \in N} \sum_{a''_i \in A_i} P_j(a_j,a''_j) y_j(a''_j) \right) \cdot u_i(a) \textnormal{ for each player } i.\end{equation}
In this case, $u^P_i(y) = u_i(x)$, where $x_i = P_i y_i \in \Delta(A_i)$ for the game's mixed extension. 

The effect of such transformations were considered before for \emph{ensuring} convergence to a fixed-point of a dynamical system in \cite{heemels2020oblique}; indeed, for a general dynamical system, linear transformations of the feasible set do not in general preserve the set of fixed-points unless they are conformal (i.e. when they preserve angles between vectors). However, for the gradient dynamics of a normal-form game, linear transformations of individual strategy sets $\Delta(A_i)$ do not change the set of its Nash equilibria. We thus wonder whether projected gradient ascent over such transformed sets have similar LP based guarantees for their outcomes.

The answer requires two parts. First, the result of \cite{ahunbay2024local} that we leverage (Proposition \ref{prop:LCCE}) applies only when the inequality constraints which delineate the strategy space satisfy a sign condition on their pairwise inner products. This condition is always satisfied by the probability simplex $\Delta(A_i)$. However, when the linear transformation $P_i$ is non-conformal, it is not necessarily the case that $P_i^{-1} \Delta(A_i)$ also satisfies this condition. Towards this end, we extend Proposition \ref{prop:LCCE} to arbitrary polyhedra when the function $h$ is tangential and bounded, with bounded, Lipschitz continuous gradients. The robustness parameter $\Delta$ for $(\epsilon,\Delta)$-local coarse correlated equilibria in this case depends on a condition number $\chi(X)$, which we define in Lemma \ref{lem:angle-bound} in Appendix \ref{sec:scale-proofs}.

\begin{theorem}\label{thm:P-scaled-sets}
    In a smooth game, suppose that players $i \in N'$ all employ the same non-increasing step sizes $\eta_t$, and have polyhedral strategy sets $X_i$. Let $h : X \rightarrow \mathbb{R}$ be a bounded tangential function, with bounded and Lipschitz gradients, which depends only on coordinates $(x_i)_{i \in N'}$ -- i.e. for some $M, G_h, L_h > 0$, for any $x, x' \in X$,
    $|h(x)| \leq M$, $ \| \nabla h(x) \| \leq G_h \textnormal{ and } \| \nabla h(x) - \nabla h(x') \| \leq L_h$, and $\nabla_j h(x) = 0$ for any $x \in X$. Then for any $T > 1$, and any $\delta \in [0,1/(L_h \cdot \chi(X)]$,
    \begin{align*}
        & \frac{1}{T} \cdot \sum_{i \in N} \sum_{t = 1}^{T} (u_i(\Pi_{X_i}[x^t_i+\delta \nabla_i h(x^t)],x^t_{-i})-u_i(x^t)) \\
        \leq \ & \delta \cdot \left( \frac{2M}{T} \left[\frac{1}{\eta_{T}} + \frac{1}{\eta_1}\right]+ \frac{\sum_{t = 1}^T \eta_t}{T} \cdot \sum_{i \in N'}\left(\frac{G_i^2 L_h}{2} + G_iL_h \cdot \sum_{j \in N'}G_j\right) \right) + \frac{1}{2T} \cdot \delta^2 G_h^2 \sum_{i \in N'} L_i.
    \end{align*}
    As a consequence, when all players use the same step sizes $\eta_t = C/\sqrt{t}$, after $T$ rounds of projected gradient ascent, the time-averaged distribution $\sigma$ which samples $x^t$ with probability $1/T$ is an $(\epsilon,\Delta)$-local coarse correlated equilibrium with respect to the set of bounded tangential functions $H$ with bounded gradients and  Lipschitz modulus $\leq 1$, where we may fix 
    $\epsilon = ( 4M/C + 2\max\{1,C\} )/\sqrt{T}$ and $\Delta = \min\{1/\chi(X),1/\sqrt{T}\}$. Moreover, evaluating $\delta \downarrow 0$, $\sigma$ is an $\epsilon$-local coarse correlated equilibrium with respect to $H$.
\end{theorem}

Second, we need to deduce the appropriate set of linear inequalities that the aggregate probability distribution satisfies. The arguments required here essentially identical to our initial construction of semicoarse equilibria, due to the multilinearity of the utilities (\ref{def:util-scaled}) over $\times_{i \in N} \Delta(A_i)$. 

\begin{theorem}\label{thm:P-scaled}
    Suppose that each player $i$ employs projected gradient ascent on their utilities over their scaled strategy sets $Y_i = P_i^{-1}\Delta(A_i)$ with decreasing step sizes $\eta_{it} = \omega(t), o(1)$. Denote the time averaged distribution $\sigma^T$, sampling each $x^t = (P_iy_i^t)_{i \in N}$ with probability $1/T$, and $(P_iP_i^T) = Z_i$. Then for any  convergent subsequence $\sigma^{T_k} \rightarrow \sigma$ and for any player $i$,
    $$ \sum_{a \in A} \sigma(a) \cdot \sum_{a'_i, a''_i \in A_i} Z_i(a'_i,a''_i) \cdot (Q_i(a''_i,a_i) + q_i(a''_i)) \cdot u_i(a'_i,a_{-i}) \leq 0$$
    whenever $Q_i$ is a symmetric matrix such that $Z_i(Q_i,q_i)$ satisfies (\ref{con:preserve}), (\ref{con:tangency}).  
\end{theorem}

We shall thus call $\sigma$ a $P$\textbf{-scaled semicoarse equilibrium} whenever it satisfies the equilibrium constraints of Theorem \ref{thm:P-scaled}. Two observations are immediate; first, for each $P$, the set of $P$-scaled semicoarse equilibria of a normal form game necessarily contains the set of its coarse correlated equilibrium. This can be seen by setting, for $x_i^* \in \Delta(A_i)$ the characteristic vector of a pure strategy, $(Q_i,q_i) = (P_iP_i^T)^{-1}(\mathbf{1},x_i^*)$. Second, we in fact may recover $w$-weighted equilibria as a specific instance of $P$-scaled equilibria. In particular, the equilibrium constraints induced on $\sigma$ by $\sigma^w$ are those obtained if each player $i$ uses projected gradient ascent over a transformed action set $P^{-1}_i \Delta(A_i)$, where the matrix $P_i = \textnormal{diag}(\sqrt{w_i})$ is the diagonal matrix obtained via the square root of the action weights. That is to say, if each player $i$ updates 
$$ y_i^{t+1} = \arg\min_{y_i \in P_i^{-1}\Delta(A_i)} \| y - y_i^{t} - \eta_{it} \cdot P_i^T \nabla_i u_i(x^t)\|^2, x_i^{t+1} = P_i^{-1} y_i^{t+1},$$
then for appropriately chosen step sizes $\eta_{it}$, the time average distribution of play $\sigma$ satisfies the $w$-weighted equilibrium constraints in the limit $T \rightarrow \infty$.

\section{Concluding Remarks \& Insights}

In this paper, we have provided an LP based refinement over coarse correlated equilibrium that is valid for the outcomes of projected gradient ascent. The equilibrium refinement allowed us to provide stronger convergence results for projected gradient ascent, purely through linear programming arguments. In discretised Bertrand competitions, our equilibrium analysis provided guarantees parallelling those of \cite{deng2022nash} for the first-price auction with mean-based learners. 

We find two insights particularly striking. First, is that our clean, duality based framework allows us to completely circumvent the epoch based arguments and the associated probability theoretic machinery employed in \cite{kolumbus2022auctions,deng2022nash} in their analysis of mean-based learning. Whereas the \emph{construction} of the framework and classification of the relevant strategy modifications, through its theoretical roots in \cite{ahunbay2024local}, required effort, analysis of semicoarse equilibria is truly not significantly more difficult in comparison to the analysis of correlated and coarse correlated equilibria in games.

More interesting for us, is \emph{how} the semicoarse correlated equilibria of the Bertrand competition (and through its equivalence in the case of inelastic demand, the first-price auction) parallels the results of \cite{deng2022nash}. Specifically, there appears to be a correspondence between convergence guarantees for mean-based learners and convergence guarantees for projected gradient ascent. The sufficient conditions for last-iterate convergence for mean-based buyers in the first-price auction are precisely those that ensure that its semicoarse equilibria are convex combinations of Nash equilibria. Meanwhile, when time-average but not last-iterate convergence is guaranteed for mean-based learners, the set of the game's semicoarse equilibrium becomes ``unstable'', collapsing onto the equilibrium under arbitrarily small changes in the set of bids itself or weighings over it. 

We conjecture the existence of some relationship between the guarantees of semicoarse equilibria, and those for mean-based learners. However, at this point we can only speculate on its precise nature. For one, our equilibrium constraints are specialised for projected gradient ascent, and we do not expect players to incur vanishing regret against the corresponding strategy deviations when they employ e.g. Hedge as their learning algorithm in general. However, it is certainly plausible for there to exist ``translations'' between the guarantees of semicoarse equilibria, and dynamics of learning with mean-based algorithms; this is at present left as a future direction of inquiry. 

Another question we ask is whether further equilibrium refinement is possible, and how convergence guarantees or price of anarchy bounds would be tightened for the outcomes of projected gradient ascent. The result of \cite{ahunbay2024local} shows that any family of twice differentiable functions over the set of mixed strategies would generate a family of strategy modifications, allowing the construction of more refined notions of equilibria. However, we show in Proposition \ref{prop:linmax} that the strategy modifications which define semicoarse equilibria are maximal within the set of linear strategy modifications -- the latter corresponds to the usual notion of correlated equilibrium.

Any hypothetical equilibrium refinement over semicoarse equilibrium, valid for the outcomes of projected gradient ascent, would thus cease to be a subset of correlated equilibrium. It would necessarily have to account for non-linear strategy modifications. The construction of a framework which can account for such non-linearities is thus another open problem. Given the correspondence pointed out by \cite{ahunbay2024local} between primal linear programs over such equilibria and their dual Lyapunov function fit problems, we expect results on non-linear Lyapunov function estimation to be relevant. For instance, \cite{souaiby2021lyapunov} considers using semidefinite programming to compute Lyapunov functions, which suggests SOS programming might provide a hierarchy of further refined notions of equilibrium. We remark, however, that such a formulation may preclude the existence of a finite, simple set of strategy modifications as those established for semicoarse equilibria in Theorem \ref{thm:strategy-def}; as the set of semidefinite matrices is not polyhedral.

\vspace{8pt}\noindent\textbf{Acknowledgements.}  This project was funded by the Deutsche Forschungsgemeinschaft (DFG, German Research Foundation) under Grant No. 405445463.

\bibliographystyle{alpha}
\bibliography{sample-bibliography}

\appendix

\section{Notations Table}\label{sec:notation}

\begin{tabular}{p{3.7cm} p{1cm} p{10.5cm}}
		\hline
		General Notation & & \\ 
		\hline
		$\mathbb{N}$ & & Set of natural numbers, in the convention without $0$ \\
        $N$ & & For $N \in \mathbb{N}$, the set $\{1,2,...,N\}$ \\
		$\mathbb{R}$ & & Set of real numbers \\
		$\mathbb{R}_+$ & & Set of non-negative real numbers \\
            $A^B$ & & Set of functions $B \rightarrow A$ \\
		$\Delta(S)$ & & Set of probability distributions over $S$ \\
            $\mathbb{E}_{x \sim \sigma}[f(x)]$ & & Expectation of $f : S\rightarrow \mathbb{R}$ when $x$ is drawn from $\sigma$ \\
		$\mathbb{I}[s = s']$ & & Indicator function, equals $1$ if $s = s'$ and $0$ otherwise \\
            $\poly(x,y,z,..)$ & & A fixed polynomial over variables $x,y,z,...$ \\
            $\| \cdot \|$ & & The usual Euclidean norm \\
		\hline
        Smooth Functions & & \\
        \hline
            $f : S \subseteq \mathbb{R}^D \rightarrow \mathbb{R}$ & & A differentiable function \\
            $\nabla f(x)$ & & Gradient of $f$ \\
            $\nabla_i f(x_i,x_{-i})$ & & Gradient of $f$ with respect to $x_i$ \\
            $G_f$ & & Upper bound on $\|\nabla f(x)\|$ over $x \in S$\\
            $L_f$ & & Upper bound on $\frac{\|\nabla f(x)-\nabla f(x')\|}{\|x-x'\|}$ for $x,x' \in S$ distinct \\
        \hline
        Common Subscripts & & \\
        \hline
        $i,j$ & & Quantity relating to player $i, j$ \\
        $-i, (-i)$ & & Quantity relating to players other than $i$ \\
        \hline \hline
		Games & &  \\
        \hline \hline
        Normal-Form Games & & \\
        \hline
        $\Gamma$ && A normal form game \\
        $N$ & $i,j$ & Number, set of agents / players \\
        $A = \times_{i \in [N]} A_i$ & $a$ & Set of action profiles / outcomes \\
        $u_i : A \rightarrow \mathbb{R}$ & & Utility function for player $i$ \\
        $x_i \in \Delta(A_i)$ & & Mixed strategy for agent $i$ \\
        $\sigma \in \Delta(A)$ & & Correlated distribution over pure strategies \\
        \hline
        Smooth Games & &\\
        \hline
         $\Gamma$ && A smooth game \\
         $N$ & $i,j$ & Number, set of players \\
        $X = \times_{i \in N} X_i$ & $x$ & Set of action profiles \\
        $u_i : X \rightarrow \mathbb{R}$ & & Utility function for player $i$ \\
        $\sigma \in \Delta(X)$ & & Correlated distribution over strategies \\
        $\Phi \subseteq \sqcup_{i \in N} X_i^{X_i}$ & & Set of action transformations / strategy modifications \\
        \hline
\end{tabular}

\begin{tabular}{p{3.7cm} p{1cm} p{10.5cm}}
        \hline
        Mixed Extensions & & \\
        \hline
        $X_i = \Delta(A_i)$ & & Set of mixed strategies \\
        $u_i(x) = \mathbb{E}_{a_j \sim x_j}[u_i(a)]$ & & Utilities, defined via expectations \\
        $\Phi_{EXT}$ & & Set of constant transformations, $\phi : x_i\mapsto x^*_i$ \\
        $\Phi_{LIN}$ & & Set of linear transformations, $\phi : \Delta(A_i) \rightarrow \Delta(A_i)$ \\
        \hline
		\hline
		Vectors \& Matrices & & \\
        \hline 
        $v,w,q,...$ & & Vectors in $\mathbb{R}^D$ of appropriate dimension, often lower case \\
        $Q,A,P,...$ & & Matrices in $\mathbb{R}^{D\times D}$ of appropriate size, often upper case \\
        (+) $Q,q$ & & Will always denote generators of a strategy modification \\
        (+) $P$ & & A linear transformation $\Delta(A_i) \rightarrow \Delta(A_i)$, \emph{except} by Theorems \ref{thm:P-scaled-sets},\ref{thm:P-scaled}, a family of invertible linear maps $P_i : \mathbb{R}^{|A_i|} \rightarrow \mathbb{R}^{|A_i|}$ \\
        $v^T, Q^T$ & & Transpose of $v$, $Q$ respectively \\
        $v^Tw$, $\bilin{v}{w}$ & & The inner product, $\sum_i v_i w_i$ \\
        $vw^T$ & & The outer product, the matrix $Q_{ij} = v_iw_j$ \\
        $\Sigma$ & & The all-ones vector of appropriate size \\
        $e_i$ & & The characteristic vector with $i$'th coordinate $1$ and all others $0$ \\
        $ \mathbf{1}$ & & The identity matrix \\
        $\mathbf{1}_C$ & & For $C \subseteq D$, the matrix with its block $C \times C$ equal to identity, with all other entries equal to $0$  \\
        $\textnormal{diag}(v)$ & & The diagonal matrix $Q_{ii} = v_i$ \\
        $\textnormal{tr}(A)$ & & The trace, $\sum_{i} A_{i}$ \\
        \hline
        \end{tabular}

\begin{tabular}{p{3.7cm} p{1cm} p{10.5cm}}
        \hline
        Convex Sets & & \\
        \hline
        $\Pi_X[x]$ & & The projection, $\arg\min_{x' \in X} \|x-x'\|^2$ \\
        $\TC{X}{x}$ & & The tangent cone to $X$ at $x$, closure of $\{\eta (y-x) \ | \ \eta \geq 0, y \in X\}$ \\
        $\NC{X}{x}$ & & The normal cone to $X$ at $x$, $\{y \ | \ \forall \ z \in \TC{X}{x}, y^Tz \leq 0 \}$ \\
        $P$ & & Denotes a polyhedron, defined via inequalities $a_i^Tx \leq b_i$ \\
        $\chi(P)$ & & A condition number on $P$, as in Lemma \ref{lem:angle-bound} \\
        \hline
        Gradient Ascent & &  \\
        \hline
        $\eta_t$ & & Step size at time $t$ \\
        $x^t$ & & Mixed strategies at time $t$ \\
        \hline
\end{tabular}

\section{Projected Gradient Ascent is Not Mean-Based}\label{sec:not-mean-based}

We had remarked in Section \ref{sec:related} that projected gradient ascent is not a mean-based algorithm in general. Here, we clarify this statement. We will work (without loss of generality) when the utilities for actions are bounded in $[0,1]$.  Recall the definition of a mean-based algorithm:

\begin{definition}[Definition 1, \cite{deng2022nash}, also cf. \cite{braverman2018selling}]
    Let $\alpha_i^t(a_i)$ be the average reward of action $a_i$ in the first $t$ rounds: $\alpha_i^t(a_i) = \frac{1}{t} \cdot \sum_{\tau = 1}^t u_i(a_i,x_{-i}^t)$. An algorithm is then $\gamma_t$-mean based if, for any $a_i \in A_i$, whenever there exists $a'_i \in A_i$ such that $\alpha_i^{t-1}(a'_i) - \alpha_i^{t-1}(a_i) > \gamma_t$, $x_i^t(a_i) \leq \gamma_t$. An algorithm is mean-based if it is $\gamma_t$-mean-based for some decreasing sequence $(\gamma_t)_{t = 1}^\infty$ such that $\gamma_t \rightarrow 0$ as $t \rightarrow \infty$. 
\end{definition}

Explicitly, we will show that online (projected) gradient ascent, with full feedback, is not mean-based for any choice of step sizes $\eta_t = C \cdot t^{-\alpha}$ for any $\alpha \in [0,1)$ in the experts setting. In this case, $u_i^t \in [0,1]^{|A_i|}$ and $x_i^{t+1} = \Pi_{\Delta(A_i)}[x_i^t + \eta_t \cdot u_i]$. Our argument is rather simple; we will fix $a^*_i \in A_i$, and let $x_i^0$ be arbitrarily fixed. Then for $T \in \mathbb{N}$ such that $T \gg 1$, we will let $u^t_i(a_i) = 1$ for any $a_i \neq a^*_i$ and $u^t_i(a_i^*) = 0$ for any $t \leq T$. Then, for the next $K$ rounds, we will let $u^t_i(a^*_i) = 1$ and $u_i^t(a_i) = 0$ otherwise. We remark over the rounds $T+1,...,T+K$, $\Delta x_i^{t} = \Delta x^t_i(a^*_i) \geq C \cdot t^{-\alpha} / 2$ so long as $x_i^{t+1}(a_i) \neq \mathbb{I}[a_i = a^*_i]$. And thus, if $K$ is such that $\sum_{t = T+1}^{T+K} C \cdot t^{-\alpha}/2 \geq 1$, then $x_i^{t+1}(a_i) = \mathbb{I}[a_i = a^*_i]$. 

We will want to be able to pick $K \leq \epsilon T$ for some small, fixed $\epsilon > 0$. In this case, for large $T$, being able to pick $K$ such that $\sum_{t = T+1}^{T+K} C \cdot t^{-\alpha}/2 \geq 1$ is possible whenever
\begin{align*}
    \sum_{t = T+1}^{T+K} C \cdot t^{-\alpha}/2 & \simeq \int_T^{T+K} C \cdot t^{-\alpha}/2 \cdot dt = \frac{C \cdot T^{1-\alpha}}{2(1-\alpha)} \cdot \left( \left( 1+\frac{K}{T}\right)^{1-\alpha} - 1 \right)\\
    & \geq \frac{C(1+\epsilon)^{1-\alpha}}{2(1-\alpha)} \cdot \frac{K}{T^{\alpha}} \gg 1 \\
    \Rightarrow \epsilon T \geq K & \gg \frac{2T^{\alpha}(1-\alpha)}{C(1+\epsilon)^{1-\alpha}}.
\end{align*}
Thus, picking $T$ large enough, we have both that $\gamma_t < 1/2$ and that such a choice of $K$ is possible. In this case, mean-basedness would require $x_i^{T+K+1}(a_i^*) \leq 1/2$, since $\alpha_i^{T+K+1}(a_i) - \alpha_i^{T+K+1}(a^*_i) = T-K > (1-\epsilon)T > 1/2$. However, by construction $x_i^{T+K+1}(a_i) = 1$.

\section{Omitted Proofs}

\setcounter{subsection}{2}

\subsection{Equilibrium Refinement \& Semicoarse Equilibria}\label{sec:semicoarse-def-proofs}

\begin{proposition*}[\ref{prop:linsubset}]
    Suppose the pair $(Q_i,q_i)$ satisfies the tangency \& symmetry conditions. Then there exists $\delta > 0$ such that the matrix
    $ P_i(a'_i,a_i) = \mathbf{1}(a'_i, a_i) + \delta \cdot \left( Q_i(a'_i,a_i) + q_i(a'_i) \right) $
    is a (left) stochastic matrix. As a consequence, if $\sigma$ is a semicoarse correlated equilibrium, then player $i$ incurs no regret against the strategy modification $x_i \mapsto P_i x_i$.
\end{proposition*}

\begin{proof}
    Since there are finitely many pairs $a'_i, a_i \in A_i$, consider $\Delta = \max_{a_i \in A_i}  \sum_{a'_i \neq a_i} Q_i(a'_i, a_i) + q_i(a'_i)$. By constraint (\ref{con:tangency}), $\Delta \geq 0$. Moreover, if $\Delta = 0$, then the matrix $Q_i(a'_i,a_i) + q_i(a'_i)$ is the zero matrix due to the constraint (\ref{con:preserve}), in which case we may fix $\delta = 1$. So suppose that $\Delta > 0$, and let $\delta = 1/\Delta$.

    To show that $P_i(a'_i,a_i)$ is left stochastic, we need to show that it is non-negative, and all of its columns sum up to $1$. Non-negativity of the off-diagonal elements follows from the non-negativity of the off-diagonals of $Q_i(a'_i,a_i) + q_i(a'_i)$. Meanwhile, by the constraint (\ref{con:preserve}),
    \begin{align*}
        \sum_{a'_i \in A_i} P_i(a'_i,a_i) & = \sum_{a'_i \in A_i} \mathbf{1}(a'_i, a_i) + \delta \cdot \left( Q_i(a'_i,a_i) + q_i(a'_i) \right) \\
        & = 1 + \delta \cdot \sum_{a'_i \in A_i} \left( Q_i(a'_i,a_i) + q_i(a'_i) \right) = 1.
    \end{align*}
    Moreover, again by (\ref{con:preserve}), 
    \begin{align*}Q_i(a_i,a_i) & = - \sum_{a'_i \neq a_i} Q_i(a'_i, a_i) + q_i(a'_i) \geq -\Delta
    \end{align*}
    by the definition of $\Delta$, which implies that $P_i(a_i,a_i) \geq 1 - \delta \cdot \Delta = 0$.
\end{proof}

The proof of Proposition \ref{prop:linmax} in turn depends on finding a triplet of strategies over which a linear transformation $\phi : \Delta(A_i) \rightarrow \Delta(A_i)$ is not induced by a gradient field, and embedding a \emph{rock-paper-scissors} game over these strategies versus another player $j$ with action set of size $> 2$. The identification of the triplet of strategies follows from a characterisation of linear transformations which have representations as in Proposition \ref{prop:linsubset}.

\begin{lemma}\label{lem:triplet}
    Let $P \in \mathbb{R}^{|A_i| \times |A_i|}$ be a left stochastic matrix, specifying a linear transformation $\Delta(A_i) \rightarrow \Delta(A_i)$. Then there exists a pair $(Q_i,q_i)$ satisfying the tangency conditions (\ref{con:preserve}), (\ref{con:tangency}), and (\ref{con:symmetry}) such that 
    $$ P_i(a'_i,a_i) = \mathbf{1}(a'_i,a_i) + Q_i(a'_i,a_i) + q_i(a'_i)$$
    if and only if for every triplet of distinct strategies $a_i, a'_i, a''_i \in A_i$,
    \begin{equation}\label{eq:triplet}
        P_i(a_i,a'_i) + P_i(a'_i, a''_i) + P_i(a''_i,a_i) = P_i(a_i,a''_i) + P_i(a''_i,a'_i) + P_i(a'_i,a_i).
    \end{equation}
\end{lemma}

\begin{proof}
    The forward direction $(\Rightarrow)$ is immediate from evaluating each side of the equation, and the symmetry constraint (\ref{con:symmetry}) on $Q$. So we need to show $(\Leftarrow)$. Given $P$, fix some $a^*_i \in A_i$, and for any pair $a_i, a'_i \in A_i$, let 
    \begin{align*}
        q_i(a'_i) & = P_i(a'_i,a^*_i) - P_i(a'_i,a^*_i), \\
        Q_i(a'_i,a_i) & = P_i(a'_i,a_i) - \mathbf{1}(a'_i,a_i) - q_i(a'_i).
    \end{align*}
    The constraints (\ref{con:preserve}) and (\ref{con:tangency}) then hold for the pair $(Q_i,q_i)$ immediately because $Q_i(a'_i,a_i) + q_i(a'_i) = P_i(a'_i,a_i) - \mathbf{1}(a'_i,a_i)$. We thus need to check whether $Q_i$ is a symmetric matrix under this definition. In this case, for any pair $a_i \neq a'_i$,
    \begin{align*}
        Q_i(a'_i,a_i) - Q_i(a_i,a'_i) & = P_i(a'_i,a_i) - q_i(a'_i) - P_i(a_i,a'_i) + q_i(a_i) \\
        & = P_i(a'_i,a_i) - P_i(a'_i,a^*_i) + P_i(a^*_i,a'_i) - P_i(a_i,a'_i) + P_i(a_i,a^*_i) - P_i(a^*_i,a_i) \\
        & = \left( P_i(a'_i,a_i) + P_i(a_i, a^*_i) + P_i(a^*_i, a'_i) \right) - \left( P_i(a'_i,a^*_i) + P_i(a^*_i,a_i) + P_i(a_i, a'_i) \right),
    \end{align*}
    which equals zero identically whenever one of $a_i, a'_i$ equals $a^*_i$. Meanwhile, if $a_i, a'_i, a^*_i$ are distinct, the expression again equals zero by our assumption (\ref{eq:triplet}).
\end{proof}

We will also want to make explicit the connection between the regret against linear transformations $\phi: \Delta(A_i) \rightarrow \Delta(A_i)$, where $\phi(x_i) = P x_i$ for some matrix $P$, and regret against tangential linear vector fields over $\Delta(A_i)$ obtained via $x_i \mapsto (P - \mathbf{1}) x_i$. The argument of the lemma is implicit in our identification of the set of semicoarse equilibria, and the equivalence results in \cite{ahunbay2024local}.

\begin{lemma}\label{lem:lin-reg}
    If a function $\phi : x_i \mapsto P_i x_i$ is a linear transformation over $\Delta(A_i) \rightarrow \Delta(A_i)$, then $f(x_i) = (P_i -\mathbf{1})x_i$ is a tangential vector field, i.e. 
    \begin{align*}
        \sum_{a'_i, a_i \in A_i} (P_i(a'_i, a_i) - \mathbf{1}(a'_i,a_i)) \cdot x_i(a_i) & = 0 \ \forall \ x_i \in \Delta(A_i), \\
        \sum_{a_i \in A_i} (P_i(a'_i,a_i) - \mathbf{1}(a'_i,a_i)) \cdot x_i(a_i) & \geq 0 \ \forall \ x_i \in \Delta(A_i) \textnormal{ such that } x_i(a'_i) = 0.
    \end{align*}
    Moreover, in this case, the regret incurred by player $i$ against the strategy modifications $\phi$ is equal to the regret they incur against the vector field $f$ for any probability distribution $\sigma \in X = \times_{j \in N} \Delta(A_j)$, i.e.
    $$ \int_X d\sigma(x) \cdot (u_i(P_ix_i,x_{-i}) - u_i(x)) = \int_X d\sigma(x) \cdot \bilin{(P_i - \mathbf{1})x_i}{\nabla_i u_i(x)}.$$
\end{lemma}

\begin{proof}
    The first statement is immediate from $P_i$ being left-stochastic; each individual column of $P_i$ sums up to $1$, and the off-diagonal elements of $P_i - \mathbf{1}$ are non-negative. For the second statement, it is sufficient to show the integrands are pointwise equal. We note that
    \begin{align*}
        u_i(x) & = \sum_{a \in A} \left( \prod_{j \in N} x_j(a_j) \right) \cdot u_i(a), \\
        u_i(P_ix_i, x_{-i}) & = \sum_{a'_i \in A_i, a_{-i} \in A_{-i}} \left( \prod_{j \neq N} x_j(a_j) \right) \cdot \sum_{a_i \in A} P_i(a'_i,a_i) x_i(a_i) \cdot u_i(a'_i,a_{-i}).
    \end{align*}
    However, $\partial u_i(x) / \partial x_i(a'_i) = \sum_{a_{-i} \in A_{-i}} \left( \prod_{j \neq N} x_j(a_j) \right) \cdot u_i(a'_i,a_{-i})$. Thus, $
        u_i(P_ix_i,x_{-i}) = \bilin{P_ix_i}{\nabla_i u_i(x)}$ and $u_i(x) = \bilin{x_i}{\nabla_i u_i(x)}$.
\end{proof}

\begin{proposition*}[\ref{prop:linmax}]
 Let $N \geq 2, (A_i)_{i \in N}$ be fixed where $|A_j| > 2$ for at least two players $j$. Then for any linear transformation $\phi : \Delta(A_i) \rightarrow \Delta(A_i)$ in $\Phi_{LIN} \setminus \Phi_{SEMI}$, there exists utilities $(u_i)_{i \in N}$ such that for some learning cycle in the resulting game, player $i$ incurs positive regret against $\phi$.
\end{proposition*}

\begin{proof}
    Note that if all players have $\leq 2$ actions, then $\Phi_{LIN} = \Phi_{SEMI}$, because the notions of correlated and coarse correlated equilibria coincide and the set of semicoarse equilibria of a game is a subset of its coarse correlated equilibria. In particular, if $\phi : \Delta(A_i) \rightarrow \Delta(A_i)$ is in $\Phi_{LIN} \setminus \Phi_{SEMI}$, then player $i$ must have at least $3$ actions.

    Now, $\phi$ corresponds to a linear transformation $x_i \mapsto P_i x_i$, which implies that $x_i \mapsto P_i x_i - x_i$ is a tangential vector field over $\Delta(A_i)$, and it is sufficient to compute the regret incurred against this vector field. Moreover, if $\phi \notin \Phi_{SEMI}$, then by Lemma \ref{lem:triplet} there exists a triplet $a_i, a'_i, a''_i$ such that (\ref{eq:triplet}) does not hold. Thus, restrict attention to the case when 
    $$ K = (P_i(a_i,a'_i) + P_i(a'_i, a''_i) + P_i(a''_i,a_i)) - (P_i(a_i,a''_i) + P_i(a''_i,a'_i) + P_i(a'_i,a_i)) > 0,$$
    as the alternate case will follow from negating all players' utilities. By assumption, there exists a player $j$ whose strategy set also contains $\geq 3$ strategies. Let $a_j, a'_j, a''_j \in A_j$ be distinct, and consider the $3 \times 3$ rock-paper-scissors game:
    \begin{center}
        \begin{tabular}{ r|c|c|c| }
            \multicolumn{1}{r}{}
             &  \multicolumn{1}{c}{$a_j$}
             & \multicolumn{1}{c}{$a'_j$} & \multicolumn{1}{c}{$a''_j$} \\
            \cline{2-4}
            $a_i$ & 0,0 & 1,-1 & -1,1 \\
            \cline{2-4}
            $a'_i$ & -1,1 & 0,0 & 1,-1 \\
            \cline{2-4}
            $a''_i$ & 1,-1 & -1,1 & 0,0 \\
            \cline{2-4}
        \end{tabular}
    \end{center} 
    We will embed this game into a larger game defined for player $N$ and action sets $(A_k)_{k \in N}$. Consider the orthonormal basis of $\mathbb{R}^{A_i}$,
    $$v_1 = \begin{pmatrix}
        1/\sqrt{6} \\
        -\sqrt{2/3} \\
        1/\sqrt{6} \\
        0 \\
        \vdots \\
        0
    \end{pmatrix}, v_2 = \begin{pmatrix}
        1/\sqrt{2} \\
        0 \\
        -1/\sqrt{2} \\
        0 \\
        \vdots \\
        0
    \end{pmatrix}, v_3 = \begin{pmatrix}
        1/\sqrt{3} \\
        1/\sqrt{3} \\
        1/\sqrt{3} \\
        0 \\
        \vdots \\
        0
    \end{pmatrix}, \textnormal{ and } v_k = e_k \textnormal{ for } k > 3,$$
    where the actions are reordered such that first three coordinates correspond to the actions $a_i, a'_i, a''_i$. Analogously define a basis $w_1, w_2, w_3, \ldots, w_{|A_j|}$ over $\mathbb{R}^{A_j}$. We will fix utilities,
    \begin{align*}
        u_i(x) & = -\sqrt{3} \cdot x_i^T ( v_1 w_2^T - v_2 w_1^T) x_j, \\
        u_j(x) & = -u_i(x) \\
        u_k(x) & = 0 \ \forall \ k \neq i,j,
    \end{align*}
    which implies that players $i,j$ have the $3 \times 3$ rock-paper-scissors game along the specified triplets of actions, and have $0$ utility otherwise. Meanwhile, players $k \neq i,j$ are non-strategic, with $0$ utility over all outcomes. Now, for small $\epsilon > 0$, fix initial conditions 
    \begin{align*}
        x_i(0) & = v_3/\sqrt{3} + \epsilon \cdot v_1, \\
        x_j(0) & = w_3/\sqrt{3} + \epsilon \cdot w_1, 
    \end{align*}
    and $x_k(0)$ chosen arbitrarily for $k \neq i,j$. Then 
    \begin{align*}
    x_i(t) & = v_3 / \sqrt{3} + \epsilon \cdot (\cos(t\sqrt{3}) \cdot v_1 + \sin(t\sqrt{3}) \cdot v_2) \\ 
    x_j(t) & = w_3 / \sqrt{3} + \epsilon \cdot (\cos(t\sqrt{3}) \cdot w_1 + \sin(t\sqrt{3}) \cdot w_2) \\
    x_k(t) & = x_k(0) \ \forall \ k \neq i,j,
    \end{align*}
    is a solution to the gradient dynamics of the game, that is, $\frac{dx_k(t)}{dt} = \nabla_k u_k(x(t))$ for every player $k \in N$. Moreover, for player $i$, the regret against the strategy modification $\phi$ is given by
    \begin{align*}
        & \quad\lim_{T \rightarrow \infty} \frac{1}{T} \int_0^T dt \cdot x_i(t)^T(P-\mathbf{1})^T\nabla_i u_i(x(t)) \\ 
        & = \frac{\sqrt{3}}{2\pi} \int_0^{2\pi/\sqrt{3}} dt \cdot \textnormal{tr}[(P-\mathbf{1})^T(\nabla_i u_i(x(t)) x_i^T)] \\
        & = \frac{\sqrt{3}}{2\pi} \int_0^{2\pi/\sqrt{3}} dt \cdot \textnormal{tr}[(x_i(t) \nabla_i u_i(x(t))^T)(P-\mathbf{1})] \\
        & = \frac{\sqrt{3}}{2\pi} \int_0^{2\pi/\sqrt{3}} dt \cdot \textnormal{tr}\left[\epsilon^2 \sqrt{3} \cdot \left(\cos^2(t\sqrt{3}) \cdot v_1 v_2^T - \sin^2(t\sqrt{3}) \cdot v_2 v_1^T \right) (P - \mathbf{1}) \right] \\
        & = \frac{\epsilon^2 \sqrt{3}}{2} \cdot \textnormal{tr}[(v_1 v_2^T - v_2 v_1^T )(P - \mathbf{1})].
    \end{align*}
    Here, the first inequality is by noting that the limit distribution is uniform on the cyclic trajectory, and by the usual trick of taking the trace of a scalar. The second inequality follows from the fact that the trace is cyclic and $\textnormal{tr}(A) = \textnormal{tr}(A^T)$, and the third inequality is via factoring out the cyclic terms whose integral evaluates to $0$. Finally, we have 
    \begin{align*}
       \sqrt{3} \cdot \textnormal{tr}[(v_1v_2^T - v_2 v_1^T)(P - \mathbf{1})] 
       = \textnormal{tr}\left[ \begin{pmatrix}
        0 & 1 & -1 \\
        -1 & 0 & 1 \\
        1 & -1 & 0 \\
    \end{pmatrix} \begin{pmatrix}
        P_i(a_i,a_i) & P_i(a_i,a'_i) & P_i(a_i,a''_i) \\
        P_i(a'_i,a_i) & P_i(a'_i,a'_i) & P_i(a'_i,a''_i) \\
        P_i(a''_i,a_i) & P_i(a''_i,a'_i) & P_i(a''_i,a''_i) \\
    \end{pmatrix} \right] = K,
    \end{align*}
    which implies that player $i$ incurs regret $\epsilon^2 K / 2$ against $\phi$.
\end{proof}

\subsubsection{Polyhedral Representations of Semicoarse Equilibria}\label{sec:representation-proofs}

\begin{theorem*}[\ref{thm:strategy-def}]
    The set of inequalities (\ref{con:semicoarse}) are equivalent to the inequalities
    \begin{equation}
        \sum_{a \in A} \sigma(a) \cdot \left[ \left( \sum_{a'_i \in A_i} P_i(a'_i,a_i) 
 \cdot u_i(a'_i,a_{-i})\right) - u_i(a) \right] \leq 0 \ \forall \ P_i \in \mathbf{P}(A_i),
    \end{equation}
    where $\mathbf{P}(A_i)$ is the set of left stochastic matrices corresponding to the strategy modifications:
    \begin{enumerate}
        \item For a proper subset $S_i \subset A_i$, each action $a_i \in S_i$ is transformed to the uniform distribution on $A_i \setminus S_i$, whereas for $a_i \in S_i \setminus A_i$, $a_i \mapsto a_i$.
        \item For a cycle of actions $C_i = \{a_{i1}, a_{i2}, \ldots, a_{ik}\}$, writing $a_{i(k+1)} = a_{i1}$, each action $a_{i\ell}$ is transformed to the uniform distribution on $\{a_{i(\ell-1)},a_{i(\ell+1)}\}$, whereas for $a_i \in A_i \setminus C_i$, $a_i \mapsto a_i$.
    \end{enumerate}
\end{theorem*}

\begin{proof}
    As mentioned, it is sufficient to characterise the rays of the cone defined via the inequalities (\ref{con:preserve}), (\ref{con:tangency}), and (\ref{con:symmetry}); any pair $(Q_i,q_i)$ will decompose into a conical combination of such rays, and each ray corresponds via Proposition \ref{prop:linsubset} to a specific strategy modification. Let $(Q_i,q_i)$ be a pair which satisfies the tangency conditions (\ref{con:preserve}), (\ref{con:tangency}) and the symmetry condition (\ref{con:symmetry}). Moreover, denote $A_i = \{1,2,..,m\}$, and up to reordering the actions of player $i$, suppose that $q_i$ is in decreasing order; $q_i(a'_i) \geq q_i(a_i)$ for any $a'_i \leq a_i$ in $A_i$. Moreover, suggested by the proof of Lemma \ref{lem:triplet}, we may fix $q_i(m) = 0$ without loss of generality. Writing $\Sigma$ for the all-ones vector of appropriate dimension, if the pair $(Q_i,q_i)$ satisfies the tangency and the symmetry constraints, then for any $c \in \mathbb{R}$, so too does the pair $(Q_i - c \Sigma \Sigma^T, q_i + c \Sigma)$; moreover, for any $x_i \in \Delta(A_i)$
    \begin{align*}
        (Q_i - c\Sigma \Sigma^T) x_i + q_i + c\Sigma & = Q_i x_i - c \Sigma (\Sigma^T x_i) + q_i + c\Sigma \\
        & = Q_i x_i - c\Sigma + q_i + c\Sigma \\
        & = Q_i x_i + q_i,
    \end{align*}
    so the pairs $(Q_i,q_i)$, $(Q_i - c \Sigma \Sigma^T, q_i + c \Sigma)$ generate the same (tangential) gradient field over $\Delta(A_i)$. We note that this implies $\pm(-\Sigma\Sigma^T,\Sigma)$ are trivial rays, corresponding to the identity on $\Delta(A_i)$.

    Now, let $s$ be the maximum positive valued index of $q_i$, and consider the pair $(R^s_i,r^s_i)$ defined as
    $$ r^s_i(a'_i) = \begin{cases}
        1 & a'_i \leq s \\
        0 & a'_i > s
    \end{cases}, R^s_i(a'_i,a_i) = \begin{cases}
        -1 & a'_i,a_i \leq s \\
        -s & s < a'_i = a_i \\
        0 & \textnormal{otherwise.}
    \end{cases}$$
    We claim that $R^s_i$ attains the pointwise minimum possible bound on its off-diagonal elements. Checking the tangency constraints, we observe that if $a_i \leq s$, then $R^s_i(a'_i,a_i) + r^s_i(a'_i) = 0$ for any $a'_i$. In turn, if $a_i > s$, then
    $$\sum_{a'_i} R^s_i(a'_i,a_i) + r^s_i(a'_i) = s \cdot 1 - s = 0.$$
    Therefore, the constraint (\ref{con:preserve}) holds. Meanwhile, the constraint (\ref{con:tangency}) necessitates
    $$ R^s_i(a'_i,a_i) + r^s_i(a'_i) \geq 0 \ \forall \ a'_i, a_i.$$
    Such an off-diagonal element is non-zero only when $a'_i \leq s, a_i > s$, in which case the constraint does not bind. However, by the symmetry constraint, $R^s_i(a'_i,a_i) = R^s_i(a_i,a'_i)$, and we have 
    $$R^s_i(a_i,a'_i) + r^s_i(a_i) = 0.$$
    Thus, accounting for the symmetry constraint, each off-diagonal entry of $R^s_i$ is set minimally. We thus turn our attention to the pair $(Q'_i,q'_i) = (Q_i - q_i(s) \cdot R^s_i,q_i - q_i(s) \cdot  r^s_i)$. By construction the pair satisfies the constraints (\ref{con:preserve}) and (\ref{con:tangency}), since they hold for the pairs $(Q_i,q_i)$ and $(R^s_i,r^s_i)$ with equality. Checking the constraint (\ref{con:tangency}), we require
    $$ Q_i(a'_i,a_i) + q_i(a'_i) \geq q_i(s) \cdot (R^s_i(a'_i,a_i) - r^s_i(a'_i))$$
    for any $a'_i \neq a_i$. Unless $a'_i \leq s$ and $a_i > s$, the RHS equals $0$ and the inequality holds since the pair $(Q_i,q_i)$ already satisfies the tangency constraints. In the former case, we thus require
    $$ Q_i(a'_i,a_i) + q_i(a'_i) \geq q_i(s).$$
    However, we know that 
    $$ Q_i(a_i,a'_i) + q_i(a_i) \geq 0,$$
    and by the choice of $s$ and the ordering of $q_i$, $q_i(a_i) = 0$. Therefore, $Q_i(a'_i,a_i) = Q_i(a_i,a'_i) \geq 0$. Meanwhile, since $a'_i \leq s$, $q_i(a'_i) \geq q_i(s)$. Therefore, the pair $(Q'_i, q'_i)$ indeed satisfies the tangency constraints. Moreover, $q'_i$ is in descending order, with either $q'_i = 0$ identically, or with a minimum non-zero index $s' < s$.

    As a consequence, we may iteratively decompose the pair $(Q_i,q_i)$ as 
    $$ (Q_i,q_i) = (Q''_i,0) + \sum_{s = 1}^{m-1} \lambda_s \cdot (R^s_i,r^s_i) $$
    for some non-negative numbers $\lambda_s$. Invoking Proposition \ref{prop:linsubset} with $\delta = 1/s$, the pair $(R^s_i, r^s_i)$ corresponds to the linear transformation $\Delta(A_i) \rightarrow \Delta(A_i)$ defined via the matrix $P^s_i$,
    $$ P^s_i(a'_i,a_i) = \begin{cases}
        \frac{1}{s} & a'_i \leq s, a_i > s \\
        1 & a'_i = a_i \leq s \\
        0 & \textnormal{otherwise,}
    \end{cases} $$
    which covers the case (1). In turn, again invoking Proposition \ref{prop:linsubset}, $P''_i = \mathbf{1} - \delta'' \cdot Q''_i$ is a symmetric left-stochastic (and hence doubly stochastic) matrix for some $\delta'' > 0$. By \cite{cruse1975note}, $P''_i$ is a convex combination 
    $$ P''_i = \sum_{P \in I} \frac{\lambda_P}{2} \cdot (P + P^T)$$
    of a set $I$ of $m \times m$ permutation matrices $P$. Therefore, 
    $$ Q''_i = \sum_{P \in I} \frac{\lambda_P}{2\delta''} \cdot (P + P^T - 2\mathbf{1}).$$
    We recall that each permutation $P$ can be decomposed into disjoint directed cycles. In this case, $P + P^T = \sum_{C \textnormal{ a cycle of } P} C + C^T$ for the directed graph of each cycle of $P$. Writing $\mathbf{1}_C$ as the diagonal matrix which equals the identity on the elements of $C$ and $0$ otherwise, we thus have
    $$ Q''_i = \sum_{P \in I} \sum_{C \textnormal{ a cycle of } P} \frac{\lambda_P}{2\delta} \cdot (C + C^T - 2\mathbf{1}_C).$$
    For each cycle $C$, the matrix $C + C^T - 2\mathbf{1}_C$ is readily seen to satisfy the tangency and symmetry constraints. Invoking Proposition \ref{prop:linsubset} once more with $\delta = 1/2$, each matrix $C + C^T - 2\mathbf{1}_C$ corresponds to the linear transformation $\Delta(A_i) \rightarrow \Delta(A_i)$, defined by the matrix $P_i^C$ where
    $$ P_i^C(a'_i,a_i) = \begin{cases}
        \frac{1}{2} & (a'_i,a_i) \textnormal{ or } (a_i,a'_i) \textnormal{ is a directed edge of }C, \\
        -1 & a'_i = a_i \in C, \\
        0 & \textnormal{otherwise.}
    \end{cases} $$
    This, of course, covers the case (2).
\end{proof}

\begin{theorem*}[\ref{thm:short-extension}]
    For a normal-form game $\Gamma$, suppose that $(\sigma,\omega,\rho)$ satisfies the linear (in)equalities:
    \begin{align*}
        \sum_{a \in A} \sigma(a) & = 1, \\
        \forall \ i \in N, \ a_i \neq a'_i, \gamma_i(a'_i,a_i) + \rho_i(a_i,a'_i) - \rho_i(a'_i,a_i) + \sum_{a_{-i} \in A_{-i}} \sigma(a) \cdot \left( u_i(a'_i,a_{-i}) - u_i(a) \right) & = 0 \\
        \forall i \in N, \forall a'_i \in A_i, \sum_{a_i \neq a'_i} \gamma_i(a'_i,a_i) + \sum_{a \in A} \sigma(a) \cdot \left( u_i(a'_i,a_{-i}) - u_i(a) \right) & = 0 \\
        \sigma,(\gamma_i)_{i \in N} & \geq 0.
    \end{align*}
    Then $\sigma$ is a semicoarse correlated equilibrium of $\Gamma$. As a consequence, the set of semicoarse correlated equilibrium of a normal-form game can be represented via an extension which is polynomial in the size of the normal-form game, with $|A| + \sum_{i \in N} \frac{3}{2}|A_i|(|A_i|-1)$ variables and $1 + \sum_{i \in N} |A_i|^2$ constraints.
\end{theorem*}

\begin{proof}
    Suppose that $\sigma$ is a semicoarse equilibrium. Then $\sigma$ is a solution of value $0$ for the feasibility problem,
    \begin{align}
        \min_{\sigma \geq 0, \sum_{a \in A} \sigma(a) = 1} 0 + \sum_{i \in N} \max_{(Q_i,q_i)} \sum_{a \in A} \sigma(a) \cdot \sum_{a'_i \in A_i} (Q_i(a'_i,a_i) & + q_i(a'_i)) \cdot u_i(a'_i,a_{-i}) \textnormal{ subject to } \label{opt:feas} \\
        \sum_{a'_i \in A_i} Q_i(a'_i,a_i) + q_i(a'_i) & = 0 \ \forall \ a_i \in A_i \tag{$\mu_i(a_i)$}\\
        -Q_i(a'_i,a_i) - q_i(a'_i) & \leq 0 \ \forall \ a_i, a'_i \in A_i, a_i \neq a'_i \tag{$\gamma_i(a'_i,a_i)$}\\
        Q_i(a'_i,a_i) - Q_i(a_i,a'_i) & = 0 \ \forall \ a_i, a'_i \in A_i, a_i \neq a'_i. \tag{$\rho_i(a'_i,a_i)$}
    \end{align}
    We will proceed by taking the LP dual of a given inner maximisation problem for each player $i$. The dual LP is given,
    \begin{align}
        \min_{\gamma_i \geq 0, \mu_i, \rho_i} 0 \textnormal{ subject to } & \\
        \mu_i(a_i) - \gamma_i(a'_i,a_i) + \rho_i(a'_i,a_i) - \rho_i(a_i,a'_i) & = \sum_{a_{-i} \in A_{-i}} \sigma(a) \cdot u_i(a'_i,a_{-i}) \tag{$Q_i(a'_i,a_i), a'_i \neq a_i$} \\
        \mu_i(a_i) & = \sum_{a_{-i} \in A_{-i}} \sigma(a) \cdot u_i(a) \tag{$Q_i(a_i,a_i)$} \label{cons:subs}\\
        \sum_{a_i \in A_i} \mu_i(a_i) - \mathbb{I}[a_i \neq a'_i] \gamma_i(a'_i,a_i) & = \sum_{a \in A} \sigma(a) \cdot u_i(a'_i,a_{-i}). \tag{$q_i(a'_i)$}
    \end{align}
    The constraints (\ref{cons:subs}) allow us to fix $\mu_i(a_i)$. After rearranging the constraints, we are left with 
    \begin{align*}
        \forall \ a_i \neq a'_i, \gamma_i(a'_i,a_i) + \rho_i(a_i,a'_i) - \rho_i(a'_i,a_i) + \sum_{a_{-i} \in A_{-i}} \sigma(a) \cdot \left( u_i(a'_i,a_{-i}) - u_i(a) \right) & = 0 \\
        \forall a'_i \in A_i, \sum_{a_i \neq a'_i} \gamma_i(a'_i,a_i) + \sum_{a \in A} \sigma(a) \cdot \left( u_i(a'_i,a_{-i}) - u_i(a) \right) & = 0.
    \end{align*}
    Therefore, the feasibility problem (\ref{opt:feas}) may be written as 
    \begin{align*}
        \min_{\sigma, (\gamma_i,\rho_i)_{i \in N}} 0 \textnormal{ subject to } & \\
        \sum_{a \in A} \sigma(a) & = 1 \\
        \forall \ i \in N, \ a_i \neq a'_i, \gamma_i(a'_i,a_i) + \rho_i(a_i,a'_i) - \rho_i(a'_i,a_i) + \sum_{a_{-i} \in A_{-i}} \sigma(a) \cdot \left( u_i(a'_i,a_{-i}) - u_i(a) \right) & = 0 \\
        \forall i \in N, \forall a'_i \in A_i, \sum_{a_i \neq a'_i} \gamma_i(a'_i,a_i) + \sum_{a \in A} \sigma(a) \cdot \left( u_i(a'_i,a_{-i}) - u_i(a) \right) & = 0 \\
        \sigma,(\gamma_i)_{i \in N} & \geq 0.
    \end{align*}
    For each player $i$, there are $|A_i|(|A_i|-1)$ variables $\gamma_i(a'_i,a_i)$. Meanwhile, we actually only require $\binom{|A_i|}{2}$ of the variables $\rho_i(a'_i,a_i)$, as we may fix $\rho_i(a'_i,a_i) = -\rho_i(a_i,a'_i)$ for any distinct pair of actions. Finally, each player $i$ has $|A_i|^2$ ``equilibrium constraints''.
\end{proof}

\subsection{Semicoarse Equilibrium of (Simple) Auction Games}

\subsubsection{Convergence to Equilibrium in Bertrand Competition}\label{sec:bertrand-proof}

We start our discussion here by first pointing out the Nash equilibria of interest for Bertrand competition, as in \cite{deng2022nash} (Proposition 3). Similarly to their paper, the proof is omitted, since the observation is immediate.

\begin{lemma}
    In our discretised Bertrand competition setting, if $(p-c/n)D(p)$ is strictly concave and there are $m \geq 2$ firms of lowest marginal cost $c/n$, the only pure Nash equilibria of the game have:
    \begin{enumerate}
        \item At least two firms $i,j \leq m$ post price $c/n$, and all other firms post prices $\geq c/n$, with the equality strict for firms $> m$.
        \item All firms of marginal cost $c/n$ post price $(c+1)/n$, and firms $> m$ post prices $\geq (c+1)n$, with the equality strict whenever firm $i > m$ has marginal cost $> (c+1)/n$.
    \end{enumerate}
    These are also all the pure strategy Nash equilibria if $(p-c/n)D(p)$ is weakly concave, and $m \geq 3$. If instead $m = 2$, $D((c+1)/n) = D((c+2)/n)$, and there are no firms of marginal price $(c+1)/n$, there are also pure strategy Nash equilibria where both firms of lowest marginal price bid $(c+2/n)$ and all other firms bid $> (c+2/n)$. 
\end{lemma}

\begin{theorem*}[\ref{thm:bertrand-unique}]
    Suppose there exists at least two firms with minimum marginal cost $c/n \in \{0,1/n,...,1-1/n\}$, and that $(p - c/n) D(p)$ is strictly concave over $[c/n,1]$. Then in every semicoarse equilibrium $\sigma$ of the discretised Bertrand competition game with bidding sets $\{0,1/n,...,1\}$, each outcome $(p_i)_{i \in N}$ assigned positive probability $\sigma(p) > 0$ is a pure strategy Nash equilibrium of the game, with at least two firms of minimum marginal cost posting prices in $c/n + \{0,1/n\}$.
\end{theorem*}

\begin{proof}
    Up to a reordering of the firms, suppose that $m \equiv \{1,2,...,m\}$ is the set of firms with minimum marginal cost, $c/n$. For each set $B^k = c/n + \{1/n,2/n,...,k/n\} \subseteq A_i$, denote by $\phi_k : \Delta(A_i) \rightarrow \Delta(A_i)$ the strategy modification induced by $S_i = A_i \setminus B^k$ in Theorem \ref{thm:strategy-def}.(1). Thus, given a price $p_i$, if $p_i \leq c/n$ or $p_i > (c+k)/n$, then $p_i$ is transformed to the uniform distribution on $c/n + \{1/n,2/n,...,k/n\}$; else, $p_i$ is unchanged. We notice immediately that for $p_i \leq c/n$, this deviation is weakly utility improving for $i$.
    
    Moreover, by our order on the firms, each firm $i > m$ has marginal cost $c_i/n > c/n$. We then denote by $\phi^i$ the strategy modification induced by $S_i = \{0,1/n,...,c_i/n-1/n\}$ Theorem \ref{thm:strategy-def}.(1). Each of these transformations has firm $i$ deviate from posting a price $p_i$ strictly below their cost to posting a price distributed uniformly on $\{c_i/n, (c_i+1/n),...,1\}$, else their posted price is not modified. Once again, we note that such a strategy modification is weakly utility improving for firm $i$, and strictly so whenever firm $i$ has posted a minimum price amongst all firms.
    
    Now, for any price vector $(p_i)_{i \in N}$, define 
    $$ d(p) = \begin{cases}
        1 & p \textnormal{ is a pure strategy Nash equilibrium,} \\
        0 & \textnormal{otherwise.}
    \end{cases}$$
    Denoting $(Q^k,q^k)$ as the pair corresponding to the strategy modification $\phi_k$, and $(R^i,r^i)$ as the strategy modification corresponding to $\phi^i$, we will want to find multipliers $\epsilon_k \geq 0$, and $(\delta_i)_{i > m}$, such that for any $p \in A$,
    \begin{equation}\label{eq:pointwise}
        \sum_{i \leq m} \sum_{k} \epsilon_k \cdot \sum_{p'_i \in A_i} (Q^k(p'_i,p_i) + q^k(p'_i)) \cdot u_i(p'_i,p_{-i}) + \sum_{i > m} \delta_i \cdot \sum_{p'_i \in A_i} (R^i(p'_i,p_i)+r^i(p'_i)) \cdot u_i(p'_i,p_{-i}) \geq d(p).
    \end{equation}
    This follows from setting $(Q_i,q_i) = \sum_{k} \epsilon_k (Q^k,q^k)$ for firms $i \leq m$ and $(Q_i,q_i) = (R^i,r^i)$ for firms $i > m$ in the dual Lyapunov function estimation problem (\ref{opt:lyapunov-LP}). Since $d$ is an indicator function, it is sufficient to find $\epsilon, \delta$ such that the LHS is non-negative for any $p$, and strictly positive whenever $d(p) > 0$. Indeed, after finding such a solution, if necessary we may scale it by a positive constant to ensure (\ref{eq:pointwise}) holds.

    We proceed by case analysis on $p$. First suppose that the minimum posted price $\bar{p}$ is strictly less than $c/n$. In this case, all winning firms have negative utility. Meanwhile, all deviations considered have utility $\geq 0$ against any price vector $p$, since no firm would deviate to bidding less than their marginal cost. As a consequence, as long as $\sum_k \epsilon_k > 0$ and for each $i > m$, $\delta_i > 0$, the LHS of (\ref{eq:pointwise}) is strictly positive.

    Now suppose that the minimum posted price equals $c/n$. If for some firm $i > m$, $p_i = c /n < c_i / n$, then the term $\delta_i \cdot \sum_{p'_i \in A_i} (R^i(p'_i,p_i)+r^i(p'_i)) \cdot u_i(p'_i,p_{-i})$ is strictly positive. So we restrict attention to the case when firms $i > m$ bid $p_i > c/n$. In this case, by the assumption on the minimum posted price, for $i \leq m$, $p_i = c/n$. Then, if $d(p) = 1$ then firm $i$ is the unique firm with minimum posted price. In this case, each deviation $\phi_k$ places probability $1/k$ onto the price $c/n + 1/n$, which is $\leq p_j$ for $j \neq i$. Therefore, if $\sum_k \epsilon_k > 0$, then the LHS is positive. Finally, if there are two firms $i,j \leq m$, then all deviations considered have expected utility $0$; but in this case, $d(p) = 0$ and (\ref{eq:pointwise}) nevertheless holds.

    Next, suppose that the minimum posted price equals $c/n + 1/n$. Here, all deviations considered still guarantee non-negative utility change; if $p_i = c/n + 1/n$, then if $i \leq m$ none of the strategy modifications $\phi_k$ modify the price of firm $i$. In turn, if $p_i > c/n + 1/n$ for firm $i$, the deviation $\phi_k$ places probability $1/k$ on $c/n + 1/n$, resulting in a positive utility change. Likewise, if $p_i = c/n + 1/n$ for $i > m$, then firm $i$ incurs a non-negative utility change under the strategy modification $\phi^i$, and strictly so whenever $c_i > (c+1)$. The only case in which all deviations considered have $0$ utility change is if all firms $i \leq m$ bid $c/n + 1/n$, and no firm $i > m$ with $c_i > c+1$ bids $(c+1)/n$, in which case $f(p) = 0$.

    Thus, the interesting case is when the minimum posted price equals $\bar{p} = c/n + \ell/n$ where $\ell > 1$. Suppose that each buyer $i$ has $p_i = \ell_i/n$. Here, since the strategy modification $\phi^i$ is weakly utility improving for each firm $i \geq m$, we observe that 
    \begin{align}
        & \quad \sum_{i \leq m} \sum_{k} \epsilon_k \cdot \sum_{p'_i \in A_i} (Q^k(p'_i,p_i) + q^k(p'_i)) \cdot u_i(p'_i,p_{-i}) + \sum_{i > m} \delta_i \cdot \sum_{p'_i \in A_i} (R^i(p'_i,p_i)+r^i(p'_i)) \cdot u_i(p'_i,p_{-i}) \nonumber \\
        & \geq \sum_{i \leq m} \sum_{k} \epsilon_k \cdot \sum_{p'_i \in A_i} (Q^k(p'_i,p_i) + q^k(p'_i)) \cdot u_i(p'_i,p_{-i}) \nonumber \\
        & = \sum_{i \leq m} \sum_{k = 1}^{\ell-1} \epsilon_k \cdot \left(\frac{1}{k} \sum_{k' = 1}^k \left(\frac{ k'}{n}\right) D\left( \frac{c + k'}{n} \right) - \frac{\mathbb{I}[\ell_i = \ell]}{|\arg\min_j p_j|} \left( \frac{ \ell}{n} \right) D \left( \frac{c + \ell}{n} \right)\right) \nonumber \\ \ldots & + \sum_{k = \ell}^{\ell_i - 1} \epsilon_k \cdot \left(\frac{1}{k} \sum_{k' = 1}^k \left(\frac{k'}{n}\right) D\left( \frac{c + k'}{n} \right) + \frac{1}{k(|\arg\min_j p_j|+1)} \left(\frac{\ell}{n}\right) D\left( \frac{c + \ell}{n} \right) \right) \nonumber \\
        & \geq \sum_{k = 1}^{\ell-1} \epsilon_k \cdot \left(\frac{m}{k} \sum_{k' = 1}^k \left(\frac{k'}{n}\right) D\left( \frac{c+k'}{n} \right) - \left( \frac{\ell}{n} \right) D \left( \frac{c+\ell}{n} \right)\right). \label{eq:tie-case}
    \end{align}
    Here, the first inequality is since the action transformations $\phi^i$ do not modify $p_i \geq c_i/n$ for $i > m$. The second inequality, in turn, follows from three reasons. First, for any bid $p_i > \bar{p}$ and any firm $i \leq m$, the deviation to the uniform distribution on $B_k$ is strictly utility improving; firm $i$ goes from winning with probability $0$ to winning with positive probability over posted prices which incur strictly positive utility. Second, firms $i \leq m$ which post price $\bar{p}$ in sum obtain utility $\leq (\bar{p}-c/n)D(\bar{p})$, accounting for the possibility that some firms $j > m$ have also posted prices $\bar{p}$.
    
    Note that (\ref{eq:tie-case}) holds only with equality if all firms $i \leq m$ post prices $\bar{p}$, whereas other firms $i > m$ post prices $> \bar{p}$. We proceed by iteratively fixing $\epsilon_\ell$ via induction on $\ell$. If $\ell = 2$, then (\ref{eq:tie-case}) equals
    \begin{equation}\label{eq:base-case} \epsilon_1 \cdot \left( m \cdot \left(\frac{1}{n}\right) D\left( \frac{c+1}{n} \right) - \left( \frac{2}{n} \right) D \left( \frac{c+2}{n} \right)\right).\end{equation}
    By the strict concavity assumption on $(p-c/n)D(p)$, 
    $$ \frac{1}{2} \cdot 0 \cdot D\left(\frac{c}{n} \right) + \frac{1}{2} \cdot \left(\frac{2}{n}\right) D\left( \frac{c+2}{n} \right) < \left( \frac{1}{n} \right) D \left( \frac{c+1}{n} \right),$$
    and therefore the term (\ref{eq:base-case}) is positive for any choice of $\epsilon_1 > 0$.

    We therefore want to determine $\epsilon_{\ell}$, given $\sum_{k' = 1}^{\ell-1} \epsilon_k > 0$. Towards this end, consider the change in (\ref{eq:tie-case}) when we let $\ell \mapsto \ell+1$, which equals
    \begin{align}
        & \quad \sum_{k = 1}^{\ell} \epsilon_k \cdot \left(\frac{m}{k} \sum_{k' = 1}^k \left(\frac{k'}{n}\right) D\left( \frac{c+k'}{n} \right) - \left( \frac{\ell+1}{n} \right) D \left( \frac{c+\ell+1}{n} \right)\right) \nonumber \\ &  - \sum_{k = 1}^{\ell-1} \epsilon_k \cdot \left(\frac{m}{k} \sum_{k' = 1}^k \left(\frac{k'}{n}\right) D\left( \frac{c+k'}{n} \right) - \left( \frac{\ell}{n} \right) D \left( \frac{c+\ell}{n} \right)\right) \nonumber \\
        & = -\left( \sum_{k = 1}^{\ell-1} \epsilon_k \right) \cdot \left( \left( \frac{\ell+1}{n} \right) D \left( \frac{c+\ell+1}{n} \right) - \left( \frac{\ell}{n} \right) D \left( \frac{c+\ell}{n} \right) \right) \nonumber \\
        & + \epsilon_{\ell} \cdot \left( \frac{m}{\ell} \sum_{k' = 1}^\ell \left(\frac{k'}{n}\right) D\left( \frac{c+k'}{n} \right) - \left( \frac{\ell+1}{n} \right) D \left( \frac{c+\ell+1}{n} \right) \right).        
    \end{align}
    In this case, the coefficient of $\epsilon_\ell$ is again always strictly positive due to the strict concavity of $(p-c/n)D(p)$ and Jensen's inequality, as the uniform distribution over $c/n + \{1/n,2/n,...,\ell/n\}$ dominates the uniform probability distribution over $c/n + \{0,(\ell+1)/n\}$ in second-order.

    Meanwhile, the coefficient of $-(\sum_{k = 1}^{\ell-1} \epsilon_k)$ is strictly positive only when $(c+\ell)/n$ is less than the \emph{monopoly price} $\ell^*/n$, where $\ell^* = \arg\max_{0 \leq k \leq n, k \in \mathbb{N}} (k-c)/n \cdot D(k/n)$. Therefore, if $\ell \geq \ell^*$, we let $\epsilon_\ell = 0$. Otherwise, we let
    $$ \epsilon_\ell = \frac{\left( \sum_{k = 1}^{\ell-1} \epsilon_k \right) \cdot \left( \left( \frac{\ell+1}{n} \right) D \left( \frac{c+\ell+1}{n} \right) - \left( \frac{\ell}{n} \right) D \left( \frac{c+\ell}{n} \right) \right)}{\frac{m}{\ell} \sum_{k' = 1}^\ell \left(\frac{k'}{n}\right) D\left( \frac{c+k'}{n} \right) - \left( \frac{\ell+1}{n} \right) D \left( \frac{c+\ell+1}{n} \right)} > 0.$$
\end{proof}

\subsection{Weighted Action Sets \& Semicoarse Equilibria Under Linear Transformations}\label{sec:scale-proofs}

\begin{theorem*}[\ref{thm:weight-constraint}]
    Suppose that $\sigma^w \in \Delta(A^w)$ is a semicoarse equilibrium of the weighted normal-form game $\Gamma^w$. Then the induced distribution $\sigma(a) = \sum_{a^w \in \kappa(a)} \sigma^w(a^w)$ on $A$ satisfies the equilibrium constraints
    \begin{equation*}
        \sum_{a \in A} \sigma(a) \cdot \left[ \left( \sum_{a'_i \in A_i} P_i(a'_i,a_i) 
 \cdot u_i(a'_i,a_{-i})\right) - u_i(a) \right] \leq 0 \ \forall \ P_i \in \mathbf{P}(A_i),
    \end{equation*}
    where $\mathbf{P}(A_i)$ is the set of left stochastic matrices corresponding to the strategy modifications:
    \begin{enumerate}
        \item For a proper subset $S_i \subset A_i$, each action $a_i \in S_i$ is transformed to the distribution $F$ on $A_i \setminus S_i$, such that $F(a'_i) \propto w_i(a'_i)$ for each $a'_i \in A_i \setminus S_i$. Meanwhile, for $a_i \in S_i \setminus A_i$, $a_i \mapsto a_i$.
        \item For a cycle of actions $C_i = \{a_{i1}, a_{i2}, \ldots, a_{ik}\}$, writing $a_{i(k+1)} = a_{i1}$ and $\delta = \min_{\ell} w_i(a_{i\ell})$, each action $a_{i\ell}$ is transformed to the distribution on $\{a_{i(\ell-1)},a_{i\ell}, a_{i(\ell+1)}\}$, where $a_{i\ell}$ has probability $(1-\delta/w_i(a_{i\ell}))$ and the actions $\{a_{i(\ell-1)},a_{i(\ell+1)}\}$ are equally likely. Whereas for $a_i \in A_i \setminus C_i$, $a_i \mapsto a_i$.
    \end{enumerate}
\end{theorem*}

\begin{proof}
    We want to determine the inequalities satisfied by $\sigma : A \rightarrow \mathbb{R}_+$, such that there exists $\sigma^w$ which satisfies the (in)equalities 
    \begin{align}
        \sigma(a) - \sum_{a^w \ | \ \lambda(a^w) = a} \sigma^w(a^w) & = 0 \label{con:ext_1}\\
        \forall \ i \in N, \forall \textnormal{ suitable }(Q^w_i,q^w_i), & \nonumber \\ 
        \sum_{a^w \in A^w} \sigma^w(a^w) \cdot \sum_{a'^w \in A^w_i} \left( Q^w_i(a'^w_i,a^w_i) + q^w_i(a'^w_i) \right) \cdot u^w_i(a'^w_i,a^w_{-i}) & \leq 0 \label{con:ext_2}\\
        \sum_{a^w} \sigma^w(a^w) & = 1 \label{con:ext_3}\\
        \sigma^w & \geq 0, \label{con:ext_4}
    \end{align}
    where the set of suitable pairs $(Q_i^w,q_i^w)$ are those which satisfy (\ref{con:preserve}), (\ref{con:tangency}), (\ref{con:symmetry}) in the $w$-weighted game. Now, two polyhedra $P,P'$ are equal if and only if for every vector $c$, $\max_{x \in P} c^T x = \max_{x \in P'} c^T x$. Moreover, for any inequality $\sum_a \sigma(a) \cdot \tilde{h}(a) \leq \tilde{b}$ which is valid and also tight for at least one such $\sigma$, $\tilde{b}$ equals the value of the LP, 
    \begin{align}
        \max_{a \in A} \sum_a \sigma(a) \cdot \tilde{h}(a) \textnormal{ subject to (\ref{con:ext_1}-\ref{con:ext_4}).} \label{opt:w-primal}
    \end{align}
    We may of course simply substitute in $\sigma(a)$ via the equalities (\ref{con:ext_1}). Then, the dual of (\ref{opt:w-primal}) is given,
    \begin{align}
        \min_{\gamma, (Q_i,q_i)_{i \in N}} \gamma \textnormal{ subject to } & \\
        \gamma + \sum_{i \in N} \sum_{a'^w_i \in A^w_i} \left( Q^w_i(a'^w_i,a^w_i) + q_i(a'^w) \right) \cdot u^w_i(a'^w_i,a^w_{-i}) & \geq \tilde{h}(\lambda(a^w)) \ \forall \ a^w \in A^w \nonumber \\
        \sum_{a'^w \in A^w_i} Q^w_i(a'^w_i,a^w_i) + q_i(a'^w) & = 0 \ \forall \ i \in N, a^w_i \nonumber \\
        Q^w_i(a'^w_i,a^w_i) + q_i(a'^w) & \geq 0 \ \forall \ i \in N, a^w_i \neq a'^w_i \nonumber \\
        Q^w_i(a'^w_i,a^w_i) - Q^w_i(a^w_i,a'^w_i) & = 0 \ \forall \ i \in N, a^w_i \neq a'^w_i. \nonumber
    \end{align}
    If equality holds for any $a^w$, then the inequality $\sum_a \sigma(a) \cdot \tilde{h}(a) \leq \tilde{b}$ is facet defining up to multiplication by a constant, and 
    \begin{equation}\label{eq:facet-def}\tilde{h}(a) = \gamma + \sum_{i \in N} \sum_{a'^w_i \in A^w_i} \left( Q^w_i(a'^w_i,a^w_i) + q_i(a'^w) \right) \cdot u_i^w(a'^w_i,a^w_i)\end{equation}
    for any $a^w$ such that $a = \lambda(a^w)$. As a consequence, we will first want to inspect the form of dual solutions $(Q_i,q_i)$, which generate any such function $\tilde{h} : A \rightarrow \mathbb{R}$.

    We first note that the set of feasible $(\gamma, (Q_i,q_i)_{i \in N})$ form a pointed cone. The rays $\gamma = \pm1$, $(Q_i,q_i) = (0,0)$ provide us with $\tilde{h}(a) = \pm1$ for any $a \in A$ respectively, which together assemble into the probability constraint $\sum_{a \in A} \sigma(a) = 1$. Thus it is sufficient to consider the case $\gamma = 0$ and $(Q_i,q_i) \neq (0,0)$ for one player $i$. Now, note that for any pair $a^w_i,a^{*w}_i$ which correspond to the same action $a_i \in A_i$, we need to have 
    $$ \sum_{a'^w_i \in A^w_i} Q^w_i(a'^w_i,a^w_i) + q_i(a'^w) = \sum_{a'^w_i \in A^w_i} Q^w_i(a'^w_i,a^{*w}_i) + q_i(a'^w).$$
    Therefore, symmetrising $(Q_i^w, q^w_i)$ over such pairs $a^w_i,a^{*w}_i$, we obtain a pair $(R^w_i,r^w_i)$ which satisfies the symmetry and tangentiality constraints, such that the equality $\tilde{h}(\lambda(a^w)) = \sum_{a'^w_i \in A^w_i} R^w_i(a'^w_i,a^w_i) + r_i(a'^w)$ is maintained. Explicitly, we fix 
    \begin{align*}
    R^w_i(a'^w_i,a^w_i) & = \frac{1}{|\kappa_i\lambda_i(a'^w_i)| |\kappa_i \lambda_i(a^w_i)|} \sum_{a''^w_i \in \kappa_i\lambda_i(a'^w_i)} \sum_{a^{*w}_i \in \kappa_i\lambda_i(a^w_i)} Q^w_i(a''^w_i, a^{*w}_i), \\
    r^w_i(a'^w_i) & = \frac{1}{|\kappa_i\lambda_i(a'^w_i)|} \sum_{a''^w_i \in \kappa_i\lambda_i(a'^w_i)} q^w_i(a''^w_i).
    \end{align*}
    Then, (\ref{eq:facet-def}) holds for $(R^w_i,r_i^w)$ with $\gamma, (Q^w_j,q_j^w)_{j \neq i}$ all fixed to $0$, given it holds for $(Q^w_i,q^w_i)$.

    The symmetry of $(R^w_i,r^w_i)$ over actions in the $w$-weighted game which correspond to the same action in the original game implies that $(R^w_i,r^w_i)$ specifies a linear transformation over $\Delta(A^w_i)$ via Proposition \ref{prop:linsubset}. In particular, given an action $a_i$ which is represented by any distribution $x_i^w \in \Delta(A^w_i)$ with support $\kappa_i(a_i)$,
    \begin{align}
     \sum_{a^w \in A} \sigma(a^w) \cdot & \sum_{a'^w_i \in A^w_i} \left( R^w_i(a'^w_i,a^w_i) \cdot x_i^w(a^w_i) + r^w_i(a'^w_i) \right) \cdot u_i^w(a'^w_i,a^w_{-i}) \nonumber \\
     = \sum_{a \in A} \sigma(a) \cdot & \sum_{a'_i \in A_i} w_i(a'_i) \cdot \left( R^w_i(\psi_i(a'_i),\psi_i(a_i)) + r^w(\psi(a'_i)) \right) \cdot u_i(a'_i,a_{-i}), \label{eq:represent}
    \end{align}
    where $\psi_i : A_i \rightarrow A^w_i$ chooses a representative of $a_i$ in $A^w_i$, as an injection such that $\psi_i(a_i) \in \kappa_i(a_i)$. Therefore, we may deduce the resulting transformation on $\Delta(A_i)$ through checking it for a representation of $A_i$ in $A^w_i$.
    
    Then, applying (the proof of) Theorem \ref{thm:strategy-def}, we conclude that $(R^w_i, r^w_i)$ decomposes into rays 
    $$ (R^w_i, r^w_i) = (R''^w_i,0) + \sum_{s = 1}^{-1 + \sum_{a_i \in A_i} w_i(a_i) } \mu_s \cdot (R^{sw}_i,r^{sw}_i) + c \cdot (\Sigma\Sigma^T,-\Sigma),$$
    where $\mu \geq 0$ and $c \in \mathbb{R}$. The ray $(\Sigma\Sigma^T,-\Sigma)$ is again trivial, corresponding to the identity on $\Delta(A_i)$. By the invariance of $r^w_i(a^w_i)$ over $a^w_i \in \kappa_i(a_i)$, each $(R^{sw}_i,r^{sw}_i)$ corresponds to an action transformation, specified by a proper subset $S_i \subseteq A_i$. Here, for each $a_i \in S_i$ and each $a^w_i \in \kappa_i(a_i)$, $a^w_i$ is transformed into the uniform distribution on $\cup_{a'_i \notin S_i} \kappa_i(a'_i) \subseteq A^w_i$. Each $\kappa_i(a'_i)$ has cardinality $w_i(a'_i)$. In turn, if $a_i \notin S_i$, $a_i$ is unmodified.

    In turn, again by the proof of Theorem \ref{thm:strategy-def}, the pair $(R''^w_i,0)$ decomposes into a conical combination of pairs $(C^w+C^{wT},0)$, where each $C^w$ is the adjacency matrix of a directed cycle $a^w_{1i} \rightarrow a^w_{2i} \rightarrow ... \rightarrow a^w_{ki} \rightarrow a^w_{1i}$. Moreover, since each $(R^{sw}_i,r^{sw}_i)$ and $(\Sigma\Sigma^T,-\Sigma)$ are constant over $\kappa_i(a'_i)\times \kappa_i(a_i)$, so too is $(R''^w_i,0)$. In case $R''^w(a'^w_i,a^w_i) > 0$ for $a_i = \lambda_i(a'^w_i) = \lambda_i(a^w_i)$, we may factor out a contribution $(R^{a_iw}_i,0)$, such that
    $$ R^{a_i w}(a'^w_i,a^w_i) = \begin{cases}
        -1 & a'^w_i = a^w_i, a^w_i \in \kappa_i(a_i) \\
        \frac{1}{w_i(a_i)-1} & a^w_i, a'^w_i \in \kappa_i(a_i), a'^w_i \neq a^w_i \\
        0 & \textnormal{otherwise,}
    \end{cases}$$
    while ensuring that $(R''^w_i,0)$ satisfies the tangetiality and symmetry constraints. Such $(R^{a_i w},0)$, of course, correspond to the identity transformation on $\Delta(A_i)$. Therefore, it suffices to restrict attention to the case when $k \geq 2$.

    In this case, note that we may again symmetrise the resulting pair $(C^w+C^{wT},0)$ over each such cycle without modifying $(R''^w_i,0)$. In this case, (\ref{eq:represent}) reads,
    $$ \sum_{a_{-i} \in A_{-i}, a_{i\ell}} \sigma(a) \cdot \left( \frac{1}{2w_i(a_{i\ell})} \cdot (u_i(a_{i(\ell-1)},a_{-i}) + u_i(a_{i(\ell+1)},a_{-i}) - 2u_i(a_{i\ell},a_{-i}))\right),$$
    where we denote $a_{i\ell} = \lambda_i(a^w_{i\ell})$. Invoking Proposition \ref{prop:linsubset}, where $\delta = \min_\ell w_i(a_{i\ell})$, we conclude that each $a_{i\ell} \mapsto (1-\delta/w_i(a_{i\ell})) a_{i\ell} + \delta/2w_i(a_{i\ell}) (a_{i(\ell+1)} + a_{i(\ell-1)})$. Other actions $a_i$ not in the cycle then remain invariant.
\end{proof}

\begin{proposition*}[\ref{prop:weighted-bertrand}]
    Suppose there exists at least two firms with minimum marginal cost $c/n \in \{0,1/n,...,1-1/n\}$, and that $(p - c/n) D(p)$ is weakly concave over $[c/n,1]$. Then if $w : \{0,1/n,...,1\} \rightarrow \mathbb{N}$ is strictly increasing, then in every $w$-weighted semicoarse equilibrium $\sigma$ of the discretised Bertrand competition game with bidding sets $\{0,1/n,...,1\}$, each outcome $(p_i)_{i \in N}$ assigned positive probability $\sigma(p) > 0$ is necessarily a Nash equilibrium. An analogous statement also holds for the first-price auction when $w$ is strictly decreasing, through its equivalence with Bertrand competition with inelastic demand.
\end{proposition*}

\begin{proof}
    This follows analogously as the proof of Theorem \ref{thm:bertrand-unique}, but the key differences are that when we fix $\epsilon_1 > 0$ and all $\delta_i > 0$ for $i > m$, it does not rule out the potential additional Nash equilibria where $m = 2$ firms of minimum marginal cost post prices $(c+2)/n$; this can only happen when $D((c+1)/n) = D((c+2)/n)$ and no firms of marginal cost $(c+1)/n$ exist. The case when the minimum price is $\leq c/n$ proceeds identically, as all deviations considered will have weakly positive utility so long as the dual multipliers $\delta_i, \epsilon_k > 0$ for $i > m$ and and $k \geq 1$, and strictly so when $p$ is not a Nash equilibrium. If the minimum price is in $(c+\{1,2\})/n$, then the strategy modifications for firms $i \in \{1,2\}$ do not change a bid of $(c+1)/n$, whereas the bid $(c+2)/n$ is transformed to the bid $(c+1)/n$. This deviation is then not strictly utility improving only when $p$ is a Nash equilibrium. 
    
    Thus the interesting case in this setting is when the minimum price is $c/n + \ell/n$ for $\ell > 2$. Again, suppose that each buyer $i \leq m$ has $p_i = c/n + \ell_i/n$. Then the term (\ref{eq:tie-case}) is modified,  
    \begin{align}
        & \quad \sum_{i \leq m} \sum_{k} \epsilon_k \cdot \sum_{p'_i \in A_i} w_i(p'_i) \cdot  (Q^k(p'_i,p_i) + q^k(p'_i)) \cdot u_i(p'_i,p_{-i})\nonumber \\ \ldots & + \sum_{i > m} \delta_i \cdot \sum_{p'_i \in A_i} w_i(p'_i) \cdot (R^i(p'_i,p_i)+r^i(p'_i)) \cdot u_i(p'_i,p_{-i}) \nonumber \\
        & \geq \sum_{i \leq m} \sum_{k} \epsilon_k \cdot \sum_{p'_i \in A_i} w_i(p'_i) \cdot (Q^k(p'_i,p_i) + q^k(p'_i)) \cdot u_i(p'_i,p_{-i}) \nonumber \\
        & \geq \sum_{i \leq m} \sum_{k = 1}^{\ell-1} \epsilon_k \cdot \left( \sum_{k' = 1}^k \frac{w_i(c/n + k'/n)}{\sum_{k'' = 1}^{k} w_i(c/n + k''/n)} \cdot\left(\frac{ k'}{n}\right) D\left( \frac{c + k'}{n} \right) - \frac{\mathbb{I}[\ell_i = \ell]}{|\arg\min_j p_j|} \left( \frac{ \ell}{n} \right) D \left( \frac{c + \ell}{n} \right)\right) \nonumber \\
        & \geq \sum_{k = 1}^{\ell-1} \epsilon_k \cdot \left( \frac{m \cdot w_i(c/n + k'/n)}{\sum_{k'' = 1}^{k} w_i(c/n + k''/n)} \cdot  \sum_{k' = 1}^k \left(\frac{k'}{n}\right) D\left( \frac{c+k'}{n} \right) - \left( \frac{\ell}{n} \right) D \left( \frac{c+\ell}{n} \right)\right). \label{eq:tie-case-w}
    \end{align}
    And thus, the interesting case is again when $m = 2$ and all firms $\leq m$ bid the same minimum price $c/n + \ell/n$. Noting that we have fixed $\epsilon_1 > 0$ arbitrarily, base case is then $\ell = 2$; for the induction step, the change in (\ref{eq:tie-case-w}) when we let $\ell \mapsto \ell+1$ equals
    \begin{align}
        & \quad \sum_{k = 1}^{\ell} \epsilon_k \cdot \left(\sum_{k' = 1}^k \frac{m \cdot w_i(c/n + k'/n)}{\sum_{k'' = 1}^{k} w_i(c/n + k''/n)} \cdot \left(\frac{k'}{n}\right) D\left( \frac{c+k'}{n} \right) - \left( \frac{\ell+1}{n} \right) D \left( \frac{c+\ell+1}{n} \right)\right) \nonumber \\ &  - \sum_{k = 1}^{\ell-1} \epsilon_k \cdot \left( \sum_{k' = 1}^k \frac{m \cdot w_i(c/n + k'/n)}{\sum_{k'' = 1}^{k} w_i(c/n + k''/n)} \cdot \left(\frac{k'}{n}\right) D\left( \frac{c+k'}{n} \right) - \left( \frac{\ell}{n} \right) D \left( \frac{c+\ell}{n} \right)\right) \nonumber \\
        & = -\left( \sum_{k = 1}^{\ell-1} \epsilon_k \right) \cdot \left( \left( \frac{\ell+1}{n} \right) D \left( \frac{c+\ell+1}{n} \right) - \left( \frac{\ell}{n} \right) D \left( \frac{c+\ell}{n} \right) \right) \nonumber \\
        & + \epsilon_{\ell} \cdot \left( \sum_{k' = 1}^\ell \frac{m \cdot w_i(c/n + k'/n)}{\sum_{k'' = 1}^{\ell} w_i(c/n + k''/n)} \cdot \left(\frac{k'}{n}\right) D\left( \frac{c+k'}{n} \right) - \left( \frac{\ell+1}{n} \right) D \left( \frac{c+\ell+1}{n} \right) \right).        
    \end{align}
    Again, we are able to pick $\epsilon_\ell > 0$ so long as $(c+\ell)/n$ is less than the \emph{minimum} monopoly price $\ell^*/n$; due to weak concavity the monopoly price need not be unique. This is because the increase condition on $w_i$ implies that the coefficient of $\epsilon_\ell$ is strictly positive. Therefore, for $\ell \geq \ell^*$, we pick $\epsilon_\ell = 0$, else, we let 
    $$ \epsilon_\ell = \frac{\left( \sum_{k = 1}^{\ell-1} \epsilon_k \right) \cdot \left( \left( \frac{\ell+1}{n} \right) D \left( \frac{c+\ell+1}{n} \right) - \left( \frac{\ell}{n} \right) D \left( \frac{c+\ell}{n} \right) \right)}{\sum_{k' = 1}^\ell \frac{m \cdot w_i(c/n + k'/n)}{\sum_{k'' = 1}^{\ell} w_i(c/n + k''/n)} \cdot \left(\frac{k'}{n}\right) D\left( \frac{c+k'}{n} \right) - \left( \frac{\ell+1}{n} \right) D \left( \frac{c+\ell+1}{n} \right)} > 0.$$
\end{proof}

To prove Theorem \ref{thm:P-scaled-sets}, we will need to define a condition number for a polyhedron $P$. In essence, given a face $F$ of $P$ defined via an additional equality constraint $a_i^T x=b_i$ and a point $x \in P \setminus F$, the condition number provides a lower bound on the angle between $\Pi_{F}[x]-x$ and $a_i$. A proof of this fact was proven by \cite{421564} for the case when $P$ is a polytope and we suspect it to be a known result, though we are unable to locate a proof for the more general case; hence we provide one.

\begin{lemma}\label{lem:angle-bound}
    Let $P \subseteq \mathbb{R}^n$ be a polyhedron, defined via (finitely many) linear inequalities $a_i^T x \leq b_i$ where for each $i$, $\|a_i\| = 1$. Then there exists a constant $\chi(P) \in (0,\infty)$ such that for any $F_\ell$ be the face of $P$ defined by the additional equality $a_\ell^T x = b_\ell$, for any $x \in P \setminus F$, $$\min_{y \in F} \frac{\|x-y\|}{b_\ell^T - a^T_\ell x} = \frac{\| \Pi_F[x]-x\|}{a_\ell^T(\Pi_F[x]-x)} \leq \chi(P).$$
\end{lemma}

\begin{proof}
    Let $\ell$ be arbitrary, and fix $x \in P \setminus F$. Let $y = \Pi_{F}[x]$. Then $x-y \in \TC{P}{y}$, since both $x,y$ are in $P$. Moreover, $x-y \in \NC{F}{y}$, due to the projection property. We conclude that $x-y \in \NC{F}{y} \cap \TC{P}{y}$, which is a polyhedral cone. Now, note that for any given $\ell$, the set $\NC{F}{y} \cap \TC{P}{y}$ depends solely on the relative interiors of the faces of $F$ and $P$ $y$ belongs to -- of which there are finitely many selections. Therefore, it is sufficient to prove the existence of a strictly positive lower bound for a given $y$, and $v = x-y \in \NC{F}{y} \cap \TC{P}{y}$.

    We first claim that $v^Ta_\ell < 0$ for any $v \neq 0$. Suppose not, then since $v^T a_\ell \in \TC{P}{y}$ and $a_\ell^T y = b_\ell$, we need $v^T a_\ell \leq 0$. If $v_T a_\ell \geq 0$, we conclude that $v^T a_\ell = 0$. In this case, since $P$ is a polyhedron, there exists $\eta > 0$ such that $y + \eta v\in P$. However, $a_\ell^T(y+\eta v) = b_\ell$, and thus $y + \eta v \in F$. This implies that $v \in \TC{F}{y}$. However, $v \in \NC{F}{y}$, and $\TC{F}{y} \cap \NC{F}{y} = \{0\}$ -- a contradiction.

    Now, the value $-\|v\| / a_\ell^Tv$ is invariant under multiplication by a positive constant, so if $v \neq 0$, we may without loss of generality assume that $a_\ell^T v = -1$. Therefore, we would like to maximise $\|v\|$ on the intersection of the cone $\NC{F}{y} \cap \TC{P}{y}$ and the hyperplane $a_\ell^T v = -1$. Now, all potential solutions are convex combinations of rays $r$ of $\TC{F}{y} \cap \NC{F}{y}$ which generate the cone, each scaled such that $r^T a_\ell = -1$. The feasible region is thus a polytope, and the maximum exists since $\|v\|$ is continuous. We then define $\chi(P)$ to be the maximum such bound over all choices $\ell$.
\end{proof}

Lemma \ref{lem:angle-bound} then allows us to bound the largest step size we may take along the gradient field of a given tangential function $h$ on a polyhedron, which will be crucial for the first step in our proof of Theorem \ref{thm:P-scaled-sets}.

\begin{lemma}\label{lem:condition}
    Suppose that $P \subseteq \mathbb{R}^n$ is a polyhedron, and $h : P \rightarrow \mathbb{R}$ a tangential function with Lipschitz continuous gradients (of modulus $L_h$). Then for any $x \in P$, $x + \nabla h(x) / (\chi(P)\cdot L_h) \in P$.
\end{lemma}

\begin{proof}
    Let $P$ be given by the set of linear inequalities, $a_i^T x \leq b_i$, where $\|a_i^T\| = 1$ for each inequality $i$. Fix $x \in P$, and consider $\Delta = \sup \{ \delta \ | \ a_i^T(x+\delta\nabla h(x)) \leq b_i \ \forall \ i \}$. By the tangentiality assumption on $h$, $a^T_i\nabla h(x) > 0$ only for inequalities $i$ such that $a^T_i x < b_i$, so $\Delta > 0$. If $\Delta = \infty$, we are done, so suppose that it is finite.

    Then, $x + \Delta \nabla h(x) \in P$, and we seek to bound $\Delta$ from below. Take any inequality $\ell$ such that $a^T_\ell x < b_\ell$ and $a^T_\ell (x+\Delta \nabla h(x)) = b_\ell$. Then the face $F$ of $P$, defined via the intersection of $P$ with the plane $a^T_\ell x = b_\ell$, is non-empty. Therefore, the projection $\Pi_F[x]$ of $x$ onto the plane $F$ exists. 
    
    Now, by the tangentiality assumption on $h$, $ a^T_\ell \nabla h(\Pi_F[x]) \leq 0$; this is because $\Pi_F[x] \in F$ and thus $a^T_\ell \Pi_F[x] = b_\ell$. We then infer,
    $$ a^T_\ell \nabla h(x) \leq a^T_\ell \left(\nabla h(x) - \nabla h\left(\Pi_F[x]\right) \right) \leq L_h \|x - \Pi_F[x] \|,$$
    where the second inequality is via the Lipschitz continuity of $\nabla h$ and the Cauchy-Schwarz inequality. Recalling again that $a_\ell^T\Pi_F(x) = b_\ell$ and invoking Lemma \ref{lem:angle-bound}, we conclude 
    $$\Delta =  \frac{b_\ell - a_\ell^Tx}{a_\ell^T \nabla h(x)} \geq \frac{1}{L_h} \cdot \frac{a_\ell^T (\Pi_F[x]-x)}{\|\Pi_F[x]-x\|} \geq \frac{1}{L_h \chi(P)}.$$
\end{proof}

\begin{theorem*}[\ref{thm:P-scaled-sets}]
    In a smooth game, suppose that players $i \in N'$ all employ the same non-increasing step sizes $\eta_t$, and have polyhedral strategy sets $X_i$. Let $h : X \rightarrow \mathbb{R}$ be a bounded tangential function, with bounded and Lipschitz gradients, which depends only on coordinates $(x_i)_{i \in N'}$ -- i.e. for some $M, G_h, L_h > 0$, for any $x, x' \in X$,
    $|h(x)| \leq M$, $ \| \nabla h(x) \| \leq G_h \textnormal{ and } \| \nabla h(x) - \nabla h(x') \| \leq L_h$, and $\nabla_j h(x) = 0$ for any $x \in X$. Then for any $T > 1$, and any $\delta \in [0,1/(L_h \cdot \chi(X))]$,
    \begin{align*}
        & \frac{1}{T} \cdot \sum_{i \in N} \sum_{t = 1}^{T} (u_i(\Pi_{X_i}[x^t_i+\delta \nabla_i h(x^t)],x^t_{-i})-u_i(x^t)) \\
        \leq \ & \delta \cdot \left( \frac{2M}{T} \left[\frac{1}{\eta_{T}} + \frac{1}{\eta_1}\right]+ \frac{\sum_{t = 1}^T \eta_t}{T} \cdot \sum_{i \in N'}\left(\frac{G_i^2 L_h}{2} + G_iL_h \cdot \sum_{j \in N'}G_j\right) \right) + \frac{1}{2T} \cdot \delta^2 G_h^2 \sum_{i \in N'} L_i.
    \end{align*}
    As a consequence, when all players use the same step sizes $\eta_t = C/\sqrt{t}$, after $T$ rounds of projected gradient ascent, the time-averaged distribution $\sigma$ which samples $x^t$ with probability $1/T$ is an $(\epsilon,\Delta)$-local coarse correlated equilibrium with respect to the set of bounded tangential functions $H$ with bounded gradients and  Lipschitz modulus $\leq 1$, where we may fix 
    $\epsilon = ( 4M/C + 2\max\{1,C\} )/\sqrt{T}$ and $\Delta = \min\{1/\chi(X),1/\sqrt{T}\}$. Moreover, evaluating $\delta \downarrow 0$, $\sigma$ is an $\epsilon$-local coarse correlated equilibrium with respect to $H$.
\end{theorem*}

\begin{proof}
    We will prove the given inequality, as the latter part of the theorem statement follows immediately. First note that since the projection operator for closed \& convex sets are contractive, $\|\Pi_{X_i}[x_i + v] - x_i\| \leq \|v\|$ for any vector $v \in \mathbb{R}^{D_i}$. Moreover, by our choice of $\delta$ and by Lemma \ref{lem:condition}, $\Pi_{X_i} [x_i^t+\delta\nabla_ih(x^t)] = x^t + \delta \nabla_i h(x^t)$ for any $t$. We then proceed, 
    \begin{equation*}
        \sum_{i \in N} \sum_{t = 1}^T \left( u_i( \Pi_{X_i} [x_i^t+\delta\nabla_ih(x^t)] )  - u_i(x^t)\right)
        \leq \sum_{i \in N'} \left( \delta \nabla_i u_i(x^t)^T\nabla_ih(x^t) + \frac{1}{2} \cdot \delta^2 \|\nabla_i h(x^t)\|^2 L_i \right).
    \end{equation*}
    Now, $\|\nabla_ih(x^t)\|^2 \leq G_h^2$, which provides the $O(\delta^2)$ term we desire. In turn, we write
    \begin{align}
        & \sum_{i \in N'} \sum_{t = 1}^T\delta  \nabla_iu_i(x^t)^T\nabla_ih(x^t) \nonumber \\
        = \ & \sum_{i \in N'} \sum_{t = 1}^T \frac{\delta}{\eta_t} \nabla_ih(x^t) \cdot (x_i^t + \eta_t \nabla_i u_i(x_i^t) - x_i^{t+1}) + \frac{\delta}{\eta_t} \nabla_i h(x^t)\cdot (x^{t+1}-x^t). \label{eq:int-step-LCCE}
    \end{align}
    By the Lipschitz condition on $h$, and its dependence solely on $(x_i)_{i \in N'}$,
    $$ \sum_{i \in N'} \nabla_i h(x^t)\cdot (x^{t+1}_i-x^t_i) \leq h(x^{t+1}) - h(x^t) + \frac{1}{2}  L_h \sum_{i \in N'} \|x^{t+1}_i-x^t_i\|^2 \leq h(x^{t+1}) - h(x^t) + \frac{1}{2} \eta_t^2 L_h \sum_{i \in N'} G_i^2.$$
    Then the second term of (\ref{eq:int-step-LCCE}) is bounded above by
    $$ \leq \delta \cdot \sum_{t = 1}^T \frac{ h(x^{t+1}) - h(x^t) }{\eta_t} + \delta\cdot  \sum_{t = 1}^T \eta_t \cdot \frac{1}{2} L_h \sum_{i \in N'} G_i^2.$$
    We want to bound the sum of the $O(1/\eta_t)$ term here. This follows from rearranging the sum,
    \begin{align*}
       & \quad \ \sum_{t = 1}^T \frac{1}{\eta_t} \cdot (h(x^{t+1})-h(x^t)) \\
       & = -\frac{1}{\eta_1} h(x^1) + \frac{1}{\eta_T} h(x^{T+1}) - \sum_{t = 1}^{T-1} h(x^{t+1}) \cdot \left( \frac{1}{\eta_{t+1}} - \frac{1}{\eta_{t}} \right) \\
       & \leq 2M \left(\frac{1}{\eta_1} + \frac{1}{\eta_T}\right).
    \end{align*}
    To bound the first term of (\ref{eq:int-step-LCCE}), we remark that $x^{t+1}_i = \Pi_{X_i}[x_i^t + \eta_t \nabla_i u_i(x^t)]$. Therefore, by the projection property,  $x_i^t + \eta_t \nabla_i u_i(x^t) - x_i^{t+1} \in \NC{X_i}{x_i^{t+1}}$. Meanwhile, since $h$ is tangential by assumption, $\nabla_i h(x^{t+1}) \in \TC{X_i}{x_i^{t+1}}$. This implies that $(x_i^t + \eta_t \nabla_i u_i(x^t) - x_i^{t+1})^T \nabla_i h(x^{t+1}) \leq 0$, and thus 
    \begin{align*}
        & \sum_{i \in N'} \sum_{t = 1}^T \frac{\delta}{\eta_t} \nabla_ih(x^t) \cdot (x_i^t + \eta_t \nabla_i u_i(x^t) - x_i^{t+1}) \\
        \leq \ & \sum_{i \in N'} \sum_{t = 1}^T \frac{\delta}{\eta_t} (\nabla_ih(x^t)-\nabla_i h(x^{t+1})) \cdot (x_i^t + \eta_t \nabla_i u_i(x^t) - x_i^{t+1}) \\
        \leq \ & \sum_{i \in N'} \sum_{t=1}^T \frac{\delta}{\eta_t} L_h \|x^{t+1}_{N'}-x^t_{N'}\| \cdot \eta_t \|\nabla_i u_i(x^t)\| \leq \delta \sum_{t= 1}^T \eta_t \sum_{i \in N'} G_i L_h \cdot \sum_{j \in N'} G_j,
    \end{align*}
    where we wrote $x^{t}_{N'}$ to denote the vector $x^t$ over the strategies of the players in $N'$ only.
\end{proof}

\begin{theorem*}[\ref{thm:P-scaled}]
    Suppose that each player $i$ employs projected gradient ascent on their utilities over their scaled strategy sets $Y_i = P_i^{-1}\Delta(A_i)$ with decreasing step sizes $\eta_{it} = \omega(t), o(1)$. Denote the time averaged distribution $\sigma^T$, sampling each $x^t = (P_iy_i^t)_{i \in N}$ with probability $1/T$, and $(P_iP_i^T) = Z_i$. Then for any  convergent subsequence $\sigma^{T_k} \rightarrow \sigma$ and for any player $i$,
    $$ \sum_{a \in A} \sigma(a) \cdot \sum_{a'_i, a''_i \in A_i} Z_i(a'_i,a''_i) \cdot (Q_i(a''_i,a_i) + q_i(a''_i)) \cdot u_i(a'_i,a_{-i}) \leq 0$$
    whenever $Q_i$ is a symmetric matrix such that $Z_i(Q_i,q_i)$ satisfies (\ref{con:preserve}), (\ref{con:tangency}). 
\end{theorem*}

\begin{proof}
    We let $y_i = P_i^{-1}x_i$ for every $x_i \in \Delta(A_i)$. Noting that $u^P_i(y) = u_i(P_iy_i,x_{-i})$ for any $y_i$, we see that $\nabla_i u^P_i(y) = P_i^T \nabla_iu_i(x)$ for any $x \in \times_{i \in N} \Delta(A_i)$. Meanwhile, for a function $h(x) = \frac{1}{2}x_i^T Q_i x_i + q_i^T x_i$, denote 
    $$ h^P(y) = \frac{1}{2} y_iP_i^T Q_iP_i y_i+ q_i^T P_iy_i.$$
    Evaluating the gradient, 
    $$ \nabla_i h^P(y)= P_i^T (Q_i P_iy_i+q_i)=P_i^T (Q_ix_i +q_i)= P_i^T \nabla_ih(x).$$
    Therefore, at any $y$,
    \begin{equation}\label{eq:bilin-y-form}
        \bilin{\nabla_iu^P_i(y)}{\nabla_ih^P(y)} = \bilin{P_i^T \nabla_iu_i(x)}{P_i^T\nabla_ih(x)} = \nabla_i u_i(x)^T (P_iP_i^T)\nabla_i h(x).
    \end{equation}
    Now, to invoke Theorem \ref{thm:P-scaled-sets}, we need to ensure that $\nabla_i h^P(y)$ is tangential. We note that the (in spirit) converse of Lemma \ref{lem:condition} holds in this case; for a function $h$ defined on a polyhedron $P$, if there exists a $\delta > 0$ such that for any $x \in P$, $x + \delta \nabla h(x) \in P$, then $h$ is tangential. Therefore, we need to ensure that for some $\delta > 0$, 
    \begin{align*}\forall y \in P_i^{-1}\Delta(A_i), y + \delta\nabla_ih^P(y) \in P_i^{-1}\Delta(A_i).\end{align*}
    This condition is equivalent to 
    \begin{align*}\forall x \in \Delta(A_i), x + \delta(P_iP_i^T)\nabla_ih(x) \in \Delta(A_i).\end{align*}
    Or in other words, $(P_iP_i^T)\nabla_ih(x)$ must be tangential over $\Delta(A_i)$; this is ensured whenever the pair $((P_iP_i^T)Q_i,(P_iP_i^T)q_i)$ satisfies (\ref{con:preserve}) and (\ref{con:tangency}), that is, writing $Z_i = P_iP_i^T$,
    \begin{align*}
        \forall \ a_i \in A_i, \sum_{a'_i,a''_i \in A_i} Z_i(a'_i,a''_i)(Q_i(a''_i,a_i)+ q_i(a''_i)) & = 0, \\
        \forall \ a_i \neq a'_i, \sum_{a''_i \in A_i} Z_i(a'_i,a''_i)(Q_i(a''_i,a_i)+ q_i(a''_i)) & \geq 0.
    \end{align*}
    The symmetry requirement on $Q_i$, in turn, is because $\nabla_i (x_i^T Q_ix_i) = (Q_i+Q_i^T)x_i$; i.e. differentiation would symmetrise $Q_i$ otherwise.

    Given that $h^P$ is tangent to $Y$, and depends solely on $y_i$, by Theorem \ref{thm:P-scaled-sets}, in the limit $T \rightarrow \infty$, $\sigma^T$ incurs vanishing regret [$o(1)$] against $h$ in first-order. This regret bound is generated by the vector field $x_i \mapsto (P_iP_i^T)(Q_ix_i+q_i)$ over $\Delta(A_i)$, which corresponds to a linear transformation $\Delta(A_i) \rightarrow \Delta(A_i)$ of the appropriate form by Lemma \ref{lem:lin-reg}. In particular, abusing notation by writing $\sigma^T$ as the probability distributions over $Y$, $X$ as well as $A$ obtained via projected gradient ascent, 
    \begin{align*} & \sum_{i \in N} \int_Y d\sigma^T(y) \cdot \bilin{\nabla_iu^P_i(y)}{\nabla_ih^P(y)} \leq o(1) \\
    \textnormal{(\ref{eq:bilin-y-form})} \Rightarrow \ & \sum_{i \in N} \int_X d\sigma^T(x) \cdot \nabla_i u_i(x)^T (P_iP_i^T)\nabla_i h(x) \\ = \ & \int_X d\sigma^T(x) \cdot \nabla_iu_i(x)^T Z_i (Q_ix_i+ q_i) \\
    = \ &  \int_X d\sigma^T(x) \cdot \nabla_iu_i(x)^T Z_i (Q_i+ q_i \Sigma^T) x_i \\ 
    = \ &  \int_X d\sigma^T(x) \cdot \sum_{a_{-i} \in A_{-i}}\left(\prod_{j \neq i} x_j(a_j) \right) \sum_{a'_i \in A_i} u_i(a'_i,a_{-i}) \cdot \\ \ldots \ & \sum_{a''_i \in A_i} Z_i(a'_i,a''_i) \cdot \sum_{a_i \in A_i} (Q_i(a''_i,a_i) + q_i(a''_i))\cdot x_i(a_i) \\
    = \ & \sum_{a\in A} \left[\int_X d\sigma^T(x) \cdot \left( \prod_{j \in N} x_j(a_j) \right) \right] \cdot \sum_{a''_i, a'_i \in A_i} Z_i(a'_i,a''_i) \cdot (Q_i(a''_i,a_i) + q_i(a''_i)) \cdot u_i(a'_i,a_{-i}) \\
    = \ & \sum_{a\in A} \sigma^T(a) \cdot \sum_{a''_i, a'_i \in A_i} Z_i(a'_i,a''_i) \cdot (Q_i(a''_i,a_i) + q_i(a''_i)) \cdot u_i(a'_i,a_{-i}) \leq o(1).
    \end{align*}
    Thus, as $T \rightarrow \infty$, for any convergent subsequence of $\sigma^T$, the limit $\sigma$ satisfies the desired equilibrium conditions.
\end{proof}

\section{Retrieving Convergence Bounds from Dual Solutions}\label{sec:explicit-convergence}

\subsection{Dual bounds for time-average convergence guarantees}

Here, we shall discuss how to adapt the proof of Theorem \ref{thm:bertrand-unique} and Theorem \ref{thm:bertrand-unique-three} for explicit last iterate convergence bounds. The first step of the argument necessitates us to choose $\epsilon_k, \delta_i$ explicitly in the proof such that equation (\ref{eq:pointwise}) holds. In what follows, we will fix $\delta_i$, then fix $\epsilon_1$ iteratively until settling on a final bound, and then will iteratively define $\epsilon_k$.

Recall that the pairs $(Q^k,q^k)$ and $(R^i,r^i)$ are chosen such that they are induced by the corresponding strategy modifications $\phi_k, \phi^i$ as deduced from in Theorem \ref{thm:strategy-def}.(1) for suitable choices of $S_i$, that is $\phi_k : x_i \mapsto (\mathbf{1}+Q^k)x_i +q^k$ for $i \leq m$, and $\phi^i : x_i \mapsto (\mathbf{1}+R^i)x_i + r^i$ for $i > m$. $\phi_k$ maps prices outside $(c+\{1,2,...,k\})/n$ to the uniform distribution on that set, while $\phi^i$ maps any strictly underbidding price to the uniform distribution on prices $\geq c_i/n$. (\ref{eq:pointwise}), for a given price vector $p$, has on its LHS the utility gain from such a deviation, given each firm $i$ posts price $p_i$.

\newcommand{\barp}{\bar{p}}

We thus repeat the case analysis on the price vector $p$. If the minimum posted price $\barp < c/n$, then all winning firms have negative utility. This negative utility is 
$$ \leq -\frac{1}{N} \left(\barp - \frac{c_i}{n}\right) D(\barp) \leq -\frac{1}{nN}D\left(\frac{c_i-1}{n}\right),$$
accounting for maximal number of firms tying at price $(c_i-1)/n$. Meanwhile, all deviations considered have $\geq 0$ utility, and this holds with equality whenever there are at least two firms who post prices $\barp$. Therefore, if we ensure $\epsilon_1 \geq nN / D((c-1)/n)$ and $\delta_i = nN / D((c_i-1)/n)$ for $i > m$, then (\ref{eq:pointwise}) holds for $p$.

Now suppose that the minimum price posted is $\barp = c/n$. Again, all firms have weakly negative utility, and all deviations considered have weakly positive utility, so the LHS is $\geq 0$. If the RHS $= 0$, then (\ref{eq:pointwise}) holds, so assume $d(p) = 1$. If any firm $i > m$ posts price $\barp$, then $\delta_i$ is set such that (\ref{eq:pointwise}) holds, so consider the case when only firms of lowest marginal cost bid $\barp$. Then if $p$ is not a Nash equilibrium, there is only one firm $i \leq m$ which posts price $c/n$. 

Furthermore, the strategy modification $\phi^{i'}$ does not change the actions of firms $i' > m$, and for a firm $i' \leq m, i \neq i'$, the deviations $\phi_k$ either leave their price $p_i$ unchanged or map it to a uniform distribution over losing prices. Meanwhile, only $\phi_1$ guarantees a positive utility change to firm $i$. In this case, their utility goes from $0$ to a quantity
$$ \geq \frac{1}{mn} D\left( \frac{c+1}{n} \right),$$
accounting for the possibility that the other $m-1$ firms of lowest marginal cost all post prices $(c+1)/n$. Since demand is non-increasing and $m \leq N$, this quantity is 
$$ \geq \frac{1}{Nn} D\left(\frac{c+1}{n}\right).$$ 
Therefore, we want to now let $\epsilon_1 = \max\left\{ nN / D\left(\frac{c-1}{n}\right), nm / D\left(\frac{c+1}{n}\right) \right\}$. Now, suppose that the minimum posted price equals $c/n+1/n$. All deviations considered are still weakly utility improving, so (\ref{eq:pointwise}) holds immediately if $d(p) = 0$. Indeed, if $d(p) = 0$, then if there is a firm of marginal price $c_i/n > (c+1)/n$ which posts price $(c+1)/n$, then $\delta_i$ is already set such that (\ref{eq:pointwise}) holds. So suppose not, in which case for $p$ to not be an equilibrium, there must exist a firm $i \leq m$ which doesn't post price $(c+1)/n$. Note that $\phi_1$ is a uniform deviation to $(c+1)/n$, so this improves the utility of firm $i$ by 
$$\geq \frac{1}{nN}D\left(\frac{c+1}{n}\right).$$
So, afterall, we let 
$\epsilon_1 = \frac{nN}{D\left(\frac{c+1}{n}\right)}$, 
which dominates all previous bounds we set on $\epsilon_1$.

Finally, we consider the case when the minimum posted price is $\barp = (c+\ell)/n$ for $\ell \geq 2$. If $\ell = 2$, then we need (\ref{eq:base-case}) to hold. By the strict concavity assumption on $(p-c/n)D(p)$, 
$$ \frac{1}{2} \cdot 0 \cdot D\left(\frac{c}{n} \right) + \frac{1}{2} \cdot \left(\frac{2}{n}\right) D\left( \frac{c+2}{n} \right) < \left( \frac{1}{n} \right) D \left( \frac{c+1}{n} \right),$$
which implies that (\ref{eq:base-case}) holds whenever $\epsilon_1 \geq 1 / \left( m \cdot \left(\frac{1}{n}\right) D\left( \frac{c+1}{n} \right) - \left( \frac{2}{n} \right) D \left( \frac{c+2}{n} \right)\right).$ The final desired bound on $\epsilon_1$ is thus 
$$\epsilon_1 = \max\left\{ \frac{1}{m \cdot \left(\frac{1}{n}\right) D\left( \frac{c+1}{n} \right) - \left( \frac{2}{n} \right) D \left( \frac{c+2}{n} \right)}, \frac{nN}{D\left(\frac{c+1}{n}\right)}\right\}.$$
The rest of the proof of Theorem \ref{thm:bertrand-unique} maintains the value of the LHS while choosing $\epsilon_\ell$, for the cases $\ell > 2$. Recalling that we denoted the monopoly price $\ell^*$, we still have, for any $2 < \ell < \ell^*$,
$$ \epsilon_\ell = \frac{\left( \sum_{k = 1}^{\ell-1} \epsilon_k \right) \cdot \left( \left( \frac{\ell+1}{n} \right) D \left( \frac{c+\ell+1}{n} \right) - \left( \frac{\ell}{n} \right) D \left( \frac{c+\ell}{n} \right) \right)}{\frac{m}{\ell} \sum_{k' = 1}^\ell \left(\frac{k'}{n}\right) D\left( \frac{c+k'}{n} \right) - \left( \frac{\ell+1}{n} \right) D \left( \frac{c+\ell+1}{n} \right)}.$$

All that remains to do is to infer the regret guarantees. Define the quadratic functions,
$$ h_i(x_i) = \begin{cases}
    \sum_{k = 1}^{\ell^*-1} \epsilon_k \cdot \left( \frac{1}{2} x_i^T Q^k x_i + q^{kT}x_i\right) & i \leq m, \\
    \delta_i \cdot \left( \frac{1}{2} x_i^T R^i x_i + r^{iT}x_i\right) & i > m.
\end{cases}$$
They are tangential, and while we can use the regret bound of Proposition \ref{prop:LCCE}, those in Theorem \ref{thm:P-scaled-sets} (after taking $\delta \downarrow0$) will be easier to wield -- and the duality theory functions identically. Theorem \ref{thm:P-scaled-sets} states (by replacing $h$ with $h_i$, extending its definition to all of $X$) that when each firm $i$ uses step sizes $\eta_{it}$, their regret against $\nabla_ih_i(x_i)$ will be bounded,
\begin{equation}\label{eq:int-bound}
\frac{1}{T} \cdot \sum_{t=1}^T \bilin{\nabla_i h_i(x_i^t)}{\nabla_iu_i(x^t)} \leq \frac{1}{T} \cdot \left( \frac{1}{\eta_{iT}} +\frac{1}{\eta_{i1}} + \sum_{t=1}^T \eta_{it} \right) \cdot \poly(G,L,G_h,L_h). \tag{$r_i$}
\end{equation}
Here, Theorem \ref{thm:P-scaled-sets} in the limit $\delta \downarrow 0$ suggests we may take $\poly(G,L,G_h,L_h) = 2M_i + \frac{3G_i^2L_{h_i}}{2}$, where $M_i$ is a bound on the maximum value of $h_i$, $G_i$ is a bound on $\|\nabla_i u_i(x)\|$ and $L_{h_i}$ is a bound on the Lipschitz modulus of $h_i$. 

We now get to derive quick and dirty bounds on these quantities, using the fact that each pair $(Q^k,q^k), (R^i,r^i)$ corresponds to a function considered in (\ref{eq:gen-ansatz}). The absolute bound on each such function is $2$, and their Lipschitz modulus is simply the largest magnitude of their eigenvalues, which happens to be $1$. As for bounds on the utility gradients, the maximum utility from any posted price is $\max_{p \in [0,1]} pD(p)$ as the monopolist's profits for marginal cost $0$, while the minimum profit from any posted price is $-D(0)$, attained when a firm of cost $1$ posts price $0$. Thus, a quick bound on the utility gradients is simply $(n+1)\cdot U$ for some constant determined by the demand function. 

As a consequence, for any firm $i$, 
$$M_i/2 \leq \begin{cases}
    \sum_{k = 1}^{\ell^*-1} \epsilon_k & i \leq m \\
    \delta_i & i > m
\end{cases},$$
and $L_{h_i}$ admits the same bound. Summing over all players, we conclude that
\begin{align*}
    & \ \ \ \ \ \frac{1}{T} \cdot \sum_{i \in N} \sum_{t=1}^T \bilin{\nabla_i h_i(x_i^t)}{\nabla_iu_i(x^t)} \\ & \leq \frac{1}{T} \cdot \sum_{i \leq m} \left( \frac{1}{\eta_{iT}} +\frac{1}{\eta_{i1}} + \sum_{t=1}^T \eta_{it} \right) \cdot \left( \sum_{k = 1}^{\ell^*-1} \epsilon_k \right) \cdot (4+3(n+1)^2 U^2/2) \\
    &  + \frac{1}{T} \cdot \sum_{i > m} \left( \frac{1}{\eta_{iT}} +\frac{1}{\eta_{i1}} + \sum_{t=1}^T \eta_{it} \right) \cdot \delta_i \cdot (4+3(n+1)^2 U^2/2).
\end{align*}

Explicit bounds follow from considering the above as a relaxed equilibrium constraint in a primal problem, leveraging (\ref{eq:pointwise}). In particular, at each time period $t$, we take the expectation of (\ref{eq:pointwise}) when $p_i$ is drawn from $x_i^t$, and average these quantities over time. In this case, the LHS of (\ref{eq:pointwise}) becomes the LHS of (\ref{eq:int-bound})\footnote{This is an LP duality argument!}. The time average convergence guarantees are immediate.

\begin{proposition}\label{prop:time-avg-conv-rate}
    In our Bertrand competition setting, suppose that each firm $i$ updates their prices using step size $\eta_{it}$, such that $\eta_{it} = o(1)$ and $\sum_{t = 1}^T \eta_{it} = o(T)$. Then after $T$ time steps,
    \begin{align*}
        &\frac{1}{T} \cdot \sum_{t = 1}^T \mathbb{E}_{p_i \sim x_i^t}[\mathbb{I}[p \textnormal{ is not a Nash equilibrium}]] \\
        \leq \ & \frac{1}{T} \cdot \sum_{i \leq m} \left( \frac{1}{\eta_{iT}} +\frac{1}{\eta_{i1}} + \sum_{t=1}^T \eta_{it} \right) \cdot \left( \sum_{k = 1}^{\ell^*-1} \epsilon_k \right) \cdot (4+3(n+1)^2 U^2/2) \\
     + \ &  \frac{1}{T} \cdot \sum_{i > m} \left( \frac{1}{\eta_{iT}} +\frac{1}{\eta_{i1}} + \sum_{t=1}^T \eta_{it} \right) \cdot \delta_i \cdot (4+3(n+1)^2 U^2/2),
    \end{align*}
    which goes to $0$ as $T \rightarrow \infty$. Here $\epsilon_k, \delta_i, U$ are defined as above.
\end{proposition}

For example, consider the case of inelastic demand (equivalent to the first price auction) when $m \geq 3$, and all firms use the same step sizes $\eta_t = 1/\sqrt{T}$. In this case, we have $\epsilon_1 = \delta_i = nN$. Moreover, the sums $\sum_{k = 1}^{\ell} \epsilon_k$ satisfy the recurrence relation,
$$ \sum_{k = 1}^{\ell+1} \epsilon_k = \sum_{k = 1}^{\ell} \epsilon_k \cdot \left( 1 + \frac{2}{(m-2)(\ell+2)} \right)$$
for any $\ell \geq 1$. This has solution,
$$ \sum_{k =1}^\ell \epsilon_k = nN \cdot \frac{2 \cdot \prod_{k=0}^{n-2} \left( 3 + \frac{2}{m-2}+k\right)}{(n+1)!} \leq \frac{1}{12} Nn(\ell+2)(\ell+3).$$
For simplicity\footnote{The bounds slightly improve with higher minimum price $c/n$.}, fix the minimum price equal to $0$, in which case the monopoly price simply has $\ell^* = n$, due to inelastic demand. Finally, the step-size related term is bounded above by $3\sqrt{T}$, and the bound on the utilities $U$ equals $2$. Thus, our overall bound is
\begin{align*}
    & \ \ \ \ \ \frac{1}{T} \cdot \sum_{i \in N} \sum_{t=1}^T \bilin{\nabla_i h_i(x_i^t)}{\nabla_iu_i(x^t)} \\ & \lesssim \frac{3}{\sqrt{T}} \cdot m \cdot \frac{1}{12} Nn(n+1)(n+2) \cdot (4+6(n+1)^2) \\
    &  + \frac{3}{\sqrt{T}} \cdot (N-m) \cdot Nn \cdot (4+6(n+1)^2) \\
    & \leq \frac{3}{\sqrt{T}} \cdot m \cdot Nn^5
\end{align*}
for any $n$ sufficiently large (say $>20$). Invoking Proposition \ref{prop:time-avg-conv-rate}, we conclude that after $T$ time steps, the average play $\sigma(p) = \frac{1}{T} \cdot \sum_{t = 1}^T \prod_{j \in N} x_j(p)$ concentrates on Nash equilibria with probability at least $1-\frac{3}{\sqrt{T}} mNn^5$. This assures time-average convergence, and for more general step-sizes in the setting of inelastic demand, the dependence on $T$ is to be modified appropriately.

\subsection{Dual Lyapunov arguments for local stability and last-iterate convergence}

Finally, we shall discuss the arguments necessary to establish last iterate convergence, which actually turns out to be a finite iterate convergence guarantee for the Bertrand competition. We recall that the dual solution in the previous section establishes a guarantee, via taking the expectation of (\ref{eq:pointwise}) at any mixed-strategy profile $x$
$$ \sum_{i \in N} \bilin{\nabla_ih_i(x_i)}{\nabla_i(u_i(x))} \geq \mathbb{E}_{p_i \sim x_i}[p \textnormal{ is not a Nash equilibrium}].$$
The LHS is a bound on the time derivative of $h(x)$ for the continuous projected gradient dynamics of the game (c.f. \cite{ahunbay2024local}, Section 5.1), and since the RHS is positive semidefinite function on $\times_{i \in N} \Delta(A_i)$ which equals $0$ if and only if $x$ is a pure or mixed NE. We infer that $h$ is a Lyapunov function, \emph{with reversed sign convention}. The equilibria\footnote{Other firms have multiple actions they may take.} where all firms of minimum marginal cost post prices $(c+1)/n$ is stable, corresponding to the set of maxima of $h$, whereas the equilibria where at least two such firms post prices $c/n$ are unstable. 

The RHS bounds below the rate of increase of $h$, except this lower bound goes to $0$ as we approach unstable equilibria. This is to be expected, as $\nabla_i u_i(x)$ for each player $i \leq m$ vanishes as $x$ is brought closer an unstable Nash equilibrium of this game. This translates to no lower bound on the \emph{rate} of increase of $h$, and indeed, getting arbitrarily close to the unstable equilibria can arbitrarily increase the time required for the dynamics to escape from it. As a result, we cannot hope for time independent convergence bounds for last-iterate convergence without additional assumptions\footnote{We emphasise that the time-average guarantee in the preceding section does not care about how the dynamics is initialised.}.

That being said, we can still \emph{``infer''} last-iterate convergence. If the initialisation $x^1$ is not a Nash equilibrium, it is easy to observe that unless there are two firms of lowest marginal cost $c/n$ who post this price, the probability that any firm will bid any quantity lower than $c/n$ will tend to $0$. So, following \cite{deng2022nash}, we will consider the case when the minimum marginal price equals $= 0$, the action sets are given $\{1/n,2/n,...,1\}$, and in all Nash equilibria of the game, firms $i \leq m$ post prices $1/n$ with probability $1$.

In this case, we remark that the dual solution of the preceding section still establishes that the only semicoarse equilibria of this fragment of the game is its Nash equilibria -- with the caveat that $\delta_i$ should be fixed zero when $c_i = 1$. However, the Nash equilibria we consider are now ``stable'', in the sense that the set of mixed and pure Nash equilibria of the game coincide exactly with the set of maximisers of $h(x) = \sum_i h_i(x_i)$. The function $h$ is concave and weakly negative, and increases along the trajectories of the game's continuous dynamics. Moreover, in the fragment of the game we consider, this decrease rate is bounded below by a function linear in distance from equilibrium.

We will demonstrate the argument for explicit convergence bounds via the first-price auction. Suppose the step-sizes are all $1/\sqrt{T}$, and that $n$ is large enough for our simpler bound. We know that at time $T$, the time averaged play concentrates on the set of game's Nash equilibria with probability $\geq 1 - \frac{3}{\sqrt{T}}mNn^5$. This implies that in the first $T = \frac{9m^2N^2n^{10}}{\epsilon^2}$ rounds, there must have been at least one time period where the mixed-strategy profile $x^t$ places probability $1-\epsilon$ on all firms $i \leq m$ posting prices $1/n$. In this case, if firm $i$ posts price $1/n$, they gain utility at least $(1-\epsilon)/N$, accounting for potential tie bids from firms with marginal cost $1/n$. On the other hand, if they post any other price, their utility can be at most $\epsilon$ (at price $1$). Therefore, if $(1-\epsilon)/N > \epsilon + 1/2N$, then bidding $1/n$ is strictly utility maximising, and at each time step posting price $1/n$ has utility at least $1/2N$ more than any other price.

We thus want $\epsilon < 1/2(N+1)$, after which last-iterate (and in fact, finite-time) convergence is assured by the strictness (c.f. \cite{MZ19}) of the equilibrium. In the first $T = 36m^2(N+1)^2N^2n^{10}+1$ rounds of gradient ascent, there is then necessarily a point in time where the probability of Nash equilibrium prices exceeds $1-\epsilon$. At this point, the utility gradients of these firms, post projection on the tangent cone to the simplex at $x^T$, have their $1/n$ component at least $1/4N$, and each of these firms post price $1/n$ with probability $\geq \epsilon$. Therefore, in $K$ more rounds where $K$ satisfies roughly $\epsilon \leq \int_T^{T+K} dt \cdot \frac{1}{\sqrt{t}} \cdot \frac{1}{4N}$, we ensure finite-time convergence. For a simple, very lax bound, we may simply set $K = \lceil(4N\epsilon)^2\rceil$.

Adapting the arguments above gives us a slightly more general bound:
\begin{proposition}\label{prop:fin-itr-conv-rate}
    In our Bertrand competition setting as in Proposition \ref{prop:time-avg-conv-rate}, let $\epsilon < D(1/n) / 2(N+D(1/n))$,  let $T$ be the number of time steps required to guarantee an $\epsilon$ bound on average non-equilibrium play, and let $K$ satisfy $\sum_{t = T+1}^{T+K} \eta_{it} \geq 4N\epsilon /D(1/n)$ for any firm $i \leq m$. Then Nash equilibrium convergence is assured within the first $T+K$ rounds.
\end{proposition}

The results of Proposition \ref{prop:time-avg-conv-rate} and Proposition \ref{prop:fin-itr-conv-rate} are reflection of a more general convergence guarantee, following the duality between primal problems providing guarantees for local coarse equilibrium and dual problems corresponding to Lyapunov function fitting, as remarked in \cite{ahunbay2024local} (Section 5). First, if the semicoarse equilibrium of a normal form game is necessarily the convex hull of its Nash equilibria, time-average convergence is immediate.

\begin{theorem}\label{thm:general-lyapunov}
    Let $O \subseteq \times_{i \in N} A_i$ be the set of pure Nash equilibria of a normal form game. Then all semicoarse equilibria of the game assign probability $1$ to the set of outcomes in $O$ if and only if there exists a tangential quadratic function $h(x) = \sum_{i \in N} \frac{1}{2}x_i^T Q_i x_i + q_i^T x_i \equiv \sum_{i \in N} h_i(x_i)$ such that for any $a \in \times_{i \in N} A_i$,
    $$ \sum_{i \in N} \bilin{\nabla_ih(e_a)}{\nabla_iu_i(e_a)} \geq \mathbb{I}[a \textnormal{ is a Nash equilibrium}],$$
    with equality on $O$. In this case, when all players $i$ employ projected gradient ascent with suitably declining step sizes, the time averaged probability of playing a Nash equilibrium is 
    $$ \geq 1- \frac{1}{T} \cdot \sum_{i \in N} \left( \frac{1}{\eta_{iT}} +\frac{1}{\eta_{i1}} + \sum_{t=1}^T \eta_{it} \right) \cdot \left( 2M_i + \frac{3G_i^2L_{h_i}}{2} \right),$$
    where $M_i$, $G_i$, $L_{h_i}$ are respectively bounds on the absolute value of $h_i$, the magnitude of the utility gradient $\|\nabla_iu_i\|$, and the Lipschitz modulus of $h_i$.
\end{theorem}

\begin{proof}
    The first part of the statement is from strong duality between the LP $\max_{\sigma} \sum_{a \notin O} \sigma(a)$ over the semicoarse equilibria of the game, and that of the Lyapunov function estimation problem (15). The second part is via sensitivity analysis on this LP when the time average play forms an $\epsilon$-local coarse correlated equilibrium with respect to the set of tangential quadratic functions on $\times_{i \in N} \Delta(A_i)$, invoking the convergence bound from Theorem \ref{thm:P-scaled-sets}.
\end{proof}

Second, when the convex hull of the vertices of $\times_{i \in N} \Delta(A_i)$ is precisely a complete face of $F$ of $\Delta(A_i)$, and if this face $F$ equals the set of maximisers of $h$, then this face $F$ is strictly attracting. This allows inferring last (in fact, finite) iterate convergence by simply inspecting the form of $h$.

\begin{theorem}
    Let $O_i \subseteq A_i$, and suppose that the set of Nash equilibria of the game coincides with the set of mixed-strategy profiles $x$ which assign probability $1$ to the outcomes in $\times_{i \in N} O_i$. Further suppose that the game admits a tangential quadratic Lyapunov function $h(x) = \sum_{i\in N} \frac{1}{2}x_i^T Q_i x_i + q_i^T x_i$ as in Theorem \ref{thm:general-lyapunov}, such that $\arg\max_{x \in \Delta_i(A_i)} h(x)$ coincides with the set of mixed Nash equilibria of the game. Then there exists a constant $\beta$ such that for any player $i$, for any pure strategy Nash equilibrium $a$, and for any action $a'_i \in A_i\setminus O_i$, $u_i(a) - u_i(a'_i,a_{-i}) \geq \beta$. As a consequence, if all players implement projected gradient ascent with suitably declining step sizes, they reach an equilibrium outcome in finite iterations.
\end{theorem}

\begin{proof}
    Consider any pure strategy Nash equilibrium $a$, and let $i$, $a'_i \in A_i\setminus O_i$ be arbitrary. Then at the characteristic vector $x'$ for the pure strategies $(a'_i,a_{-i})$, we have 
    $$ \sum_{j \in N} \bilin{\nabla_j h_j(x'_j)}{\nabla_j u_j(x')} \geq 1.$$
    On the other hand, by the assumption on the set of maxima of $h$ and its additive separability $h(x) = \sum_{j \in N} h_j(x_j)$, each $h_j$ must be a function whose set of maxima are mixed strategies which assign probability $1$ to actions in $O_j$. As a consequence, $\nabla_j h_j(x'_j) = 0$ for any $j \neq 1$, and we have
    $$ \bilin{\nabla_i h_i(x'_i)}{\nabla_i u_i(x')} \geq 1.$$
    Suppose we have picked the action $a'_i$ with the highest non-best response utility. Then RHS is bounded above by any deviation (scaled by the $a'_i$ component of $\nabla_i h_i(x'_i)$) which moves to a best response. Since all players $j \neq i$ play an action in $O_j$, by the assumption on the set of Nash equilibria of the game, all actions in $O_i$ are best responses. Thus,
    $$ \|\nabla_ih_i(x'_i) \| (u_i(a) - u_i(a'_i,a_{-i})) \geq |\nabla_i h_i(x'_i)_{a'_i}|(u_i(a) - u_i(a'_i,a_{-i})) \geq 1,$$
    and we may pick $\beta = 1/G_h$.

    Now, replicating our previous arguments for the Bertrand competition, let $\epsilon < \beta/2(\bar{U}-\beta)$, where $\bar{U} = \max_{i \in N, a \in A} u_i(a)$. By Theorem \ref{thm:general-lyapunov}, there is a time $T$ such that for a period $t$ within the first $T$ rounds of the game, the mixed strategies $x^t$ place probability $1-\epsilon$ on a pure strategy Nash equilibrium. By choice of $\epsilon$, $u_i(a_i,x^t_{-i}) > u_i(a'_i,x^t_{-i}) + \beta/2$ for any player $i$, and any pair of actions $a_i \in O_i$ and $a'_i \in A_i \setminus O_i$. As a consequence, the set $O_i$ is attracting. By considering projection of the utility gradient, we infer that the probability assigned to $O_i$ increases by an amount $\geq \eta_{it'} \cdot \beta/4$ at each time period until it reaches $1$. Fixing $K$ such that $\sum_{t = T+1}^T \eta_{it} \geq 4\epsilon/\beta$, we conclude that all players will assign probability $1$ on actions in $O_i$ after $T+K$ rounds of projected gradient ascent, reaching a Nash equilibrium. 
\end{proof}

\end{document}